\documentclass[reqno]{amsart}
\usepackage{amssymb}
\usepackage{enumerate}
\usepackage{hyperref}
\usepackage{doi,url,color}

\usepackage{mathtools}

\usepackage{graphicx}

\numberwithin{equation}{section}
\newtheorem{theorem}{Theorem}[section]
\newtheorem{proposition}[theorem]{Proposition}
\newtheorem{lemma}[theorem]{Lemma}
\newtheorem{corollary}[theorem]{Corollary}
\newtheorem{definition}[theorem]{Definition}
\newtheorem{remark}[theorem]{Remark}

\newtheorem*{hall}{Hall's Marriage Theorem}

\DeclareMathOperator{\supp}{supp}
\DeclareMathOperator{\tr}{tr}

\DeclareMathOperator{\Rank}{Rank}

\DeclareMathOperator{\dist}{dist}

\DeclareMathOperator{\diam}{diam}

\makeatletter
\DeclareRobustCommand\widecheck[1]{{\mathpalette\@widecheck{#1}}}
\def\@widecheck#1#2{%
    \setbox\z@\hbox{\m@th$#1#2$}%
    \setbox\tw@\hbox{\m@th$#1%
       \widehat{%
          \vrule\@width\z@\@height\ht\z@
          \vrule\@height\z@\@width\wd\z@}$}%
    \dp\tw@-\ht\z@
    \@tempdima\ht\z@ \advance\@tempdima2\ht\tw@ \divide\@tempdima\thr@@
    \setbox\tw@\hbox{%
       \raise\@tempdima\hbox{\scalebox{1}[-1]{\lower\@tempdima\box
\tw@}}}%
    {\ooalign{\box\tw@ \cr \box\z@}}}
\makeatother

\renewcommand\H{\mathcal{H}}

\renewcommand\L{\mathrm{L}}
\newcommand\R{\mathbb R}
\newcommand\N{\mathbb N}
\newcommand\C{\mathbb C}
\newcommand\Z{\mathbb Z}
\newcommand\G{\mathbb{G}}
\newcommand\M{\mathbb M}

\newcommand\e{\mathrm{e}}
\newcommand{\la}{\langle}
\newcommand{\ra}{\rangle}
\renewcommand\P{\mathbb P}
\newcommand\E{\mathbb E}
\newcommand\cE{\mathcal{E}}

\newcommand\cL{\mathcal{L}}
\newcommand\cB{\mathcal{B}}
\newcommand\cN{\mathcal{N}}

\newcommand\cG{\mathcal{G}}
\newcommand\cV{\mathcal{V}}

\newcommand\cA{\mathcal{A}}

\newcommand\cF{\mathcal{F}}
\newcommand\cS{\mathcal{S}}
\newcommand\cY{\mathcal{Y}}
\newcommand\cR{\mathcal{R}}

\newcommand\eps{\varepsilon}
\newcommand\vphi{\varphi}
\newcommand\vtheta{\vartheta}
\newcommand\vrho{\varrho}

\newcommand{\pr}{\prime}

\newcommand\what{\widehat}
\newcommand\wtilde{\widetilde}

\newcommand\ttau{\widetilde{\tau}}
\newcommand\tzeta{\widetilde{\zeta}}
\newcommand{\partialin}{\partial_{\mathrm{in}}}
\newcommand{\partialex}{\partial_{\mathrm{ex}}}

\newcommand{\bom}{{\boldsymbol{{\omega}}}}

\newcommand\beq{\begin{equation}}
\newcommand\eeq{\end{equation}}

\newcommand{\abs}[1]{\left\lvert #1 \right\rvert}
\newcommand{\norm}[1]{\left\lVert #1 \right\rVert}
\newcommand{\scal}[1]{\left\langle #1 \right\rangle}
\newcommand{\set}[1]{\left\{ #1 \right\}}
\newcommand{\pa}[1]{\left( #1 \right)}

\newcommand{\fl}[1]{\left\lfloor #1 \right\rfloor}
\newcommand{\br}[1]{\left [ #1 \right]}

\newcommand\La{\Lambda}

\newcommand\Th{\Theta}
\newcommand\Ups{\Upsilon}

\newcommand{\eq}[1]{\eqref{#1}}

\newcommand{\up}[1]{^{(#1)}}

\newcommand{\qtx}[1]{\quad\text{#1}\quad}
\newcommand{\mqtx}[1]{\; \ \text{#1}\; \  }
\newcommand{\sqtx}[1]{\;\text{#1}\;}

\newcommand\Chi{\raisebox{.2ex}{$\chi$}}

\begin{document}

\title[An eigensystem approach to Anderson localization]{An eigensystem approach to Anderson localization}

\author{Alexander Elgart}
\address[A. Elgart]{Department of Mathematics; Virginia Tech; Blacksburg, VA, 24061, USA}
 \email{aelgart@vt.edu}

\author{Abel Klein}
\address[A. Klein]{University of California, Irvine;
Department of Mathematics;
Irvine, CA 92697-3875,  USA}
 \email{aklein@uci.edu}

\thanks{A.E. was  supported in part by the NSF under grant DMS-1210982.}
\thanks{A.K. was  supported in part by the NSF under grant DMS-1001509.}

%\date{Version of \today}

\begin{abstract} We introduce a new approach for proving  localization (pure point spectrum with exponentially decaying eigenfunctions, dynamical localization) for the Anderson model at high disorder. In contrast to the usual strategy, we do not study  finite volume  Green's functions.   Instead, we perform a multiscale analysis based on finite volume eigensystems (eigenvalues and eigenfunctions). Information about eigensystems at a given scale is used to derive information about eigensystems at larger scales.  This eigensystem multiscale analysis
treats all energies of the finite volume operator at the same time, establishing level spacing and  localization of  eigenfunctions in a fixed box with high probability.
  A new feature is the  labeling of the eigenvalues and eigenfunctions  by the sites of the box. 
\end{abstract}

\maketitle

\tableofcontents

\section*{Introduction}

The Anderson model \cite{And} is the prototype for the study of localization properties of quantum states of
single electrons in disordered solids.  It is given by a  random Schr\"odinger operator
$H_{\eps,\bom}= -\eps \Delta + V_{\bom}$ acting on $\ell^2(\Z^d)$, where $\Delta$ is the discrete Laplacian, $V_{\bom}$ is a random potential, and   $\eps >0$ is the reciprocal of the disorder parameter  (see Definition~\ref{defineAnd} for the details). The basic phenomenon, known as the Anderson localization, is that  high disorder ($\epsilon\ll1$) leads to  localization of electron states. Its most basic  manifestation is that  $H_{\eps,\bom}$ has pure point spectrum with   exponentially decaying eigenfunctions  with probability one:   for almost every configuration of the random potential,  $H_{\eps,\bom}$ has a complete orthonormal basis of eigenvalues  $\set{\psi_{\eps,\bom,j}}_{j\in \N}$ such that
\[
\abs{\psi_{\eps,\bom,j}(x)}\le C_{\eps,\bom,j} \e^{-m_\eps\norm{x}}\qtx{for all} x \in \Z^d \qtx{and} j\in \N,
\]
where $m_\eps>0$, the reciprocal of the  localization length, is nonrandom and independent of $j\in \N$.   Other manifestations include dynamical localization and SULE (semi-uniformly localized eigenfunctions).  (See, for example,  \cite{AW,Ki,Kle}.)

These manifestations of localization suggest that 
 truncation of the system to a finite  box $\Lambda_L$ of side  $L\gg \frac 1 {m_\eps}$ should not affect localization properties deep inside the box. This leads to the expectation  that if  one could establish an appropriate  analogue of localization for a sequence of  boxes $\Lambda_{L_n}$, with $L_n\rightarrow\infty$, then localization  should hold in the whole of $\Z^d$ as well. This strategy can be indeed be implemented and is known as the multiscale analysis. In a nutshell, the multiscale analysis uses as input localizing properties at  scale $L_n$  to establish localizing properties at  scale $L_{n+1}$. The question is what kind of information we want to carry from scale to scale. In the traditional approach to Anderson localization, such information is encoded in the decay properties of the underlying Green's function. For single-particle systems, the Green's function $G_{\eps,\bom}(x,y; \lambda)=\scal{\delta_x,(H_{\eps,\bom}-\lambda)^{-1}\delta_y}$ is an extremely convenient object to study. Its usefulness comes from  two key properties: (a)  Green's functions for boxes at different scales are related by  the first resolvent identity;   (b) knowledge of the decay properties of the Green's functions for all energies (or for all energies in a fixed interval)  can be translated into localization properties of the eigenfunctions (in the fixed interval).
 
 The well known methods developed for proving localization for random Schr\"odinger operators, the multiscale analysis \cite{FS,FMSS,Dr,DK,Sp,CH,FK,GKboot,Kle,BK,GKber} and the fractional moment method \cite{AM,A,ASFH,AENSS,AW},  are based on the study of   finite volume  Green's functions.  Multiscale analyses based on  Green's functions  are performed  either at a fixed energy in a single box, or for all energies but with two boxes with an `either or' statement for each energy. 
 
Recently there has been an intensive effort in the physics community to create a coherent theory of many-body localization (MBL); see, e.g.,  \cite{FlA, AlGKL, GoMP, BAA,BurO,OH,PH,NH, FrWBSE,EFG}.  On the mathematical level, not much progress have been made, besides studies of exactly solvable models; see, e.g., \cite{HSS,PS,AS}.  One of the key  difficulties in studying  MBL is associated with the fact that  Green's functions do not appear to be such a   valuable  tool as in the single-particle theory, due to the product state nature of the underlying Hilbert space.  The objects that do appear in the most physical descriptions of MBL are the eigenstates of the system. This suggests that finding a more direct, eigensystem based approach to  localization, even in the single-particle case, could be  beneficial for understanding MBL. Such approach has been  advocated by Imbrie in a context of both  single and many-body localization \cite{Im1,Im2}. 

In this paper we provide a mathematically rigorous implementation of  a multiscale analysis for the Anderson model at high disorder based on  finite volume eigensystems (eigenvalues and eigenfunctions). In contrast to the usual strategy, we do not study  finite volume  Green's functions. 
 Information about eigensystems at a given scale is used to derive information about eigensystems at larger scales. This eigensystem multiscale analysis
treats all energies of the finite volume operator at the same time, giving a complete picture in a fixed box.  For this reason it does not use a Wegner estimate as in a Green's functions multiscale analysis, it uses instead   a probability estimate for level spacing  derived by Klein and Molchanov  from Minami's estimate \cite[Lemma~2]{KlM}. 

A new feature provided by the eigensystem multiscale analysis is the  labeling of the eigenvalues and eigenfunctions  by the sites of the box.  We establish this labeling by the multiscale analysis using an argument based on Hall's Marriage Theorem  (e.g., \cite[Chapter~2]{BDM}).

Our main result, stated in Theorem \ref{thmMSA} can be loosely described as follows:  If  $\epsilon\ll1$,  with high probability the eigenvalues and eigenfunctions of $H_{\eps,\bom,\La_L}$, the restriction of $H_{\eps,\bom}$ to a finite  box $\Lambda_L$ of side  $L\gg 1$,  can be labeled  by the sites of $\Lambda_L$, i.e., they can be written in the form $\set{(\vphi_x,\lambda_x)}_{x\in \La_L}$, with the eigenvalues $\set{\lambda_x}_{x\in \La_L}$ satisfying a level spacing condition, 
and the eigenfunctions $\set{\vphi_x}_{x\in \La_L}$ exhibiting localization around the label, i.e., for all $x \in \La_L$ we have 
\[
\abs{\vphi_x(y)}\le \e^{-m_\eps\norm{y-x}}\qtx{for all} y \in \La_L \qtx{with} \norm{y-x}\ge L^\tau,
\]
where $m_\eps >0$ is nonrandom  and $0<\tau<1$ is a fixed parameter.

Theorem \ref{thmMSA} yields  Anderson localization  (pure point spectrum with exponentially decaying eigenfunctions, dynamical localization) for $H_{\eps,\bom}$. It is our  hope that the eigensystem multiscale analysis is a step towards developing new methods that may be useful in the study of MBL.

We also investigate the connection between  the eigensystem multiscale analysis and the Green's functions multiscale analysis.
We  show that the  conclusions of the Green's functions multiscale analysis can be derived from  the conclusions of the  eigensystem multiscale analysis.  Conversely, we show that the conclusions of the  eigensystem multiscale analysis can be derived from the Green's functions energy interval  multiscale analysis  with the addition of  the labeling argument based on Hall's Marriage Theorem we present in this paper.

The results in this paper concern localization for the Anderson model in the whole spectrum, which in practice requires high disorder. For the Anderson model, the Green's function methods for proving localization can be applied in energy intervals, and hence
 localization has also been proved  at fixed disorder  in an interval at the edge of the spectrum (or, more generally, in the vicinity of a spectral gap), and  for  a fixed interval of energies at the bottom of the spectrum  for sufficiently  high disorder.  (See, for example, \cite{HM,KSS,FK1,ASFH,GKfinvol,Ki,GKber,AW}).
In a forthcoming paper \cite{EK}, we generalize the version of the eigensystem multiscale analysis  presented  in this paper to  establish  localization   for the Anderson model in an energy interval. This extension yields  localization at fixed disorder  on an interval at the edge of the spectrum (or in  the vicinity of a spectral gap), and  at  a fixed interval at the bottom of the spectrum  for sufficiently  high disorder.  

Klein and Tsang \cite{KlT} have used a bootstrap argument as in \cite{GKboot,Kle} to enhance the eigensystem multiscale analysis for the Anderson model at high disorder developed in this paper.   The  only input required  to initiate the eigensystem bootstrap multiscale analysis  is polynomial decay of the finite volume eigenfunctions for  sufficiently large scale with some minimal, scale-independent probability. It yields a result analogous to \eq{concMSA} in Theorem~\ref{thmMSA}, for all $0<\xi<1$, with $\eps_0$ independent of $\xi$.

Our main results and definitions are stated in Section~\ref{secmain}.  Theorem~\ref{thmMSA} is our main result, the conclusions of the eigensystem multiscale analysis, which  we prove  in   Section~\ref{secEMSA}.  Theorem~\ref{thmloc}, derived from  Theorem~\ref{thmMSA}, encapsulates localization for the Anderson model. Corollary~\ref{corloc} contains typical statements  of Anderson localization and dynamical localization. Theorem~\ref{thmloc} and Corollary~\ref{corloc} are proven in Section~\ref{seclocproof}. In Section~\ref{secprobest} we  adapt an  estimate for the probability of level spacing  derived by Klein and Molchanov (reviewed in Appendix~\ref{appspspaced}) to our setting.
Section~\ref{secprep} contains definitions and lemmas required for the proof of the eigensystem multiscale analysis given in  Section~\ref{secEMSA}. The connection with  the Green's functions multiscale analysis is established in Section~\ref{secGreen}.  Hall's Marriage Theorem, used in  Section~\ref{secEMSA} for labeling eigenvalues and eigenfunctions, is reviewed in Appendix~\ref{appHall}.

\section{Main results}\label{secmain}

We start by introducing the Anderson model in a convenient form. 

\begin{definition}\label{defineAnd} The Anderson model is the
 random Schr\"odinger  operator
\beq \label{defAnd}
H_{\eps,\bom} :=  -\eps \Delta + V_{\bom} \quad \text{on} \quad  \ell^2(\Z^d), 
\eeq 
where
\begin{enumerate}
\item  $\Delta$ is the  (centered) discrete  Laplacian:  
 \begin{equation}\label{defDelta}
  (\Delta \varphi)(x):=  \sum_{\substack{y\in\Z^d\\ |y-x|=1}} \varphi(y)  \qtx{for} \varphi\in\ell^2(\Z^d).
\end{equation}

\item   $V_{\bom}$ is a random potential:      $V_{\bom}(x)= \omega_x$ for  $ x \in \Z^d$, where
$\bom=\{ \omega_x \}_{x\in
\Z^d}$ is a family of independent 
identically distributed random
variables,  whose  common probability 
distribution $\mu$ is non-degenerate with bounded support.  We assume $\mu$ is H\"older continuous of order $\alpha \in ( \frac 12,1]$: 
 \beq\label{Holdercont}
S_\mu(t) \le K t^\alpha \qtx{for all} t \in [0,1],
\eeq
where $K$ is  a constant and  $S_\mu(t):= \sup_{a\in \R} \mu \set{[a, a+t]} $ is the concentration function of the measure $\mu$.

\item  $\eps >0$ is the reciprocal of the disorder parameter (i.e., $\frac
 1 \eps$ is the disorder parameter).

\end{enumerate}
\end{definition}

We recall that  $\sigma(-\Delta)=[-2d,2d]$ and  (see \cite[Theorem~3.9]{Ki})
\beq\label{Sigma}
 \sigma (H_{\eps,\bom})= \Sigma_\eps:= [-2\eps d,2\eps d] + \supp \mu \quad \text{with probability one}.
\eeq

By a discrete  Schr\"odinger operator we will always mean an operator 
 $H=-\eps\Delta +V$ on $\ell^2(\Z^d)$, where $V$ is a  bounded potential and $\eps \ge 0$.

We use the following definitions and notation:
\begin{itemize}
\item 
 If $x=(x_1,x_2,\ldots, x_d)\in \R^d$, we set $\abs{x}=\abs{x}_2= \pa{\sum_{j=1}^dx_j^2}^{\frac 12}$,  $\abs{x}_1=\sum_{j=1}^d \abs{x_j}$, and $\norm{x}=\abs{x}_\infty= \max_{j=1,2,\ldots,d} \abs{x_j}$. If $x\in \R^d$ and   
$\Xi\subset \R^d$, we set $\dist (x,\Xi)= \inf_{y\in \Xi} \norm{y-x}$.  The diameter of a set $\Xi\subset \R^d$ is given by $\diam \Xi= \sup_{x,y \in \Xi} \norm{y-x}$.

 \item 
We consider $\Z^d$ as a subset of $\R^d$ and  use  boxes in $\Z^d$ centered at points in $\R^{d}$.    The box in $\Z^d$  of side $L>0$  centered at $x\in \R^{d}$ is given by
\beq 
\La_L(x)=\La^\R_L(x)\cap \Z^d=  \set{y \in \Z^d;\  \norm{y-x} \le  \tfrac{L}{2}}  \subset \Z^d,
\eeq
where $\La^\R_L(x)$ is the box in $\R^d$  of side $L>0$  centered at $x\in \R^{d}$, given by
 \beq 
\La^\R_L(x)= \set{y \in \R^d;\  \norm{y-x} \le  \tfrac{L}{2}} \subset \R^d.
\eeq
By a box  $\La_L$ we will mean a box $\La_L(x)$ for some $x\in \R^d$.  Note that for all scales  $L\ge 2$ and  $x\in \R^d$ we have
 \beq 
 (L-2)^{d}<  \pa{ 2\fl{\tfrac L 2}}^d \le\abs  {\La_L(x)}\le \pa{ 2\fl{\tfrac L 2}+1}^d\le (L+1)^{d}. \eeq

  \item 
Given   $\Phi\subset \Theta\subset \Z^d$, we consider $\ell^2(\Phi)\subset \ell^2(\Theta)$ by extending functions on $\Phi$ to functions on $\Theta$ that are identically $0$ on $\Theta\setminus \Phi$. 
If  $\vphi$ is a function on   $\Theta$,  we write $\vphi_{\Phi}=
 \Chi_{\Phi}\vphi $.  If $\Th \subset \Z^d$ and $\vphi \in  \ell^2(\Theta)$,  we let $\norm{\vphi}=\norm{\vphi}_2$ and   $\norm{\vphi}_\infty= \max_{y \in \Theta} \abs{\vphi(y)}$.
  
\item    We will work with finite volume operators.
  If  $K$ is  a bounded operator on  $\ell^2(\Z^d)$ and  $\Theta\subset \Z^d$, we let $K_\Theta$ be the restriction of $ \Chi_\Theta K\Chi_\Theta$ to $\ell^2(\Theta)$.
  
  \item
By a  constant  we  always mean a finite constant.  We will use  $C_{a,b, \ldots}$, $C^{\pr}_{a,b, \ldots}$,  $C(a,b, \ldots)$, etc., to  denote a constant depending  on the parameters
$a,b, \ldots$. Note that $C_{a,b, \ldots}$ may denote different constants in different equations, and even in the same equation.

\item  If $\cE$ is an event, we  denote its complementary event by $\cE^c$.
  
  \end{itemize}

 We fix   $\xi,\zeta, \beta,\tau \in (0,1)$ and $\gamma >1$ such that
\begin{gather}\label{ttauzeta0}
0<\xi< \zeta<\beta<\frac 1 \gamma <1<\gamma < \sqrt {\tfrac \zeta \xi} \mqtx{and}   \max\set{ \gamma \beta, \tfrac {(\gamma-1)\beta +1}{\gamma}  }      <  \tau <1,
\end{gather}
and note that  
\beq\label{ttauzeta}
0<\xi<\xi\gamma^2< \zeta < \beta <\frac \tau \gamma   <\frac 1 \gamma <\tau <1< \frac {1-\beta}{\tau-\beta} < \gamma <\frac \tau \beta.
\eeq
We also take
\beq\label{ttauzeta2}
\tzeta= \frac {\zeta +\beta}2  \in (\zeta, \beta)  \qtx{and} {\ttau}= \frac {1 +\tau}2  \in (\tau,1).
\eeq
Given a scale $L\ge 1$,  we set
\[
 \ell= L^{\frac 1 \gamma} \; (\text{i.e.,}\; L=\ell^\gamma), \quad  L_{\tau}=\fl{L^{\tau}}, \qtx{and} L_{{\ttau}}= \lfloor{L^{{\ttau}}}\rfloor.
\]

  The following definitions are for a fixed discrete  Schr\"odinger operator
 $H_\eps$.
  We omit $\eps$ from the notation (i.e., we write $H$ for $H_\eps$, $H_\Th$ for $H_{\eps,\Th}$) when it does not lead to confusion.
 
 \begin{definition}
Given $\Theta \subset \Z^d$,  we call $(\vphi,\lambda)$ an eigenpair for $H_\Th$ if $\vphi\in \ell^2(\Theta)$ with $\norm{\vphi}=1$,  
 $\lambda \in \R$, and $H_\Theta \vphi=\lambda \vphi$.  (In other words,    $\lambda$ is an eigenvalue for $H_\Th$ and $\vphi$ is a corresponding normalized eigenfunction.)  A collection $\set{(\vphi_j,\lambda_j)}_{j\in J}$ of eigenpairs for $H_\Th$ will be called an eigensystem for $H_\Th$ if   $\set{\vphi_j}_{j\in J}$ is an orthonormal basis for $\ell^2(\Th)$. If   all eigenvalues of  $H_\Th$ are simple, we can rewrite
 the eigensystem as  $\set{(\psi_\lambda,\lambda)}_{\lambda\in \sigma(H_\Th)}$.
  \end{definition}
 
 \begin{definition} Let $\La_L$ be a box, $x\in \La_L$, and  $m>0$.  Then   $\vphi\in \ell^2(\La_L)$  is said to be $(x,m)$-localized if  $\norm{\vphi}=1$  and
\beq\label{hypdec}
\abs{\vphi(y)}\le \e^{-m\norm{y-x}}\qtx{for all} y \in \La_L \qtx{with} \norm{y-x}\ge L_\tau.
\eeq
In particular,
\beq\label{hypdec2}
\abs{\vphi(y)}\le \e^{mL_\tau}\e^{-m\norm{y-x}}\qtx{for all} y \in \La_L.\eeq
\end{definition}

\begin{definition}  Let $R>0$.   A finite set $\Theta \subset \Z^d$ will be called  
$R$-level spacing  for $H$ if  $\abs{ \sigma(H_{\Theta})}=\abs{\Theta}$ (i.e., all eigenvalues of $H_\Th$ are simple) and  $\abs{\lambda- \lambda^\pr}\ge \e^{-{R}^\beta}$ for all $\lambda, \lambda^\pr\in \sigma(H_{\Theta})$, $\lambda\ne  \lambda^\pr$.

In the special case when $\Theta$ is a box $\La_L$  and $R=L$,  we will simply say that
 $\La_L$ is level spacing  for $H$.
\end{definition}

\begin{definition}    Let $m>0$.   A box $\La_L$ will be called $m$-localizing for $H$ if the following holds:

\begin{enumerate}
\item  $\La_L$ is level spacing  for $H$.
\item  There exists an  $m$-localized eigensystem for $H_{\La_L}$, that is,    an  eigensystem   $\set{(\vphi_x, \lambda_x)}_{x\in \La_L}$ for $H_{\La_L}$ such that $\vphi_x$ is  $(x,m)$-localized for all $x \in  \La_L$.

\end{enumerate}

 \end{definition}

 The eigensystem multiscale analysis yields the following theorem.

  \begin{theorem} \label{thmMSA}  Let $H_{\eps,\bom}$ be an Anderson model. There exists a finite scale   $L_0$ such that  for  all    $0< \eps\le \eps_0= \frac 1 {4d} \e^{-L_0^{ \beta}}$ we have 
  \begin{align}\label{concMSA}
\inf_{x\in \R^d} \P\set{\La_{L} (x) \sqtx{is}  \tfrac {m_{\eps,L_0} }4 \text{-localizing for} \; H_{\eps,\bom}} \ge 1 -  \e^{-L^\xi} \mqtx{for all} L\ge L_0^\gamma,
\end{align}
 where
 \beq\label{mepsL}
m_{\eps,L}= \log \pa{1 + \tfrac {\e^{-L^\beta}}{2d\eps} }\qtx{for} \eps>0 \qtx{and}L\ge 1.
\eeq
  \end{theorem}
 Note that for  $ 0<\eps\le \eps_0= \frac 1 {4d} \e^{-L_0^{ \beta}}$ we have  
 \beq\label{mepsL33}
m_{\eps,L_0}= \log \pa{1 + \tfrac {2\eps_0}{\eps} }\ge \log 3.
\eeq

Theorem~\ref{thmMSA} is proved in Section~\ref{secEMSA}. It  yields all the usual forms of  localization.  To see this we need to introduce  some notation and definitions.
We fix $\nu > \frac d 2$, which will be usually omitted from the notation. Given $a \in \Z^{d}$, we let 
 $T_{a} $  be  the operator on   $\ell^2(\Z^d)$ given by multiplication by the function
$T_{a}(x):= \la x-a\ra^{\nu}$, where  $ \la x\ra= \pa{1 + \norm{x}^2}^{\frac 12}$. Note that  $
\| T_{a} T_{b}^{-1} \| \le 2^{\frac  {\nu}  2} \la  a -b \ra^{\nu}$ since  $\langle a +b \rangle \le \sqrt{2}\langle a \rangle
\langle b\rangle$.

A function
$\psi\colon \Z^d \to \C$ will be called a $\nu$-generalized eigenfunction for  the discrete  Schr\"odinger operator $H_\eps$ if $\psi$ is a generalized eigenfunction as in Definition~\ref{defgeneig}
and $\norm{T_0^{-1} \psi}<\infty$.  (Note  $\norm{T_0^{-1} \psi}<\infty$ if and only if  $\norm{T_a^{-1} \psi}<\infty$ for all $a\in \Z^d$.)
We let 
$\cV_\eps({\lambda})$ denote the collection of  $\nu$-generalized eigenfunctions for $H_\eps$ with generalized eigenvalue ${\lambda} \in \R$.
 (We will usually drop  $\nu$ from the notation.)

Given   ${\lambda} \in \R$ and $a,b \in \Z^{d}$, we set
 \begin{align} \label{defGWx}
W_{\eps,\lambda}\up{a}({b}):=\begin{cases} 
\sup_{\psi \in\cV_\eps({\lambda}) }
\ \frac {\abs{\psi(b)}}
{\|T_{a}^{-1}\psi \|}&
 \text{if $\cV_\eps({\lambda})\not=\emptyset$}\\0 & \text{otherwise}\end{cases}.
\end{align}
Note that for $a,b,c \in \Z^d$ we have 
\begin{equation}\label{boundGW}
W\up{a}_{\eps,\lambda}({a})\le 1,\quad  W\up{a}_{\eps,\lambda}({b})\le\la b-a\ra^\nu,  \qtx{and} W\up{a}_{\eps,\lambda}({c})\le 2^{\frac \nu 2} \la b-a\ra^\nu 
W\up{b}_{\eps,\lambda}({c}).
\end{equation}
  
  The following theorem, derived from  Theorem~\ref{thmMSA}, encapsulates localization for the Anderson model.

\begin{theorem}\label{thmloc}   Suppose Theorem~\ref{thmMSA} holds  for an  Anderson model $H_{\eps,\bom}$,   let $L_0$  be the scale  given in Theorem~\ref{thmMSA}, and let  $\eps_0= \frac 1 {4d} \e^{-L_0^{ \beta}}$.   There exists a finite scale   $\cL_1 $ such that, given  $\cL_1\le \ell \in 2\N$ and  $a\in \Z^d$, then for  all $0< \eps\le \eps_0$, setting $m_\eps=  \frac {m_{\eps,L_0} }4\ge \frac {\log 3} 4$,    there exists an event  $\cY_{\eps,\ell,a}$ with the following properties:
  
  \begin{enumerate}
\item $\cY_{\eps,\ell,a}$  depends only on the random variables $\set{\omega_{x}}_{x \in \Lambda_{5\ell}(a)}$,  and  
\beq\label{cUdesiredint}
  \P\set{\cY_{\eps,\ell,a} }\ge  1 - C_{\eps_0} \e^{-\ell^\xi}.
  \eeq

\item If  $\bom \in \cY_{\eps,\ell,a}$,   for all  $\lambda \in \R$ we have that 
\beq \label{locimpl}
\max_{b\in \La_{\frac \ell 3}(a)} W\up{a}_{\eps,\bom,\lambda}(b)> \e^{-\frac 1 4 {m_\eps} \ell}\quad  \Longrightarrow \quad \max_{y\in A_\ell(a)} W\up{a}_{\eps,\bom,\lambda}(y)\le \e^{-\frac 7 {132} m_\eps \norm{y-a}},
\eeq
where
 \beq\label{Aell}
 A_\ell(a):=  \set{y\in \Z^d; \ \tfrac 8 7 \ell \le   \norm{y-a}\le \tfrac {33}{14} \ell}.
  \eeq
In particular, for all $\bom \in \cY_{\eps,\ell,a}$ and  $\lambda \in \R$ we have
\beq  \label{WW}
W\up{a}_{\eps,\bom,\lambda}(a)W\up{a}_{\eps,\bom,\lambda}(y)\le 
 \e^{- \frac 7 {132} m_\eps \norm{y-a}}\mqtx{for all} y\in A_\ell(a).
\eeq

   \end{enumerate}
  
  \end{theorem}

Theorem~\ref{thmloc}  implies Anderson localization and dynamical localization, and more, as shown in  \cite{GKsudec,GKber}.  In particular, we get the following corollary.

\begin{corollary} \label{corloc} Let $H_{\eps,\bom}$ be an Anderson model, and suppose  Theorem~\ref{thmloc} holds. 
Then for $0<\eps \le \eps_0$ the following holds with probability one:
  \begin{enumerate}
 \item  {$H_{\eps,\bom}$}  has  pure point spectrum.

\item   If   {$\psi_\lambda$}  is an {eigenfunction} of $H_{\eps,\bom}$
with eigenvalue  {$\lambda$}, then $\psi_\lambda$ is exponentially localized  with rate of decay $ \frac 7 {132} m_\eps$,   more precisely,
\begin{equation}\label{expdecay222}
\abs{\psi_\lambda(x)} \le C_{\eps,\bom,\lambda}\norm{T_0^{-1} \psi}\, e^{-  \frac 7 {132} m_\eps\norm{x}} \qquad \text{for all}\quad  x \in \R^{d}.
\end{equation}

 \item For all $\lambda \in \R$  and  $x,y \in \Z^d$ we have 
 \begin{align}\label{eqWW}
&W\up{x}_{\eps,\bom,\lambda}(x)W\up{x}_{\eps,\bom,\lambda}(y)
\le  C_{\eps,m_\eps,\bom,\nu}\e^{(\frac 4 {33} +\nu)m_\eps  (2d\log \scal{x})^{\frac 1 \xi}}  \e^{- \frac 7 {132} m_\eps \norm{y-x}}.
\end{align}

\item  For all $\lambda \in \R$ and $ \psi\in \Chi_{\set{\lambda}}(H_{\eps,\bom})$, we have
\begin{align}\label{eqWW2}
& \abs{\psi(x)}\abs{\psi(y)}
 \le C_{\eps,m_\eps,\bom,\nu} \, \norm{T_0^{-1} \psi}^2\la x\ra^{2\nu}\e^{(\frac 4 {33} +\nu)m_\eps  (2d\log \scal{x})^{\frac 1 \xi}}  \e^{- \frac 7 {132} m_\eps \norm{y-x}},
 \end{align}
for all $x,y \in \Z^d$.

\item  For all $\lambda \in \R$ there exists $x_\lambda=x_{\eps,\bom,\lambda} \in \Z^d$, such that for $ \psi\in \Chi_{\set{\lambda}}(H_{\eps,\bom})$ we have
\begin{align}\label{SULE}
\abs{\psi(x)} 
& \le  C_{\eps,m_\eps,\bom,\nu}\norm{T_{x_\lambda}^{-1} \psi} \e^{(\frac 4 {33} +\nu)m_\eps  (2d\log \scal{x_\lambda})^{\frac 1 \xi}}   \e^{- \frac 7 {132} m_\eps \norm{x-x_\lambda}}\\ \notag
  &   \le  
 2^{\frac \nu 2} C_{\eps,m_\eps,\bom,\nu}\norm{T_0^{-1} \psi}\la x_\lambda\ra^{\nu} \e^{(\frac 4 {33} +\nu)m_\eps  (2d\log \scal{x_\lambda})^{\frac 1 \xi}}   \e^{- \frac 7 {132} m_\eps \norm{x-x_\lambda}},
 \end{align}
for all $x \in \Z^d$.
\end{enumerate}
\end{corollary}

 In Corollary~\ref{corloc}, (i) and (ii)  are statements of Anderson localization, (iii) and (iv) are statements of dynamical localization  ((iv) is called   SUDEC (summable uniform decay of eigenfunction correlations) in \cite{GKsudec}), and (v) is SULE (semi-uniformly localized eigenfunctions; see
\cite{DRJLS0,DRJLS}).

 We can also derive statements of localization in expectation, as in \cite{GKsudec,GKber}.

Theorem~\ref{thmloc} and Corollary~\ref{corloc} are proven in Section~\ref{seclocproof}.

   \section{Probability estimate for   level spacing}\label{secprobest}

We  adapt a probabilistic estimate of Klein and Molchanov  \cite[Lemma~2]{KlM} to our setting.  (This estimate is reviewed in Appendix~\ref{appspspaced}.)

If $J\subset \R$, we set $\diam J =\sup_{s,t \in J} \abs{s-t}$.

\begin{lemma}\label{lemSep}  Let  $H_{\eps,\bom}$ be an Anderson model as in Definition~\ref{defAnd}.
Let $\Th \subset \Z^d$ and $L>1$. Then, for all $0< \eps \le \eps_0$,
\begin{align}
\P\set{\Th \sqtx{is} L\text{-level spacing for}\;H_{\eps,\bom} }\ge 1 -Y_{\eps_0} \e^{-(2\alpha-1)L^\beta}\abs{\Th}^2.
\end{align}
where
\beq
 Y_{\eps_0} = 2^{2\alpha-1} \wtilde{K}^2 \pa{ \diam \supp \mu + 2d \eps_0 +1},
\eeq
with $\wtilde{K}=K$ if $\alpha=1$ and $\wtilde{K}=8K$ if  $\alpha \in ( \frac 12,1)$.

In the special case of a box $\La_L$, we have
\begin{align}\label{probsep}
\P\set{\La_L \sqtx{is level spacing for} H }\ge 1 -Y_{\eps_0}\pa{L+1}^{2d} \e^{-(2\alpha-1)L^\beta}.
\end{align}
\end{lemma}

\begin{proof}
Recalling  \eq{Sigma},  we have   
\beq\Sigma_\eps \subset I_\eps := [\inf \Sigma_\eps,\sup \Sigma_\eps] \qtx{and}\abs{I_\eps}= \diam \supp \mu + 2d \eps.
\eeq
Thus, it follows from Lemma~\ref{lemKlM} that 
\begin{align}
&\P\set{\Th \sqtx{is} L\text{-level spacing for}\; H }\\ \notag &  \hspace{50pt}\ge 1 -2^{2\alpha-1} \wtilde{K}^2 \pa{ \diam \supp \mu + 2d \eps +1}\e^{-(2\alpha-1)L^\beta}\abs{\Th}^2.
\end{align}
\end{proof}

\section{Preparation for the multiscale analysis}\label{secprep}

We consider a fixed   discrete  Schr\"odinger operator
$H=-\eps \Delta +V$ on $\ell^2(\Z^d)$, where  
 $V$ is a bounded potential  and $0<\eps \le \eps_0$ for a fixed $\eps_0$.

\subsection{Subsets and boundaries}
 
  Let  $\Phi \subset \Theta\subset \Z^d$.  We set the boundary, exterior boundary, and interior boundary of $\Phi$ relative to $\Theta$, respectively,  by
 \begin{align}\label{defbdry}
  \boldsymbol{ \partial}^{ \Theta} \Phi &=\set{(u,v) \in \Phi\times \pa{\Theta\setminus \Phi}; \  \abs{u-v}=1},
   \\
 \partial_{\mathrm{ex}}^{ \Theta} \Phi &=\set{v \in\pa{\Theta\setminus \Phi}; \ (u,v) \ \in  \boldsymbol{ \partial}^{ \Theta} \Phi\qtx{for some}u \in \Phi},\notag
    \\
 \partial^{ \Theta}_{\mathrm{in}}\Phi &=\set{u \in {\Phi}; \ (u,v) \ \in  \boldsymbol{ \partial}^{ \Theta} \Phi\qtx{for some}v \in \Theta\setminus \Phi}.\notag
   \end{align}

 We have  
   \beq\label{Hdecomp}
H_{ \Theta}= H_{ \Phi}\oplus H_{ \Theta\setminus  \Phi} + \eps\Gamma_{ \boldsymbol{ \partial}^{ \Theta}  \Phi}
\qtx{on} \ell^2( \Theta)=\ell^2( \Phi)\oplus \ell^2( \Theta\setminus \Phi),
\eeq 
\beq\label{Hdecomp1}
\text{where}\quad \Gamma_{   \boldsymbol{ \partial}^{ \Theta}  \Phi}(u,v)=
\begin{cases}
-1 & \text{if either}\  (u,v) \sqtx{or}(v,u) \in   \boldsymbol{ \partial}^{ \Theta}  \Phi\\
\ \ 0 & \text{otherwise}
\end{cases}.
\eeq

Given a box   $\La_L\subset \Th \subset \Z^d$,  for each
 $v\in \partial_{\mathrm{ex}}^\Th { \La_L} $ there exists a unique  $\hat{v} \in   \partial_{\mathrm{in}}^\Th { \La_L} $
 such that $(\hat{v},v)\in \boldsymbol{\partial}^\Th { \La_L} $,  which implies $\abs{\partial_{\mathrm{ex}}^\Th { \La_L} }= \abs{\boldsymbol{ \partial}^\Th { \La_L} }$.  Given $v \in \Th$, we define  $\hat{v}$ as above if $v\in \partial_{\mathrm{ex}}^\Th{ \La_L} $, and set  $\hat{v}= v $ otherwise.
For $ L\ge 2$  
 we have  
  \beq
\abs{\partial_{\mathrm{in}}^{ \Th} \La_L }\le\abs{\partial_{\mathrm{ex}}^{ \Th} \La_L } =\abs{ \boldsymbol{ \partial}^{ \Th} \La_L }\le s_{d} L^{d-1}, \qtx{where} s_d= 2^{d} d.
\eeq

Let   $\Psi \subset \Th \subset \Z^d$.  Given $t\ge 1$, we set
\begin{align}\notag
 \Psi^{\Th,t}& = \set{y\in \Psi;   \; \La_{2 t}(y)\cap \Th \subset \Psi}= \set{y\in \Psi;   \; \dist \pa{y,{ \Th}\setminus \Psi }> \fl{t}},\\ \label{defLatTh} 
 {\partial}_{\mathrm{in}}^{\Th,t} \Psi  & = \Psi \setminus  \Psi^{\Th,t} =\set{y\in \Psi;   \;  \dist \pa{y,{ \Th}\setminus \Psi } \le  \fl{t}},\\ \notag   
 {\partial}^{\Th,t} \Psi  & =  {\partial}_{\mathrm{in}}^{\Th,t} \Psi \cup \partial_{\mathrm{ex}}^{ \Th} \Psi.
 \end{align}
 Note that  $  \Psi^{\Th,t}= \Psi^{\Th,\fl{t}}$.   For a box $\La_L(x)\subset \Th \subset \Z^d$ we write   $\La_L^{\Th,t}(x)= \pa{\La_L(x)}^{\Th,t}$.  We also set $ \La_L^{t}(x)= \La_L^{\Z^d,t}(x)$.

\subsection{Generalized eigenfunctions}

 \begin{definition}\label{defgeneig}
 Given ${\Theta}\subset \Z^d$, a function  $\psi\colon{\Theta} \to \C$ is called a generalized eigenfunction for $H_{\Theta}$ with generalized eigenvalue $\lambda \in \R$ if $\psi$ is not identically
 zero and 
 \beq\label{pointeig2}   
 -\eps\!\! \sum_{\substack{y\in\Th\\ |y-x|=1}}\!\! \psi(y) + (V(x)-\lambda) \psi(x)  =0 \qtx{for all} x\in{\Theta},
 \eeq
 or, equivalently,
 \beq \label{pointeig} 
 \scal{(H_\Th -\lambda)\vphi, \psi}=0 \qtx{for all} \vphi \in \ell^2(\Th)\quad \text{with finite support}.
 \eeq
 In this case we call $(\psi,\lambda)$ a generalized eigenpair for $H_{\Theta}$.
\end{definition}

 If $\psi\in \ell^2(\Th)$, $\psi$ is an eigenfunction for   $H_{\Theta}$ with eigenvalue $\lambda$.
If $\Th$ is finite there is no difference between  generalized eigenfunctions and eigenfunctions.  For arbitrary $\Th$ the difference is that we do not require generalized eigenfunctions to be in $\ell^2(\Th)$, we only require the pointwise equality in \eq{pointeig}.

\subsection{Eigenpairs and eigensystems}

Let $\Theta \subset \Z^d$ and consider  an eigensystem $\set{(\vphi_j,\lambda_j)}_{j\in J}$ for $H_\Th$.  We  have
 \begin{align}\label{deltaya}
 \delta_y &=\sum_{j\in J} \overline{\vphi_j(y)} \vphi_j \qtx{for all} y \in \Th,\\ \label{psiy}
 \psi(y) &=\scal{\delta_y,\psi}  =   \sum_{j\in J} {\vphi_j(y)} \scal{ \vphi_j,\psi}  \qtx{for all} \psi \in \ell^2(\Th) \qtx{and}  y \in \Th.
 \end{align}

\begin{lemma}\label{lemdecayvphi}     Let  $\Phi \subset \Theta\subset\Z^d $ and suppose $(\vphi,\lambda)$ is an eigenpair for $H_\Phi$.   Then
\beq\label{HTheta}
\pa{\pa{H_{\Theta} -\lambda}\vphi}(y)= \eps \pa{ \textstyle\sum_{\substack{y^\pr \in \partial_{\mathrm{in}}^{ \Theta}{\Phi} \\ \abs{y^\pr-y}=1}}\vphi(y^\pr)} \Chi_{  \partial_{\mathrm{ex}}^{ \Theta}{\Phi}}(y)\qtx{for all} y\in \Th.
\eeq
Moreover, we have 
	\beq\label{distspt}
\dist\pa{\lambda, \sigma(H_{\Theta})}\le\norm{\pa{H_{\Theta} -\lambda}{\vphi} }\le  (2d-1) \eps\abs{ \partial_{\mathrm{ex}}^{ \Theta}{\Phi}}^{\frac 12}\norm{\vphi_{ \partial_{\mathrm{in}}^{ \Theta}{\Phi}}}_\infty.
\eeq

In the special case when $\Phi$ is a box $\La_L$, we have
\beq\label{HThetaL}
\pa{\pa{H_{\Theta} -\lambda}\vphi}(y)= \eps \vphi(\hat{y}) \Chi_{  \partial_{\mathrm{ex}}^{ \Theta}{\La_L}}(y)\qtx{for all} y\in \Th,
\eeq
and
	\beq\label{distspt2}
\dist\pa{\lambda, \sigma(H_{\Theta})}\le\norm{\pa{H_{\Theta} -\lambda}{\vphi} }\le   \eps  \sqrt{s_d}\,L^{\frac {d-1}2}\norm{\vphi_{  \partial_{\mathrm{in}}^{ \Theta}{\La_L}}}_\infty.
\eeq

\end{lemma}

\begin{proof}   We have 
 \beq\label{HTheta6}
\pa{\pa{H_{\Theta} -\lambda}\vphi}(y) = 
\begin{cases}
0 & \text{if} \   y \in    {\Phi}\\
\eps  \sum_{\substack{y^\pr \in \partial_{\mathrm{in}}^{ \Th}{\Phi} \\ \abs{y^\pr-y}=1}}\vphi(y^\pr)& \text{if} \   y \in    \partial_{\mathrm{ex}}^{ \Theta}{\Phi}\\0 & \text{if} \   y \in \Th\setminus\pa{\Phi \cup  \partial_{\mathrm{ex}}^{ \Theta}{\Phi}}
\end{cases},
\eeq
which is the same as \eq{HTheta}.  It follows that 
\begin{align}\label{HTheta1}
\abs{\pa{\pa{H_{\Theta} -\lambda}\vphi}(y) }\le (2d-1) \eps \Chi_{ \partial_{\mathrm{ex}}^{ \Theta}{\Phi}}(y)\norm{\vphi_{ \partial_{\mathrm{in}}^{ \Theta}{\Phi}}}_\infty \qtx{for all} y \in \Theta,
\end{align}
which yields \eq{distspt}.
\end{proof}

\begin{lemma}  \label{lemdeltaeta}   Let  $\Theta\subset \Z^d$ and $0<4\delta< \eta$.  Suppose:

 \begin{enumerate}
 \item  $\mu$ is a simple eigenvalue of $H_\Theta$ with normalized eigenfunction $\psi_\mu$, with    $\dist \pa{\mu, \sigma(H_\Theta)\setminus\set{\mu}} \ge \eta $.

\item  $ \norm{\pa{H_{\Theta} -\lambda}\vphi }\le \delta$, where $\vphi\in \ell^2(\Theta)$ with $\norm{\vphi}=1$ and  $\lambda \in \R$  with $\abs{\lambda -\mu} \le \delta$.

\end{enumerate}
Define $ \vphi^\perp$ by $\vphi= \scal{\psi_\mu,\vphi}\psi_\mu + \vphi^\perp$.
Then we have
\beq
\abs{\scal{\psi_\mu,\vphi}}^2\ge 1 - \frac{2\delta^2 }{\eta^2}  \qtx{and} \norm{ \vphi^\perp}\le \frac{\sqrt{2}\delta}{\eta}.
\eeq
Moreover, if we set  $\what{\vphi}= {\vtheta} \vphi$, where ${\vtheta} \in \C$ with $\abs{\vtheta}=1$ is chosen so 
\beq\label{aligning}
\scal{\psi_\mu,\what{\vphi}}=\abs{\scal{\psi_\mu,\vphi}}>0,
\eeq
we have
\beq\label{deltaeta}
\norm{\what{\vphi} -\psi_\mu} \le \frac{3\delta }{2\eta} < \frac{2\delta }{\eta}.
\eeq

 \end{lemma}

\begin{proof}We have
\beq
\vphi= \vphi_\mu + \vphi^\perp,\qtx{where}  \vphi_\mu= \scal{\psi_\mu,\vphi}\psi_\mu \qtx{and}  \scal{\psi_\mu,\vphi^\perp}=0.
\eeq
Thus
\begin{align}
\pa{H_{\Theta} -\lambda}\vphi = (\mu -\lambda)\vphi_\mu  + \pa{H_{\Theta} -\lambda}\vphi^\perp,
\end{align}
and
\begin{align}
\norm{ \pa{H_{\Theta} -\lambda}\vphi^\perp}&\ge \norm{ \pa{H_{\Theta} -\mu}\vphi^\perp}-\norm{ \pa{\mu -\lambda}\vphi^\perp}\\ \notag & \ge
\pa{\eta- \abs{\mu -\lambda}} \norm{ \vphi^\perp},
\end{align}
which gives
\begin{align}
\delta^2&\ge \norm{\pa{H_{\Theta} -\lambda}\vphi}^2 = \pa{\mu -\lambda}^2\norm{\vphi_\mu}^2  +\norm{ \pa{H_{\Theta} -\lambda}\vphi^\perp}^2\\ \notag
& \ge   \pa{\mu -\lambda}^2\norm{\vphi_\mu}^2 +\pa{\eta- \abs{\mu -\lambda}}^2 \norm{ \vphi^\perp}^2 \\ \notag
&= {\pa{\mu -\lambda}^2\pa{1-\norm{ \vphi^\perp}^2 } +\pa{\eta- \abs{\mu -\lambda}}^2 \norm{ \vphi^\perp}^2}
\\ \notag
&= {  \pa{\mu -\lambda}^2 + \pa{\pa{\eta- \abs{\mu -\lambda}}^2-  \pa{\mu -\lambda}^2}\norm{ \vphi^\perp}^2 }\\ \notag
&= {\pa{\mu -\lambda}^2 +\pa{\eta^2- 2\eta \abs{\mu -\lambda}}\norm{ \vphi^\perp}^2 },
\end{align}
and we conclude, using $4\delta< \eta$, that
\beq
\norm{ \vphi^\perp}^2\le \frac{\delta^2-\pa{\mu -\lambda}^2 } {{\eta^2- 2\abs{\mu -\lambda}\eta}}
\le \frac{\delta^2 }{\eta\pa{\eta -2 \delta} }\le \frac{2\delta^2 }{\eta^2}, \qtx{so} \abs{\scal{\psi_\mu,\vphi}}^2\ge 1 - \frac{2\delta^2 }{\eta^2}.
\eeq
It follows that, if we set set  $\what{\vphi}= {\vtheta} \vphi$, where ${\vtheta} \in \C$ with $\abs{{\vtheta}}=1$ is chosen so $\scal{\psi_\mu,\what{\vphi}}>0$, we have
\begin{align} 
 \norm{\what{\vphi} -\psi_\mu}^2 &=
\abs{1- \scal{\psi_\mu,\what{\vphi}}}^2 + \norm{ \vphi^\perp}^2= \abs{1- \abs{\scal{\psi_\mu,{\vphi}}}}^2 + \norm{ \vphi^\perp}^2   \\ \notag  &
\le  \abs{1- \pa{1 - \frac{2\delta^2 }{\eta^2}}^{\frac 12}}^2 + \frac{2\delta^2 }{\eta^2}
\le {{\frac{9\delta^2 }{4\eta^2}}},
\end{align}
where we have used  $1- \pa{1-x}^{\frac 12} \le x$ for  $x\in[0,1]$ and  $4\delta < \eta$.
\end{proof}

\subsection{Localizing boxes}

\begin{lemma}\label{lemdecayvphi2}   Let $\La_L$ be a box, $x\in \La_L$,   $m\ge m_->0$,  and  suppose $\vphi\in \ell^2(\La_L)$  is an  $(x,m)$-localized eigenfunction of $H_{\La_L}$ with eigenvalue $\lambda \in \R$. 
Then 
 for  all subsets  $\La_L\subset \Theta \subset \Z^d$ such that $x \in \La_L^{\Th,L_\tau}$  
 we have 
\beq\label{appeig}
\dist\pa{\lambda, \sigma(H_{\Theta})}\le \norm{\pa{H_{\Theta} -\lambda}{\vphi} }\le   e^{-m_1 L_\tau},
\eeq
 for 
 $L\ge \cL(d,m_-,\eps_0)$,
where
\beq
 m_1=m_1(L)\ge  m\pa{1- C_{d,m_-,\eps_0}\tfrac {\log L}{L^\tau}}. \label{m1}
\eeq
\end{lemma}

\begin{proof}

Since  $x \in \La_L^{\Th,L_\tau}$, we have $\dist \pa{x,\partial_{\mathrm{in}}^\Theta \La_L}\ge {L_\tau}$, so it follows from \eq{distspt2} in Lemma~\ref{lemdecayvphi} that  
\beq\label{appeig4}
 \norm{\pa{H_{\Theta} -\lambda}{\vphi} }\le  \eps  \sqrt{ s_d }L^{\frac {d-1}2}\norm{\vphi_{ \partial_{\mathrm{in}}^\Theta{\La_L}}}_\infty\le \eps_0 \sqrt{ s_d }L^{\frac {d-1}2} \e^{-m {L_\tau}}\le  e^{-m_1 L_\tau},
 \eeq
 where $m_1$ is  as in \eq{m1}.
\end{proof}

  \begin{lemma}\label{lemdecay2}   Let  $\Theta\subset \Z^d$,  fix  $ m_->0$, and let $m\ge m_-$.  Let $\psi\colon{\Theta} \to \C$ be a generalized eigenfunction for $H_{\Theta}$ with generalized eigenvalue $\lambda \in \R$. 
Consider a box  $ \La_\ell \subset{\Theta}$  such that $\La_\ell$ is   $m$-localizing   with   an $m$-localized eigensystem $\set{\vphi_u,\nu_u}_{u \in  {\Lambda}_{\ell}}$, and suppose
\beq\label{distpointeig}
\abs{\lambda - \nu_u}
\ge \tfrac 1 2\e^{-L^\beta}\qtx{for all} u\in \La_\ell^{{\Theta},\ell_\tau} .
\eeq
 Then the following holds for sufficiently large $L$:
\begin{enumerate}
\item If  $y\in \La_\ell^{{\Theta},2\ell_\tau}$ we have
\beq\label{decayest00}
\abs{\psi(y)}\le \e^{-m_2{\ell_{\tau}}}\abs{\psi(y_1)} \qtx{for some} y_1 \in {\partial}^{\Th, 2\ell_\tau }  \Lambda_{\ell},
\eeq
where
\beq\label{m414}
m_2=m_2(\ell) \ge m\pa{1-   C_{d,m_-,\eps_0} \ell^{\gamma \beta -\tau}}.
\eeq

\item If $y\in \La_\ell^{\Theta,2\ell_{{\ttau}}}$ we have
\beq\label{decayest12}
\abs{\psi(y)}\le \e^{-m_3{\norm{y_2-y}}}  \abs{\psi(y_2)} \sqtx{for some} y_2 \in{\partial}^{\Th,\ell_{{\ttau}}}  \Lambda_{\ell},  \sqtx{so}\norm{y_2-y}>\ell_{{\ttau}}, 
\eeq
where
\beq\label{m4}
m_3= m_3(\ell)\ge m\pa{1-   C_{d,m_-,\eps_0} \ell^{\frac {\tau -1}2}}.
\eeq
\end{enumerate}
  \end{lemma}

\begin{proof} Given $y\in \La_\ell$, we have  (see \eq{psiy} and \eq{defLatTh}) 
\begin{align}\label{sum000}
\psi(y)  =   \sum_{u\in \Lambda_{\ell} }{\vphi_u(y)} \scal{ \vphi_u,\psi}=  \sum_{u\in  \La_\ell^{{\Theta},{\ell_\tau}}}{\vphi_u(y)} \scal{ \vphi_u,\psi} +   \sum_{u\in \partialin^{\Th,{\ell_\tau}} \La_\ell}{\vphi_u(y)} \scal{ \vphi_u,\psi}.
\end{align}

Let us fix  $u \in   \La_\ell^{{\Theta},{\ell_\tau}}$.  We have  
   $\abs{\lambda -\nu_u}\ge \tfrac 1 2\e^{-L^\beta}$ by \eq{distpointeig}.
 Since  $\La_\ell$ is finite,   \eq{pointeig}  gives
\beq
\scal{ \vphi_u,\psi}= \pa{\lambda- \nu_u }^{-1}\scal{\pa{H_{\Theta}- \nu_u} \vphi_u, \psi} .
\eeq
It follows from  \eq{HThetaL} that
\begin{align}\label{vphipsi399}
\abs{{\vphi_u(y)}\scal{ \vphi_u,\psi}}& \le 2\e^{L^\beta} \eps  \sum_{  v \in   \partial_{\mathrm{ex}}^{{\Theta}}  \Lambda_{\ell}}\abs{ {\vphi_u(y)}\vphi_u (\hat{v})}\abs{\psi(v)}.
\end{align}
  If $v^\pr \in \partial_{\mathrm{in}}^{{\Theta}}   \Lambda_{\ell} $, we have $ \norm{v^\pr-u}\ge {\ell_{\tau}}$, so it follows from   \eq{hypdec} that  
\beq
\abs{\vphi_u (v^\pr)}\le \e^{-m\norm{v^\pr-u}}\le \e^{-m \ell_{\tau}}. 
\eeq 
Since $\norm{\vphi_u}=1$, we get from 
\eq{vphipsi399} that 
\begin{align}\label{vphipsi59987}
\abs{{\vphi_u(y)}\scal{ \vphi_u,\psi}}& \le 2 \eps\e^{L^\beta}  \e^{-m{\ell_\tau}}\!\!\sum_{  v \in    \partial_{\mathrm{ex}}^{{\Theta}} \Lambda_{\ell}} \abs{\psi(v)}  
\le  2 \eps\e^{L^\beta} s_d\ell^{d-1} \e^{-m{\ell_\tau}}\abs{\psi(v_1)},
\end{align} 
for some $ v_1 \in\partial_{\mathrm{ex}}^{{\Theta}} \Lambda_{\ell}$.
It follows that
\beq\label{sum29987}
\abs{ \sum_{u\in \La_\ell^{{\Theta},{\ell_\tau}}}{\vphi_u(y)} \scal{ \vphi_u,\psi}}\le 2 \eps s_d\e^{L^\beta} \ell^{2d-1} \e^{-m{\ell_\tau}}\abs{\psi(v_2)}
\eeq
for some $v_2 \in\partial_{\mathrm{ex}}^{{\Theta}}  \Lambda_{\ell}$.

Let $y\in \La_\ell^{{\Theta},2\ell_\tau}$.  If $u\in \partial_{\mathrm{in}}^{\Theta,\ell_\tau} \La_\ell$     we  have 
$\norm{u-y}\ge 2\ell_\tau -\ell_\tau=\ell_\tau$, so 
\eq{hypdec} gives $\abs{\vphi_u(y)} \le \e^{-m {\norm{u-y}}}\le   \e^{-m\ell_\tau}$, and thus 
\begin{align}\label{sum19987}
\abs{\sum_{u\in \partial_{\mathrm{in}}^{\Theta,\ell_\tau} \La_\ell}{\vphi_u(y)} \scal{ \vphi_u,\psi} }\le \ell^{d }   \e^{-m\ell_\tau} \norm{\psi \Chi_{\La_\ell}}\le  \ell^{\frac {3d}2 }   \e^{-m\ell_\tau}\abs{\psi(v_3)} \end{align}
 for some  $v_3 \in \Lambda_{\ell}$.
 Combining \eq{sum000}, \eq{sum29987}, and       \eq{sum19987}, we get ($\ell$ large)
\beq\label{m41455}
\abs{\psi(y)}\le (1+\eps_0) \e^{L^\beta} \ell^{2d} \e^{-m{\ell_\tau}}\abs{\psi(y_1)}\le  \e^{-m_2{\ell_{\tau}}}\abs{\psi(y_1)} .
\eeq
for some $y_1 \in \La_\ell \cup\partial_{\mathrm{ex}}^{{\Theta}}  \Lambda_{\ell}$,  where $m_2$ is given in \eq{m414}.
(Note $\tau> \gamma \beta$.)
If $y_1\in \partial^{\Th, 2\ell_{\tau}}\La_\ell $ we have  \eq{decayest00}.
If not, we repeat the procedure to  estimate $\abs{\psi(y_1)}$.      Since we can suppose $\psi(y)\ne 0$ without loss of generality,  it is clear that the procedure must stop after finitely many times,  and at that time we must have \eq{decayest00}.

Now let   $y\in \La_\ell^{{\Theta},\ell_{\ttau}}$, so  $\norm{y-v^\pr} \ge\ell_{\ttau}$ for $v^\pr \in \partialin^{{\Theta}}  \Lambda_{\ell} $.   Thus for   $u \in   \La_\ell^{{\Theta},{\ell_\tau}}$  and  $v^\pr \in \partialin^{{\Theta}}  \Lambda_{\ell}$ we have 
 \beq\label{vphipsi499}
\abs{{\vphi_u(y)}\vphi_u (v^\pr)}\le 
  \begin{cases}
\e^{-m\pa{\norm{y-u}+\norm{v^\pr-u}}}\le \e^{-m \norm{v^\pr -y}} & \qtx{if} \norm{y-u}\ge 
{\ell_ \tau}\\
\e^{-m\pa{\norm{v^\pr-u}}}\le \e^{-m^\pr_1\norm{v^\pr -y}}   &\qtx{if} \norm{y-u}   <\ell_\tau
\end{cases},
\eeq
where 
\beq
 m^\pr_1\ge  m(1-   2 \ell^{\tau-\ttau})=  m(1-  2  \ell^{\frac {\tau -1}2}),
\eeq
since for $\norm{y-u}   <{\ell_ \tau}$ we have 
\beq
\norm{v^\pr-u} \ge \norm{v^\pr-y}- \norm{y-u}\ge \norm{v^\pr-y} - {\ell_\tau}\ge  \norm{v^\pr-y} \pa{1- \tfrac {\ell_\tau}{\ell_ {\ttau}}}.
\eeq
Combining \eq{vphipsi399} and \eq{vphipsi499}, we get
\begin{align}\label{vphipsi599}
&\abs{{\vphi_u(y)}\scal{ \vphi_u,\psi}}\le 2 \eps\e^{L^\beta}   \sum_{  v \in   \partialex^{{\Theta}}  \Lambda_{\ell}} \e^{-m^\pr_1 \pa{\norm{v -y}-1}}\abs{\psi(v)}\\ & \notag   \qquad 
\le  2 \eps\e^{L^\beta} s_d\ell^{d-1}\ \e^{-m^\pr_1 \pa{\norm{v_1 -y}-1}}\abs{\psi(v_1)}\le
 \e^{-m^\pr_2 {\norm{v_1 -y}}}\abs{\psi(v_1)}
\end{align}
for some $v_1 \in\partialex^{{\Theta}}  \Lambda_{\ell}$, where
\beq
m^\pr_2\ge  m^\pr_1\pa{1- C_{d,m_-,\eps_0} \ell^{\gamma \beta- \ttau}}\ge m\pa{1-   C_{d,m_-,\eps_0} \ell^{\frac {\tau -1}2}},
\eeq
where we used  $\norm{v_1 -y}\ge  \ell_{\ttau}$ and  $\ttau >\tau> \gamma \beta$.
It follows that
\beq\label{sum299}
\abs{ \sum_{u\in \La_\ell^{{\Theta},{\ell_\tau}}}{\vphi_u(y)} \scal{ \vphi_u,\psi}}\le \ell^d \e^{-m^\pr_2 {\norm{v_2 -y}}}\abs{\psi(v_2)}\le \e^{-m^\pr_3 {\norm{v_2 -y}}}\abs{\psi(v_2)}
\eeq
for some $v_2 \in\partialex^{{\Theta}}  \Lambda_{\ell}$,
where
\beq
m^\pr_3\ge m^\pr_2\pa{1- C_{d,m_-,\eps_0} \tfrac {\log \ell}{\ell^{\ttau}}}\ge m\pa{1-   C_{d,m_-,\eps_0} \ell^{\frac {\tau -1}2}}.
\eeq

If $u\in \partialin^{\Th,\ell_\tau} \La_\ell$    we must have 
$\norm{u-y}\ge \ell_{\ttau} - \ell_{\tau} >\frac 1 2 \ell_{\ttau}$, so 
\eq{hypdec} gives $\abs{\vphi_u(y)} \le \e^{-m {\norm{u-y}}}$ and, using \eq{hypdec2} for   $\vphi_u$, we get 
\begin{align}
&\abs{\scal{ \vphi_u,\psi}}=\abs{ \sum_{v \in \Lambda_{\ell} } {\vphi_u(v)}\psi(v)}  \le   \sum_{v \in \Lambda_{\ell}} \e^{-m \pa{\norm{v-u}-\ell_\tau}} \abs{\psi(v)} ,
\end{align}
so we conclude that
\begin{align}
&\abs{{\vphi_u(y)} \scal{ \vphi_u,\psi}} \le  \sum_{v \in \Lambda_{\ell}} \e^{-m \pa{\norm{u-y}-\ell_\tau + \norm{v-u}}} \abs{\psi(v)}  \\ \notag & \quad \le   (\ell+1)^d \e^{-m \pa{\norm{u-y}-\ell_\tau }- m \norm{v_3-u}} \abs{\psi(v_3)} \\ \notag & \quad  \le  \e^{-m^\pr_4 \norm{u-y} - m \norm{v_3-u}} \abs{\psi(v_3)}  \\ & \notag\quad  \le    \e^{-m^\pr_4 \max\set{\norm{v_3-y},\norm{u-y}}} \abs{\psi(v_3)}\le    \e^{-m^\pr_4 \max\set{\norm{v_3-y},\ \frac 1 2  \ell_{\ttau}}}  \abs{\psi(v_3)},
\end{align}
for some $ v_3 \in \Lambda_{\ell} $, 
where we used $ \norm{u-y} \ge \frac 1 2  \ell_{\ttau}$ and took 
\begin{align}
m^\pr_4 &\ge   m(1-  C_{d,m_-}  \ell^{\frac {\tau -1}2}).
\end{align}
It follows that 
\begin{align}\label{sum199}
\abs{\sum_{u\in\partialin^{\Th,\ell_\tau} \La_\ell}{\vphi_u(y)} \scal{ \vphi_u,\psi} }&\le  \ell^{d } \e^{-m^\pr_4\max\set{\norm{v_3-y},\ \frac 1 2  \ell_{\ttau }}} \abs{\psi(v_3)} \\
& \notag \le \e^{-m^\pr_5 \max\set{\norm{v_3-y},\ \frac 1 2 \ell_{\ttau }}} \abs{\psi(v_3)} 
\end{align}
for some  $ v_3 \in \Lambda_{\ell}$, where
\begin{align}
m^\pr_5 \ge m^\pr_4 (1-   C_{d,m_-}\pa{ \log \ell}  \ell^{-\ttau})\ge m(1-  C_{d,m_-} \ell^{\frac {\tau -1}2}).
\end{align}
 Combining \eq{sum000}, \eq{sum299}, and       \eq{sum199}, we get
\beq\label{decayest0022} 
\abs{\psi(y)}\le \e^{-m_3 \max\set{\norm{y_1-y},\ \frac 1 2  \ell_{\ttau }}} \abs{\psi(y_1)} \qtx{for some} y_1 \in \La_\ell \cup\partialex^{{\Theta}}  \Lambda_{\ell},
\eeq
where $m_3$ is given in \eq{m4}.

 If in \eq{decayest0022}  we get $y_1\notin \partial^{{\Th, \ell_{\ttau } }}\La_\ell $ we  repeat the procedure to estimate $\abs{\psi(y_1)}$.    Since we can suppose $\psi(y)\ne 0$ without loss of generality, the procedure must stop after finitely many times,  and at that time we must have
\beq\label{decayest339} 
\abs{\psi(y)}\le \e^{-m_3 \max\set{\norm{\wtilde{y}-y},\ \frac 1 2  \ell_{\ttau }}}  \abs{\psi(\wtilde{y})} \qtx{for some} \wtilde{y} \in \partial^{{\Th, \ell_{\ttau } }}\La_\ell  .
\eeq

If $y\in \La_\ell^{\Theta,2\ell_{{\ttau}}}$, \eq{decayest12} is an immediate consequence of \eq{decayest339}.
\end{proof}

 \begin{lemma}\label{lem:ident_eigensyst} Let the finite set  $\Theta\subset\Z^d$ be $L$-level spacing for $H$, and  let  $\set{(\psi_\lambda,\lambda)}_{\lambda \in \sigma(H_\Th)}$ be an eigensystem for $H_\Th$.  
 
 Then the following holds for sufficiently large $L$:

 \begin{enumerate}

\item   Let   $\Lambda_\ell(a)\subset \Th$,  where  $ a\in \R^d$, be an  $m$-localizing box with an $m$-localized eigensystem $\set{(\vphi_x\up{a}, \lambda_x\up{a})}_{x\in \La_\ell(a)}$.

\begin{enumerate}
\item

There exists an injection    
\beq\label{injection}
x\in \La_\ell^{\Th,{\ell_\tau}}(a) \mapsto  \wtilde{\lambda}\up{a}_x\in \sigma(H_\Th),
\eeq
such that for all $ x\in \La_\ell^{\Th,{\ell_\tau}}(a)$ we have
\begin{align}
\abs{ \wtilde{\lambda}\up{a}_x -\lambda\up{a}_x} \le \e^{-m_1{\ell_\tau}},\mqtx{with} m_1=m_1(\ell) \;\;\text{as in \eq{m1}},  \label{tildedist132}
\end{align}
and, redefining each ${\vphi\up{a}_x}$ by multiplying it by a suitable phase factor (as in \eq{aligning}),  
\beq\label{difeq86}
\norm{\psi_{\wtilde{\lambda}\up{a}_x}-{\vphi\up{a}_x}}\le 2 \e^{-m_1{\ell_\tau}}\e^{L^\beta}.
\eeq
\item Let
\beq
 \sigma_{\set{a}}(H_\Th):= \set{ \wtilde{\lambda}\up{a}_x;\ x\in \La_\ell^{\Th,{\ell_\tau}}(a)}.
\eeq
Then for 
 $\lambda \in \sigma_{\set{a}}(H_\Th)$  we have
\beq \label{psidecout}
\abs{\psi_\lambda(y)}\le  2 \e^{-m_1{\ell_\tau}}\e^{L^\beta} \qtx{for all} y\in \Th \setminus \La_\ell(a).
\eeq

\item If  $\lambda \in \sigma(H_\Th)\setminus  \sigma_{\set{a}}(H_\Th)$, we have
 \beq\label{distpointeiga}
 \abs{\lambda - \lambda\up{a}_x}\ge  \tfrac 12 \e^{-\L^\beta} \qtx{for all} x\in \La_\ell^{{\Theta},\ell_\tau}(a),
 \eeq
and 
\beq
\abs{\psi_\lambda(y)}\le \e^{-m_2 {{\ell_\tau}}} \sqtx{for} y\in \La_\ell^{{\Theta},2{{\ell_\tau}}}(a), \sqtx{with} m_2=m_2(\ell) \;\text{as in \eq{m414}}.   \label{psidecgood}
\eeq
Moreover, if  $y\in \La_\ell^{{\Theta},2\ell_{\ttau}}(a)$ we have 
\beq\label{psidecgoodpr}
\abs{\psi_\lambda(y)}\le \e^{-m_3{\norm{y_1-y}}}  \abs{\psi_\lambda(y_1)} \qtx{for some} y_1 \in{\partial}^{\Th,\ell_{\ttau }}  \Lambda_{\ell}(a),
\eeq
 with $m_3=m_3(\ell)$ as is  in \eq{m4}.

  \end{enumerate}        
\item    Let  $\set{\Lambda_\ell(a)}_{a\in \cG} $, where $\cG \subset \R^d$ and  $\Lambda_\ell(a)\subset \Th$ for all $a \in \cG$,   be a collection of  $m$-localizing boxes with $m$-localized eigensystems   $\set{(\vphi_x\up{a}, \lambda_x\up{a})}_{x\in \La_\ell(a)}$, and set  
\begin{align}\label{defcE}
 \cE_\cG^\Th(\lambda)&=  \set{{\lambda\up{a}_x}; \;a\in \cG,\, x\in \La_\ell^{\Th,{\ell_\tau}}(a), \sqtx{and} \wtilde{\lambda}\up{a}_x=\lambda}\sqtx{for} \lambda\in \sigma(H_\Th), \\ \notag
\sigma_\cG(H_\Th)&=\set{ \lambda\in \sigma(H_\Th); \, \cE_\cG^\Th(\lambda)\ne \emptyset} =\textstyle\bigcup_{a\in \cG} \sigma_{\set{a}}(H_\Th).
\end{align}

\begin{enumerate}
\item 
Let $a,b \in \cG$, $a\ne b$,  Then, for  $x\in  \La_\ell^{\Th,{\ell_\tau}}(a)$ and $y\in  \La_\ell^{\Th,{\ell_\tau}}(b)$, 
\beq \label{vphilambda}
 {\lambda\up{a}_x},{\lambda\up{b}_y}\in \cE_\cG^\Th(\lambda) \quad \Longrightarrow \quad  \norm{x-y} < 2 {{\ell_\tau}}.
\eeq
As a consequence,
\beq\label{sigmaab}
\Lambda_\ell(a) \cap \Lambda_\ell(b)=\emptyset   \quad \Longrightarrow \quad \sigma_{\set{a}}(H_\Th)\cap \sigma_{\set{b}}(H_\Th)=\emptyset  .
\eeq
 
  \item  If  $\lambda \in\sigma_\cG(H_\Th)$, we have
  \beq \label{psidecout63}
\abs{\psi_\lambda(y)}\le  2 \e^{-m_1{\ell_\tau}}\e^{L^\beta}   \sqtx{for all} y\in\Th \setminus {\Th}_{\cG}, \sqtx{where}  {\Th}_{\cG}:= \bigcup_{a\in \cG}\La_\ell(a).
\eeq

 \item  If  $\lambda \in \sigma(H_\Th)\setminus\sigma_\cG(H_\Th)$, we have
\beq\label{psidecgood63}
\abs{\psi_\lambda(y)}\le\e^{-m_2{{\ell_\tau}}}   \qtx{for all} y\in{\Th}_{\cG,\tau}:= \bigcup_{a\in \cG}\La_\ell^{{\Theta},2{{\ell_\tau}}}(a).
\eeq

\item  If $\abs{\Th} \le (L+1)^d$,
it follows that
\beq\label{cGsigma}
\abs{{\Th}_{\cG,\tau}}\le\abs{\sigma_\cG(H_\Th)} \le \abs{{\Th}_{\cG}}.
\eeq
  
 \end{enumerate}
\end{enumerate} 

\end{lemma}

\begin{proof}  Let   $\Lambda_\ell(a)\subset \Th$,  where  $ a\in \R^d$, be a $m$-localizing box with an $m$-localized eigensystem $\set{(\vphi_x\up{a}, \lambda_x\up{a})}_{x\in \La_\ell(a)}$.   Given $x\in \La_\ell^{\Th,\ell_{\tau}}(a)$, the existence of $\wtilde{\lambda}\up{a}_x\in \sigma(H_\Th)$ satisfying \eq{tildedist132} 
follows from Lemma~\ref{lemdecayvphi2}.  Uniqueness follows from  the fact that $\Theta$ is $L$-level spacing  and $\gamma \beta <\tau$. In addition, note that $\wtilde{\lambda}\up{a}_x \ne \wtilde{\lambda}\up{a}_y$ if $x,y \in  \La_\ell^{\Th,\ell_{\tau}}(a)$, $x\ne y$, because in this case we have
 \beq
\abs{ {\wtilde{\lambda}\up{a}_x} -\wtilde{\lambda}\up{a}_y}\ge  \abs{ {\lambda\up{a}_x} -\lambda\up{a}_y}- \abs{\wtilde{\lambda}\up{a}_x - \lambda\up{a}_x}  -  \abs{\wtilde{\lambda}\up{a}_y - \lambda\up{a}_y}\ge \e^{-\ell^\beta} - 2\e^{-m_1{\ell^\tau}}\ge  \tfrac 12 \e^{-\ell^\beta},
\eeq
$\La_\ell(a)$ is level spacing for $H$,  and $\beta < \tau$.
Moreover, it follows from Lemma~\ref{lemdeltaeta} that,  after multiplying  ${\vphi\up{a}_x}$ by a phase factor if necessary, we have \eq{difeq86}.

If $\lambda \in \sigma_{\set{a}}(H_\Th)$ , we have $\lambda=  \wtilde{\lambda}\up{a}_x$ for some $x\in \La_\ell^{\Th,{\ell_\tau}}(a) $, so \eq{psidecout} follows from \eq{difeq86} as
$\vphi\up{a}_x(y)=0$ for all  $y\in \Th \setminus \La_\ell(a)$.

 Let $\lambda \in \sigma(H_\Th)\setminus  \sigma_{\set{a}}(H_\Th)$.  Then for all $x\in \La_\ell^{{\Theta},\ell_\tau}(a)$ we have
 \beq\label{distpointsigeigab}
 \abs{\lambda - \lambda\up{a}_x}\ge  \abs{\lambda - \wtilde{\lambda}\up{a}_x} -  \abs{\wtilde{\lambda}\up{a}_x - \lambda\up{a}_x}\ge   \e^{-\L^\beta} -\e^{-m_1{\ell^\tau}}\ge \tfrac 12 \e^{-\L^\beta},
 \eeq
since $\Theta$ is $L$-level spacing for $H$, we have \eq{tildedist132}, and $\gamma \beta <\tau$. Thus  \eq{psidecgood} follows from Lemma~\ref{lemdecay2}(i) and $\norm{\psi_\lambda}=1$, and  \eq{psidecgoodpr} follows from Lemma~\ref{lemdecay2}(ii).

Now let 
$\set{\Lambda_\ell(a)}_{a\in \cG} $, where $\cG \subset \R^d$ and  $\Lambda_\ell(a)\subset \Th$ for all $a \in \cG$,   be a collection of  $m$-localizing boxes with $m$-localized eigensystems   $\set{(\vphi_x\up{a}, \lambda_x\up{a})}_{x\in \La_\ell(a)}$. 
Let  $\lambda\in \sigma(H_\Th)$, $a,b \in \cG$, $a\ne b$,  $x\in  \La_\ell^{\Th,\ell_{\tau}}(a)$, and $y\in  \La_\ell^{\Th,\ell_{\tau}}(b)$. Suppose ${\lambda\up{a}_x},{\lambda\up{b}_y}\in \cE^\Th_\cG(\lambda)$, where $\cE_\cG^\Th(\lambda)$ is given in \eq{defcE}.
It then  follows from  \eq{difeq86} that
\beq \label{difeq}
\norm{{\vphi\up{a}_x}- {\vphi\up{b}_y}}\le  4 \e^{-m_1{\ell_{\tau}}}\e^{L^\beta},
\eeq
so
\beq\label{disjs}
\abs{\scal{{\vphi\up{a}_x},{\vphi\up{b}_y}}}\ge  \Re \scal{{\vphi\up{a}_x},{\vphi\up{b}_y}}\ge 1- 8\e^{-2m_1{\ell_\tau}}\e^{2L^\beta} .
\eeq
On the other hand,  it follows from \eq{hypdec} that
\beq\label{disjs2}
\norm{x-y} \ge 2 {\ell_\tau} \quad \Longrightarrow  \quad \abs{\scal{{\vphi\up{a}_x},{\vphi\up{b}_y}}}\le   (\ell+1)^d  \e^{-m\ell_\tau    }.
\eeq
Combining \eq{disjs} and \eq{disjs2} we conclude that
\beq \label{vphiabxy}
{\lambda\up{a}_x},{\lambda\up{b}_y}\in \cE^\Th_\cG (\lambda) \quad \Longrightarrow \quad \norm{x-y} < 2 {\ell_{\tau}}.
\eeq

To prove \eq{sigmaab}, let $a,b \in \cG$, $a\ne b$. If $\Lambda_\ell(a) \cap \Lambda_\ell(b)=\emptyset $,  we have that
\beq
x\in  \La_\ell^{\Th,\ell_{\tau}}(a) \qtx{and}y\in  \La_\ell^{\Th,\ell_{\tau}}(b)\quad \Longrightarrow \quad  \norm{x-y} \ge 2\ell_{\tau},
\eeq
so it follows from  \eq{vphilambda}  that $\sigma_{\set{a}}(H_\Th)\cap \sigma_{\set{b}}(H_\Th)=\emptyset $.

Parts (ii)(b) and (ii)(c) are immediate consequence of parts (i)(b) and (i)(c), respectively.  To prove  part (ii)(d), note that, letting $P_\cG$ denote the orthogonal projection onto the span of $\set{\psi_\lambda; \ \lambda\in\sigma_\cG(H_{\Th}) }$, it follows from \eq{psidecgood63} that
\beq
\norm{(1-P_\cG) \delta_y} \le e^{-m_2{\ell_\tau}} \abs{\Th}^{\frac 1 2} \qtx{for all} y \in \Th_{\cG,\tau},
\eeq
so
\beq
\norm{(1-P_\cG) \Chi_{\Th_{\cG,\tau}}}\le \abs{{\Th_{\cG,\tau}}}^{\frac 1 2}  \abs{\Th}^{\frac 1 2}e^{-m_2{\ell_\tau}}\le \abs{\Th}e^{-m_2{\ell_\tau}}.
\eeq
If we assume $\abs{\Th} \le (L+1)^d$, we get
\beq
\norm{(1-P_\cG) \Chi_{\Th_{\cG,\tau}}}\le  (L+1)^d e^{-m_2{\ell_\tau}} < 1,
\eeq
so we conclude from   Lemma~\ref{lemPQ}  that 
\beq\label{Hallc43}
\abs{\Th_{\cG,\tau}} = \tr \Chi_{\Th_{\cG,\tau}} \le  \tr  P_\cG= \abs{\sigma_\cG(H_{\Th})}.
\eeq
A similar argument, using \eq{psidecout63}, proves  $\abs{\sigma_\cG(H_{\Th})}\le \abs{\Th_{\cG}}.$
\end{proof}

\subsection{Buffered subsets}  We will need to consider boxes $\La_\ell \subset \La_L $ that are  not $m$-localizing for $H$.  Instead of studying eigensystems for such boxes, we will surround them with a buffer of $m$-localizing boxes and study eigensystems for the augmented subset.

  \begin{definition}\label{defbuff} Let $\La_L=\La_L(x_0)$, $x_0 \in \R^d$, 
 and $m\ge m_->0$.  We call $\Ups\subset \La_L$ a buffered subset of $\La_L$   if 
 the following holds:
  
  \begin{enumerate}
\item  $\Ups $ is a connected set in $\Z^d$ of the form
\beq\label{defUpsinitial}
 \Ups= \bigcup_{j=1}^J \La_{R_j}(a_j)\cap \La_L,
 \eeq  
 where $J\in \N$, $a_1,a_2,\ldots, a_J \in \La^\R_L$, and $\ell \le R_j\le L$ for $j=1,2,\ldots,J$.
  
  \item ${\Upsilon}$  is $L$-level spacing for $H$.

  \item  There exists $\cG_\Ups \subset \La^\R_L$ such that:
  \begin{enumerate}
\item For all $a\in\cG_\Ups$ we have    $\La_\ell(a) \subset\Ups$, and    $\La_\ell(a)$ is  an  $m$-localizing box   for $H$.
\item  For all $y \in \partialin^{\La_L}\Ups$ there exists $a_y \in\cG_\Ups$ such that $y\in \La_\ell^{ {\Ups}, 2{\ell_\tau}}(a_y)$.
\end{enumerate}

  \end{enumerate}   
 In this case we set 
   \beq\label{defUpscheck}
\widecheck{\Upsilon} =\bigcup_{a \in \cG_\Ups}\La_\ell (a), \quad  \widecheck{\Upsilon}_{\tau} = \bigcup_{a \in \cG_\Ups}\La_\ell^{ {\Upsilon}, 2{\ell_\tau}}(a), \quad  \widehat{\Upsilon} = {\Upsilon} \setminus  \widecheck{\Upsilon},  \qtx{and} \widehat{\Upsilon}_\tau = {\Upsilon} \setminus  \widecheck{\Upsilon}_\tau .
\eeq
($\widecheck{\Upsilon} = \Upsilon_{\cG_\Ups}$ and $\widecheck{\Upsilon}_\tau = \Upsilon_{\cG_\Ups,\tau}$ in the notation of Lemma~\ref{lem:ident_eigensyst}.)

\end{definition}

  The set $\widecheck{\Upsilon}_{\tau}\supset \partialin^{\La_L}\Ups$ is a localizing buffer between  $ \widehat{\Upsilon}$ and $\La_L \setminus \Ups$, as shown in the following lemma.

 \begin{lemma}
 Let  $\Ups$  be a buffered subset of $\La_L$,  and let $\set{(\psi_\nu,\nu)}_{\nu \in \sigma(H_{\Upsilon})}$ be an eigensystem  for $H_{\Upsilon}$. 
   Let $\cG=\cG_\Ups$  and set
 \beq\label{sigmabad}
\sigma_\cB(H_{{\Upsilon}})= \sigma(H_{{\Upsilon}})\setminus \sigma_\cG(H_{{\Upsilon}}),
\eeq
where $\sigma_\cG(H_\Ups)$ is as in \eq{defcE}.  Then the following holds for sufficiently large $L$:

\begin{enumerate}
\item 
For all $\nu\in\sigma_\cB(H_{{\Upsilon}})$ we have 
\beq
\abs{\psi_\nu(y)} \le  e^{-m_2{\ell_\tau}} \sqtx{for all} y \in \widecheck{\Upsilon}_{\tau}, \sqtx{with} m_2=m_2(\ell) \;\text{as in \eq{m414}},  \label{vthetasmall}
\eeq
and
\beq
\abs{\widehat{\Upsilon}}\le  \abs{\sigma_\cB(H_{{\Upsilon}})}\le \abs{\widehat{\Upsilon}_{\tau}}.
\eeq

\item Let $ \La_L$ be level spacing for $H$, and let $\set{(\phi_\lambda,\lambda)}_{\lambda \in \sigma(H_{\La_L})}$ be an eigensystem  for $H_{\La_L}$. 
There exists an injection    
\beq\label{injbad}
 \nu\in\sigma_\cB(H_{{\Upsilon}})  \mapsto \widetilde{\nu}\in \sigma(H_{\La_L})\setminus\sigma_\cG (H_{\La_L}),
\eeq
such that for  $\nu\in\sigma_\cB(H_{{\Upsilon}})$ we have 
\begin{align}\label{tildedist}
\abs{\wtilde{\nu} -\nu} \le \e^{-m_4{\ell_\tau}}, \qtx{where}m_4=m_4(\ell) \ge  m\pa{1- C_{d,m_-,\eps_0}\ell^{\gamma \beta -\tau}},
\end{align}
and, multiplying each ${\psi_\nu}$ by a suitable phase factor as in \eq{aligning},  
\beq \label{difeqr}
\norm{\phi_{\wtilde{\nu}}- \psi_\nu}\le 2 \e^{-m_4 {\ell_\tau}}\e^{L^\beta}.
\eeq
\end{enumerate}
 \end{lemma}

\begin{proof} Part (i) follows immediately from Lemma~\ref{lem:ident_eigensyst}(ii)(c) and (ii)(d).

Now let $ \La_L$ be level spacing for $H$, and let $\set{(\phi_\lambda,\lambda)}_{\lambda \in \sigma(H_{\La_L})}$ be an eigensystem  for $H_{\La_L}$.  It follows from \eq{distspt} in Lemma~\ref{lemdecayvphi} that for $ \nu\in\sigma_\cB(H_{{\Upsilon}})$ we have
%% CORRECTION
 \beq\label{appeig47549}
 \norm{\pa{H_{\La_L} -\nu}{\psi_\nu} }\le (2d-1) \eps \abs{ \partialex^{ \La_L}{\Ups}}^{\frac 12}\norm{\psi_\nu\Chi_{ \partialin^{ \La_L}{\Ups}}}_\infty\le  (2d-1) \eps L^{\frac {d}2} \e^{-  m_2{\ell_\tau}}\le  \e^{-m_4\ell_\tau},
 \eeq
where we used 
$\partialin^{ \La_L}{\Ups} \subset  \widecheck{\Upsilon}_{\tau}$ and \eq{vthetasmall}, and $m_4$ is given in \eq{tildedist}.
Since $ \La_L$  and $\Ups$ are $L$-level spacing for $H$, the map in \eq{injbad} is a well defined injection into $\sigma (H_{\La_L})$, and   \eq{difeqr} follows from \eq{tildedist} and \eq{deltaeta}. 
To finish the proof we must show that $\wtilde\nu\notin \sigma_\cG(H_{\La_L})$ for all $\nu\in\sigma_\cB(H_{{\Upsilon}})$. 

Suppose  $\wtilde{\nu_1}\in  \sigma_\cG(H_{\La_L})$ for some $\nu_1\in\sigma_\cB(H_{{\Upsilon}}) $. Then there is $a\in \cG$ and $x\in \La_\ell^{\La_L,\ell_{\tau}}(a)$ such that
${\lambda\up{a}_x} \in\cE_\cG^{\La_L}(\wtilde{\nu_1})$.  On the other hand, it follows from   Lemma~\ref{lem:ident_eigensyst}(i)(a) that ${\lambda\up{a}_x} \in\cE_\cG^{\Ups}(\lambda_1) $ for some $\lambda_1 \in \sigma_\cG(H_{\Ups})$.  We conclude  from \eq{difeq86} and \eq{difeqr}  that
\begin{align}
\sqrt{2}= \norm{\psi_{\lambda_1} -\psi_{\nu_1}}&\le 
\norm{\psi_{\lambda_1}-{\vphi\up{a}_x}  }+\norm{{\vphi\up{a}_x} -{\phi}_{\wtilde{\nu_1}} }+ \norm{{\phi}_{\wtilde{\nu_1}} -\psi_{\nu_1} }\\
\notag & \le {4}\e^{-m_1{\ell_\tau}}\e^{L^\beta}+  {2}\e^{-m_4{ \ell_\tau}}\e^{L^\beta}< 1,
\end{align}
 a contradiction.
\end{proof}

 \begin{lemma}\label{lembad}   Let $\La_L=\La_L(x_0)$, $x_0 \in \R^d$, $m\ge m_-$.   
Let $ {\Upsilon}$ be    a buffered subset of  $\La_L$.  
 Let  $\cG=\cG_\Ups$  and set
\begin{align}\label{defcEUps}
\cE_\cG^{\La_L} (\nu)&=  \set{{\lambda\up{a}_x}; \; a\in \cG,\, x\in \La_\ell^{\La_L,\ell_{\tau}}(a), \sqtx{and} \wtilde{\lambda}\up{a}_x=\nu}\subset \cE_\cG^{\Ups} (\nu)\sqtx{for} \nu\in \sigma(H_{\Upsilon}), \notag \\ 
\sigma_\cG^{\La_L} (H_{{\Upsilon}})&=\set{\nu \in \sigma(H_{{\Upsilon}}); \; \cE_\cG^{\La_L} (\nu)\ne \emptyset} \subset \sigma_\cG (H_{{\Upsilon}}).
\end{align}
The following holds for sufficiently large $L$:
\begin{enumerate}
\item Let  $(\psi,\lambda)$ be an  eigenpair for $H_{\La_L }$ such that
\beq\label{distpointeig1}
\abs{\lambda - \nu}
\ge \tfrac 1 2\e^{-L^\beta}\qtx{for all} \nu \in\sigma_\cG^{\La_L} (H_{{\Upsilon}})\cup \sigma_\cB  (H_{{\Upsilon}}).
\eeq
Then for all $y\in \Ups^{{\La_L},2{\ell_\tau} }$  we have
\beq\label{gggsum}
\abs{\psi(y)}\le    \e^{-m_5{\ell_\tau}}\abs{\psi(v)}\qtx{for some} v\in{\partial}^{\La_L, 2{\ell_\tau} }  {\Upsilon},
\eeq
where 
\beq\label{M21}
m_5=m_5(\ell)\ge   m\pa{1-   C_{d,m_-,\eps_0} \ell^{\gamma \beta -\tau}}.
\eeq
\item Let $\La_L$ be  level spacing for $H$,  let $\set{(\psi_\lambda,\lambda)}_{\lambda \in \sigma(H_{\La_L})}$ be an eigensystem  for $H_{\La_L}$, recall \eq{injbad},  and set
 \beq
\sigma_\Ups (H_{\La_L})= \set{\wtilde\nu;  \nu\in\sigma_\cB(H_{{\Upsilon}})}\subset  \sigma(H_{\La_L})\setminus\sigma_{\cG} (H_{\La_L}).
\eeq
  Then for all
   \[
\lambda \in \sigma (H_{{\La_L}})\setminus\pa{  \sigma_\cG (H_{{\La_L}})\cup  \sigma_\Ups (H_{\La_L})},
\]  
the condition \eq{distpointeig1} is satisfied, and $\psi_\lambda$  satisfies  \eq{gggsum}.
\end{enumerate}
  \end{lemma}

\begin{proof}  For $a\in \cG$ we have $\La_\ell(a)\subset \Ups\subset \La_L$, which implies
$\La_\ell^{\La_L,\ell_{\tau}}(a)\subset \La_\ell^{\Ups,\ell_{\tau}}(a)$.  Thus $\cE_\cG^{\La_L} (\nu)\subset \cE_\cG^{\Ups} (\nu)$ for $\nu\in \sigma(H_{\Upsilon})$.

 Let  $\set{(\vtheta_\nu,\nu)}_{\nu \in \sigma(H_{\Upsilon})}$ be an eigensystem  for $H_{\Upsilon}$.    
For each  $\nu \in \sigma_\cG (H_{\Upsilon})$ we fix $\lambda\up{a_\nu}_{x_\nu}  \in \cE_\cG^\Ups (\nu)$,  where 
  $a_\nu\in \cG,\, x_\nu\in \La_\ell^{{\Upsilon},\ell_{\tau}}(a_\nu)$,  picking  $\lambda\up{a_\nu}_{x_\nu}  \in \cE_\cG^{\La_L} (\nu)$ if $\nu \in \sigma_\cG^{\La_L} (H_{\Upsilon})$, so   $ x_\nu\in \La_\ell^{{\La_L},\ell_{\tau}}(a_\nu)$.  If  $\nu \in\sigma_\cG (H_{{\Upsilon}})\setminus \sigma_\cG^{\La_L}(H_{{\Upsilon}})$ we have  $ x_\nu\in \La_\ell^{{\Upsilon},\ell_{\tau}}(a_\nu)\setminus \La_\ell^{{\La_L},\ell_{\tau}}(a_\nu)$.

Let $y\in{\Upsilon}$.  Using  \eq{psiy} we    have 
\begin{align}\label{sum000227}
\psi(y) &=\scal{\delta_y,\psi}  =   \sum_{\nu \in\sigma(H_ {\Upsilon})}{{\vtheta_\nu}(y)} \scal{ {\vtheta_\nu},\psi}\\   \notag  &  =  \sum_{\nu \in \sigma_\cG^{\La_L} (H_{{\Upsilon}})\cup \sigma_\cB  (H_{{\Upsilon}})}{{\vtheta_\nu}(y)} \scal{ {\vtheta_\nu},\psi} +   \sum_{\nu \in\sigma_\cG (H_{{\Upsilon}})\setminus \sigma_\cG^{\La_L} (H_{{\Upsilon}})}{{\vtheta_\nu}(y)} \scal{ {\vtheta_\nu},\psi}.
\end{align}

 Let $(\psi,\lambda)$  be an  eigenpair for $H_{\La_L }$ such that \eq{distpointeig1} holds.
Given  $\nu \in\sigma_\cG^{\La_L} (H_{{\Upsilon}})\cup \sigma_\cB  (H_{{\Upsilon}})$,  we have
\beq
\scal{ {\vtheta_\nu},\psi}= \pa{\lambda- \nu }^{-1}\scal{ {\vtheta_\nu},  \pa{H_{{\La_L}}- \nu}\psi}= \pa{\lambda- \nu }^{-1}\scal{\pa{H_{{\La_L}}- \nu} {\vtheta_\nu}, \psi} .
\eeq
It follows from  \eq{distpointeig1} and  \eq{HTheta} that
\begin{align}\label{vphipsi39909}
\abs{{{\vtheta_\nu}(y)} \scal{ {\vtheta_\nu},\psi}}&\le 2\e^{L^\beta} \eps  \abs{ \vtheta_\nu(y)} \sum_{  v \in    \partialex^{{{\La_L}}}  {\Upsilon}} \pa{ \textstyle\sum_{\substack{v^\pr \in \partial_{\mathrm{in}}^{ \La_L}{\Ups} \\ \abs{v^\pr-v}=1}}\abs{ \vtheta_\nu({v^\pr})}} \abs{\psi(v)}\\
& \notag \le  2\eps L^{d}\e^{L^\beta} \set{2d \max_{  u \in   \partialin^{{{\La_L}}}  {\Upsilon}} {\abs{\vtheta_\nu(u)}} }  \abs{\psi(v_1)} \qtx{for some} v_1 \in\partialex^{{{\La_L}}} {\Upsilon}.
\end{align}
If $\nu \in \sigma_\cB  (H_{{\Upsilon}})$ it follows from \eq{vthetasmall} that 
\beq\max_{  u \in    \partialin^{{\La_L}}  {\Upsilon}} {\abs{\vtheta_\nu(u)}}   \le  e^{-m_2{\ell_\tau}}.
\eeq
If  $\nu \in \sigma_\cG^{\La_L} (H_{{\Upsilon}})$, it follows from  \eq{difeq86} and   \eq{hypdec} that
\begin{align}
&\max_{  u \in   \partialin^{{{\La_L}}}  {\Upsilon}} {\abs{\vtheta_\nu(u)}}   \le 
\max_{  u \in   \partialin^{{{\La_L}}}  {\Upsilon}} \pa{{\abs{\vtheta_\nu(u)-\vphi\up{a_\nu}_{x_\nu} (u)}}+\abs{\vphi\up{a_\nu}_{x_\nu} (u)}}\\ 
&  \qquad  \notag \le {2}\e^{-m_1{\ell^\tau}}\e^{L^\beta} +  e^{-m{  {\ell_\tau}}}\le{3}\e^{-m_1{\ell_\tau}}\e^{L^\beta},
\end{align}
recalling \eq{m414} and \eq{m41455}. 
It follows that (note $m_1(\ell) \ge m_2(\ell)$ for $\ell$ large)
\begin{align}\label{gsum53} 
\abs{ \sum_{\nu \in\sigma_\cG^{\La_L} (H_{{\Upsilon}})\cup \sigma_\cB  (H_{{\Upsilon}})}{{\vtheta_\nu}(y)} \scal{ {\vtheta_\nu},\psi}}&\le 
4d \eps L^{2d}\e^{L^\beta}   \pa{{3}\e^{-m_2{\ell_\tau}}\e^{L^\beta}}    \abs{\psi(v_2)} \\ \notag   & \le 12 d \eps L^{2d} \e^{2L^\beta}\e^{-m_2{\ell_\tau}} \abs{\psi(v_2)},
 \end{align}
for some  $ v_2 \in \partialex^{{\La_L}} {\Upsilon}$.

Now let $\nu \in\sigma_\cG  (H_{{\Upsilon}})\setminus \sigma_\cG^{\La_L} (H_{{\Upsilon}})$.  In this case  we have $ x_\nu\in \La_\ell^{{\Upsilon},\ell_{\tau}}(a_\nu)\setminus \La_\ell^{{\La_L},\ell_{\tau}}(a_\nu)$, so we have 
\beq
  \dist \pa{x_\nu, \Ups \setminus \La_\ell(a_\nu)} > \ell_\tau \qtx{and}\dist \pa{x_\nu, {\La_L}\setminus \La_\ell(a_\nu)} \le \ell_\tau,
 \eeq
 so there is $u_0 \in {\La_L} \setminus \Ups$ such that $\norm{x_\nu- u_0  }\le  {\ell_\tau}$.
We now assume $y\in  \Ups^{{\La_L},2{\ell_\tau} }$, so we have  $\norm{y- u_0  }>2 \ell_\tau$.  We conclude that
\beq
\abs{x_\nu-y}\ge  \norm{y- u_0  }- \norm{x_\nu- u_0  } > 2\ell_\tau-\ell_\tau=\ell_\tau.
\eeq
Thus
\beq\label{vthetabadg}
\abs{{\vtheta_\nu}(y)}\le \abs{{\vtheta_\nu}(y) - \vphi\up{a_\nu}_{x_\nu}(y)}+  \abs{ \vphi\up{a_\nu}_{x_\nu}(y)}\le {2}\e^{-m_1\ell_\tau}\e^{L^\beta}+ \e^{-   m \ell_\tau}\le 3\e^{-m_1\ell_\tau}\e^{L^\beta},
 \eeq 
using  \eq{difeq86} and  \eq{hypdec}. It follows that
\beq\label{bsum40}
\abs{\sum_{\nu \in\sigma_\cG (H_{{\Upsilon}})\setminus \sigma_\cG^{\La_L}(H_{{\Upsilon}})}{{\vtheta_\nu}(y)} \scal{ {\vtheta_\nu},\psi}}
\le3 (L+1)^{\frac {3d}2}\e^{-m_1\ell_\tau}\e^{L^\beta}\abs{\psi(v_3)},
\eeq

for some  $v_3 \in {\Upsilon}$.

Combining \eq{sum000227}, \eq{gsum53} and \eq{bsum40}, we get for $y\in  \Ups^{{\La_L},2\fl{\ell^\tau} }$ that 
\beq
\abs{\psi(y)}\le (1 + 12 d \eps) L^{2d}\e^{2L^\beta}\e^{-m_2 {\ell_\tau}}\abs{\psi(v_4)}\le  \e^{-m_5  \ell_\tau} \abs{\psi(v_4)},
\eeq
for some $v_4 \in {\Upsilon} \cup\partialex^{ {\La_L}}  {\Upsilon}$, where $m_5$ is given in \eq{M21}.
If $v_4\in \Ups^{{\La_L},2{\ell_\tau} }$
 we can repeat the procedure to estimate  $\abs{\psi(v_4)}$.  If $\psi(y)= 0$ there is nothing to prove, so  we can assume $\psi(y)\ne 0$.  In this case we can only repeat the procedure a finite number of times without getting $\abs{\psi(y)}<\abs{\psi(y)}$, so  \eq{gggsum} holds.
 
Now  suppose $\La_L$ is level spacing for $H$. If $\lambda \notin   \sigma_\cG (H_{{\La_L}})$, it follows from Lemma \ref{lem:ident_eigensyst}(i)(c) that \eq{distpointeiga} holds for all $a \in\cG$. If $\lambda \notin   \sigma_\Ups (H_{\La_L})$, the  argument in \eq{distpointsigeigab}, modified by the use of  \eq{tildedist} instead of  \eq{tildedist132}, gives $\abs{\lambda -\nu}\ge  \tfrac 12 \e^{-\L^\beta}$ for all $ \nu\in\sigma_\cB(H_{{\Upsilon}})$.  Thus  we have \eq{distpointeig1}, which implies \eq{gggsum}.
 \end{proof}

\subsection{Suitable covers of a box}
To perform the multiscale analysis in an efficient way, it is convenient to use  a canonical  way to cover a box of side $L$ by boxes of side $\ell <L$.  We will use
  suitable covers of a box as in \cite[Definition~3.12]{GKber}, adapted to the discrete case.

\begin{definition}\label{defcov} Let $\La_L=\La_L(x_0)$, $x_0 \in \R^d$ be  a box in $\Z^d$, and let $\ell < L$.
A suitable $\ell$-cover of $\La_L$ is
the collection of  boxes  
\begin{align}\label{standardcover}
{\mathcal C}_{L,\ell} \left(x_0 \right)= \set{ {\Lambda}_{\ell}(a)}_{a \in  \Xi_{L,\ell}},
\end{align}
where
\beq  \label{bbG}
 \Xi_{L,\ell}:= \set{ x_0+ {\rho}\ell  \Z^{d}}\cap \La_L^\R\quad 
\text{with}  \quad {\rho}\in  \br{\tfrac {3} {5},\tfrac {4} {5}}   \cap \set{\tfrac {L-\ell}{2 \ell k}; \, k \in \N }.
\eeq
We call ${\mathcal C}_{L,\ell} \left(x_0 \right)$ the suitable $\ell$-cover of $\La_L$  if   ${\rho} ={\rho}_{L,\ell}: =\max \br{\tfrac {3} {5},\tfrac {4} {5}}   \cap \set{\tfrac {L-\ell}{2 \ell k}; \, k \in \N }.
$
\end{definition}

We recall \cite[Lemma~3.13]{GKber}, which we rewrite in our context.

\begin{lemma}\label{lemcover } Let $\ell \le \frac  L   6$. Then for  every  box   $\La_L=\La_L(x_0)$, $x_0 \in \R^d$, a suitable $\ell$-cover
  ${\mathcal C}_{L,\ell}\left(x_0\right) $  satisfies
  \begin{align}\label{nestingproperty} 
&\La_L=\bigcup_{a \in  \Xi_{L,\ell}} {\Lambda}_{\ell}(a);\\ \label{covproperty}
&\text{for all}\; \; b \in\La_L  \sqtx{there is} {\Lambda}_{\ell}^{(b)} \in {\mathcal C}_{L,\ell} \left(x_0 \right) \sqtx{such that}  
  b\in \pa{ {\Lambda}_{\ell}^{(b)}}^{\La_L, \frac \ell {10}},\\ \notag
 & \qtx{i.e.,} \La_L=\bigcup_{a \in  \Xi_{L,\ell}} {\Lambda}_{\ell}^{\La_L, \frac \ell {10}}(a);
    \\
\label{freeguarantee}
&\La_{\frac{\ell}{5}}(a)\cap \La_{\ell}(b)=\emptyset
\quad \text{for all} \quad a, b\in x_0 + {\rho}\ell  \Z^{d}, \ a\ne b;\\ \label{number}
&  \pa{\tfrac{L} {\ell}}^{d}\le \# \Xi_{L,\ell}= \pa{ \tfrac{L-\ell} {{\rho} \ell}+1}^{d }\le   \pa{\tfrac{2L} {\ell}}^{d}.  
\end{align}
Moreover,  given $a \in x_0 + {\rho}\ell  \Z^{d}$ and $k \in \N$, it follows that
\beq \label{nesting}
{\Lambda}_{(2  k {\rho}  +1) \ell}(a)= \bigcup_{b \in  \{ x_0 + {\rho}\ell  \Z^{d}\}\cap {\Lambda}^{\R}_{(2k {\rho}  +1) \ell}(a) } {\Lambda}_{\ell}(b),
\eeq
and  $ \{ \Lambda_{\ell}(b)\}_{b \in  \{ x_0 + {\rho}\ell  \Z^{d}\}\cap {\Lambda}^{\R}_{(2k {\rho}  +1) \ell}(a) }$ is a suitable $\ell$-cover of the box $\Lambda_{(2 k{\rho} +1) \ell}(a)$. 
\end{lemma}

Note that $\Lambda_{\ell}^{(b)}$  does not denote a box  centered at $b$, just some box in ${\mathcal C}_{L,\ell} \left(x_0 \right)$ satisfying \eq{covproperty}.     By  $\Lambda_{\ell}^{(b)}$ we will always mean such a box.

\begin{remark}    
Note that  ${\rho} \ge \frac 3 5$ implies \eq{freeguarantee} and ${\rho} \le \frac 4 5$ yields \eq{covproperty}.   (We do not use \eq{freeguarantee} in this paper.) We specified ${\rho}={\rho}_{L,\ell}$ in   \emph{the} suitable $\ell$-cover for convenience, so there is no ambiguity  in the definition of ${\mathcal C}_{L,\ell} \left(x_0 \right) $.
\end{remark}

\begin{remark} \label{bufferedsuitable} Suitable covers are convenient for the construction of buffered subsets. 
   \end{remark}

\section{Eigensystem multiscale analysis}\label{secEMSA}
In this section we consider an   Anderson model $H_{\eps,\bom}$ and prove Theorem~\ref{thmMSA} as a corollary  to the following proposition.
We  recall that   $m_{\eps,L} $ is defined in \eq{mepsL}.

 \begin{proposition} \label{thmfullMSA} There exists a finite scale   $\cL $ such that, given $L_0\ge \cL $ and setting
 $L_{k+1}=L_k^\gamma$ for $k=0,1,\ldots$,
for all  $ \eps\le  \frac 1 {4d} \e^{-L_0^\beta}$ we have
 \begin{align}\label{MSAlk}
\inf_{x\in \R^d} \P\set{\La_{L_k} (x) \sqtx{is}   \tfrac {m_{\eps,L_0} }2 \text{-localizing for} \; H_{\eps,\bom}} \ge 1 -  \e^{-L_k^\zeta} \mqtx{for} k=0,1,\ldots.
\end{align}
  \end{proposition}

{Proposition~\ref{thmfullMSA}  is an immediate consequence of the following two propositions.

\begin{proposition}\label{propInitial} Let $\eps>0$ and $L\ge 1$. 
 Then 
\beq\label{probmloc}
\inf_{x\in \R^d} \P\set{\La_L(x) \sqtx{is} m_{\eps,L}\text{-localizing for} \; H_{\eps,\bom}} \ge 1 -  \tfrac 1 2 K {(L+1)^{2d}} \pa{8d\eps +2\e^{-L^\beta} }^\alpha.
\eeq
In particular,
if   $L$ is sufficiently large,  for all  $0< \eps\le  \frac 1 {4d} \e^{-L^\beta}$  we have $m_{\eps,L}\ge \log 3$ and 
\begin{align}\label{probmloc2}
\inf_{x\in \R^d} \P\set{\La_L(x) \sqtx{is} m_{\eps,L}\text{-localizing for} \; H_{\eps,\bom}}
 \ge 1 -  \e^{-L^\zeta}.
\end{align}
\end{proposition}

 \begin{proposition}\label{propMSA}  Fix $\eps_0 >0$ and  $m_->0$.   There exists a finite scale   $\cL(\eps_0,m_-) $ with the following property: Suppose for some scale $L_0\ge \cL(\eps_0,m_-) $, $0<\eps\le \eps_0$, and $m_0\ge m_-$ we have
  \begin{align}\label{initialconinduc}
\inf_{x\in \R^d} \P\set{\La_{L_0} (x) \sqtx{is}  m_0 \text{-localizing for} \; H_{\eps,\bom}} \ge 1 -  \e^{-L_0^\zeta}.
\end{align}
 Then,  setting
 $L_{k+1}=L_k^\gamma$  for $k=0,1,\ldots$,  we have 
  \begin{align}
\inf_{x\in \R^d} \P\set{\La_{L_k} (x) \sqtx{is}  \tfrac {m_0}2   \text{-localizing for} \; H_{\eps,\bom}} \ge 1 -  \e^{-L_k^\zeta} \mqtx{for} k=0,1,\ldots.
\end{align}
 \end{proposition}

\subsection{Initial step}

In this subsection we  prove  Proposition~\ref{propInitial}.

\begin{lemma}\label{lemInitial}
Let  $H_\eps=-\eps\Delta +V$ on $\ell^2(\Z^d)$, where $V$ is a  bounded potential and $\eps > 0$.  
Let $\Theta \subset \Z^d$, and suppose there is $\eta >0$ such that 
\beq\label{Vsep}
\abs{V(x)-V(y)}\ge \eta \qtx{for all} x, y \in\Th, x\ne y.
\eeq
Then for   $\eps < \frac \eta {4d}$ the operator  $H_{\eps,\Th}$ has an eigensystem  $\set{(\psi_x,\lambda_x)}_{x\in \Th}$ such that
\beq\label{Hthsep}
\abs{\lambda_x-\lambda_y}\ge \eta- 4d\eps >0 \qtx{for all} x, y \in\Th, x\ne y,
\eeq
and for all $y\in \Th$ we have 
\beq\label{initialdec}
\abs{\psi_y(x)} \le  \pa{ \tfrac {2d\eps }{ \eta- 2d\eps}}^{\abs{x-y}_1} \qtx{for all } x \in \Th.
\eeq
\end{lemma}

\begin{proof}We take $\eps <  \frac \eta {4d}$ and 
treat $H_{\eps,\Th}$ as a perturbation of $V_\Th$. Since $\sigma(V_\Th)=\set{V(x)}_{x \in \Th}$ is simple and $\norm{ \Delta_\Th}\le 2d $, it follows from  \eq{Vsep} and Weyl's inequality (e.g.,  \cite[Theorem~4.3.1]{HJ}), that  $H_{\eps,\Th}$ has simple spectrum   $\sigma(H_{\eps,\Th})=\set{\lambda_x}_{x \in \Th}$ with
\beq\label{lambda-V}
\abs{\lambda_x-V(x)}\le 2d\eps <  \tfrac \eta {2} \qtx{for all} x\in \Th,
\eeq
so we have  \eq{Hthsep} and $H_{\eps,\Th}$ has an  eigensystem $\set{(\psi_x,\lambda_x)}_{x\in \Th}$. 

 Let  $y\in \Th$. Then for any $x\in \Th$, $x\neq y$, we have, 
  \beq
\abs{\lambda_y-V(x)}\ge \abs{V(y)-V(x)}-\abs{\lambda_y-V(y)}\ge\eta- 2d\eps,
\eeq
where we used  \eq{Vsep}  and \eq{lambda-V}, 
and
\begin{align}\label{psicalc}
\psi_y(x)& =\scal{\delta_x,\psi_y} =\pa{\lambda_y-V(x)}^{-1}\scal{(H_{\eps,\Th}-V_\Th)\delta_x,\psi_y}\\  \notag &=\eps\pa{\lambda_y-V(x)}^{-1}\scal{-\Delta_\Th\delta_x,\psi_y}=\eps\pa{\lambda_y-V(x)}^{-1}  \sum_{\substack{z\in\Th\\ |z-x|=1}} {\psi_y(z)}.
\end{align}
We conclude  that
\beq\label{abspsi}
\abs{\psi_y(x)} \leq \tfrac{\eps}{\eta- 2d\eps}\sum_{\substack{z\in\Th\\ |z-x|=1}}\abs{\psi_y(z)}\le  \tfrac{2d\eps}{\eta- 2d\eps} \abs{\psi_y(z_1)} ,
\eeq
for some  $z_1\in\Th$ with  with  $ |z_1-x|=1$.  If  $z_1\ne y$ we can estimate  $\abs{\psi_y(z_1)}$ by  
\eq{abspsi}.  Since we can  perform  this procedure at least $\abs{
x-y}_1$ times, we obtain \eq{initialdec}.
\end{proof}

\begin{proof}[Proof of Proposition~\ref{propInitial}]  Let $\eps >0$ and
 $\La_L= \La_L(x_0)$ for some $x_0 \in \R^d$. Let $\kappa = 4d\eps  \e^{L^\beta}$ and suppose
\beq\label{VsepL}
\abs{V(x)-V(y)}\ge (1+\kappa)\e^{-L^\beta} \qtx{for all} x, y \in\La_L, x\ne y.
\eeq
It  follows from 
 Lemma~\ref{lemInitial}  that $H_{\eps,\La_L}$ has an eigensystem  $\set{(\psi_x,\lambda_x)}_{x\in \La_L}$ satisfying  \eq{Hthsep} and   \eq{initialdec} with $\eta=   (1+\kappa)\e^{-L^\beta}$. Since
 $\eta - 4d\eps =\e^{-L^\beta}$, we conclude from \eq{Hthsep} that $\La_L$ is level spacing for $H_{\eps}$.  Moreover, $\frac {2d\eps }{ \eta- 2d\eps}= \frac {\kappa}{2+\kappa } $ and 
 and $\norm{x}\le \abs{x}_1$, so  \eq{initialdec}  yields
\beq\label{initialdec9}
\abs{\psi_y(x)} \le\pa{ \tfrac {\kappa}{2+\kappa } }^{\norm{x-y}}= \e^{-m_{\eps,L}\norm{x-y}}\qtx{for all} y, x \in \La_L,
\eeq
where 
\beq
m_{\eps,L}= -\log \pa{ \tfrac {\kappa}{2+\kappa } }=  \log \tfrac { \eta- 2d\eps}{2d\eps }= \log \pa{1 + \tfrac {\e^{-L^\beta}}{2d\eps} }.  
\eeq
In particular, $\La_L$ is $m_{\eps,L}$-localizing.

We conclude that
\begin{align}
&\P\set{\La_L \sqtx{is not} m_{\eps,L}\text{-localizing}}\le \P\set{\text{\eq{VsepL} does not hold}}
\\
\notag & \qquad \qquad  \le \tfrac {(L+1)^{2d}}2 S_\mu \pa{2(1+\kappa)\e^{-L^\beta} }= \tfrac {(L+1)^{2d}}2 S_\mu \pa{8d\eps +2\e^{-L^\beta} }\\
\notag & \qquad \qquad \le  \tfrac 1 2 K {(L+1)^{2d}} \pa{8d\eps +2\e^{-L^\beta} }^\alpha,
\end{align}
which yields \eq{probmloc}.   (We assumed $8d\eps +2\e^{-L^\beta}\le 1$ to use \eq{Holdercont} as stated; if not \eq{probmloc} holds trivially.)

If $0< \eps\le  \frac 1 {4d} \e^{-L^\beta}$,  we have  $m_{\eps,L}\ge \log 3$ and
\begin{align}
\inf_{x\in \R^d} \P\set{\La_L(x) \sqtx{is} m_{\eps,L}\text{-localizing for} \; H_\eps}
\ge
1- 2^{2\alpha -1} K {(L+1)^{2d}} \e^{-\alpha L^\beta} ,
\end{align}
which gives \eq{probmloc2} for large $L$ since $\zeta < \beta$.
\end{proof}

\subsection{Multiscale analysis}

In this subsection we prove  Proposition~\ref{propMSA}.  We start with the induction step for the multiscale analysis.

\begin{lemma}\label{lemInduction}
 Fix $\eps_0 >0$ and  $m_->0$.   Suppose for some scale $\ell$, $0<\eps\le \eps_0$, and $m\ge m_-$ we have
 \begin{align}\label{hypMSAlem}
\inf_{x\in \R^d} \P\set{\La_{\ell} (x) \sqtx{is}  m \text{-localizing for} \; H_{\eps,\bom}} \ge 1 -  \e^{-\ell^\zeta}.
\end{align}
Then, if $\ell$ is sufficiently large, we have (recall $L=\ell^\gamma$)
 \begin{align}
\inf_{x\in \R^d} \P\set{\La_{L} (x) \sqtx{is}   M   \text{-localizing for} \; H_{\eps,\bom}} \ge 1 -  \e^{-L^\zeta},
\end{align}
where
\beq\label{Minduc}
M\ge m\pa{1-   C_{d,m_-,\eps_0} \ell^{-\min\set{\frac {1- \tau }2, \gamma \tau- (\gamma-1)\tzeta  -1}}}.
\eeq
\end{lemma}

\begin{proof} We fix $0<\eps\le \eps_0$, and $m\ge m_-$  and assume \eq{hypMSAlem} for some scale $\ell$.  We take
$\La_L=\La(x_0)$, where $x_0\in \R^d$,  and let
${\mathcal C}_{L,\ell}={\mathcal C}_{L,\ell} \left(x_0 \right)$ be the suitable $\ell$-cover of $\La_L$.  Given $N \in \N$,
let  $\cB_N$ denote the event that there exist at most $N$ disjoint boxes in ${\mathcal C}_{L,\ell}$ that are not $m$-localizing for  $ H_{\eps,\bom}$. We have, using \eq{number},  \eq{hypMSAlem}, and the fact that events on disjoint boxes are independent, that
\begin{align}\label{probBN}
\P\set{\cB_N^c}\le \pa{\tfrac{2L} {\ell}}^{(N+1)d} \e^{-(N+1)\ell^\zeta}= 2^{(N+1)d} \ell^{(\gamma-1)(N+1)d}\e^{-(N+1)\ell^\zeta} < \tfrac 12 \e^{-L^\zeta},
\end{align}
 if $N+1 > \ell^{(\gamma-1)\zeta}$ and $\ell$ is sufficiently large.
  For this reason we take (recall \eq{ttauzeta2})
\beq\label{setN}
N= N_\ell= \fl{\ell^{(\gamma-1)\tzeta}} \quad \Longrightarrow  \quad \P\set{\cB_{N_\ell}^c}\le  \tfrac 12  \e^{-L^\zeta} \quad\text {for all} \; \ell \; \text{sufficiently large}.
\eeq

We  now fix   $\bom \in\cB_N$.  There exist $\cA_N=\cA_N(\bom)\subset  \Xi_{L,\ell}=\Xi_{L,\ell} \left(x_0 \right)$, with $\abs{\cA_N}\le N$ and  $\norm{a -b}\ge 2{\rho} \ell$  (i.e., $\La_\ell(a)\cap \La_\ell(b)=\emptyset$) if $a,b \in \cA_N$ and $a\ne b$, such that for all $a\in \Xi_{L,\ell} $ with
$\dist (a,\cA_N)\ge  2{\rho} \ell$ (i.e.,  $\La_\ell(a)\cap \La_\ell(b)=\emptyset$ for all $b\in \cA_N$) the box $\La_\ell(a)$ is $m$-localizing for  $ H_{\eps,\bom}$.
 In other words,
\beq\label{implyloc}
a \in  \Xi_{L,\ell} \setminus \bigcup_{b\in \cA_N}\La_{(2{\rho} +1)\ell}(b) \quad \Longrightarrow  \quad \La_\ell(a) \qtx{is} m\text{-localizing for} \quad  H_{\eps,\bom}.
\eeq

We  want to  embed the   boxes $\set{\La_\ell(b) }_{b \in \cA_{N}}$  into  buffered subsets of $\La_L$. To do so, we consider  graphs  $\G_i=\pa{ \Xi_{L,\ell}, \E_i}$, $i=1,2$,   
both having   $ \Xi_{L,\ell}$ as the set of vertices, with sets of edges given by
\begin{align}
 \E_1&=\set{\set{a,b}\in  \Xi_{L,\ell}^2;\;  \norm{a-b}={\rho}\ell }\\  \notag
 & =\set{\set{a,b}\in  \Xi_{L,\ell}^2;\;  a\ne b \sqtx{and} \La_\ell(a)\cap \La_\ell(b)\ne \emptyset },\\
 \notag
 \E_2&=\set{\set{a,b}\in  \Xi_{L,\ell}^2;\mqtx{either} \norm{a-b}=2{\rho}\ell \mqtx{or} \norm{a-b}=3{\rho} \ell }\\
 \notag & =\set{\set{a,b}\in  \Xi_{L,\ell}^2;\; \La_\ell(a)\cap \La_\ell(b)=\emptyset\sqtx{and} \La_{(2{\rho} +1)\ell}(a)\cap \La_{(2{\rho}  +1)\ell}(b)\ne \emptyset}.
 \end{align}

 Given $\Psi \subset \Xi_{L,\ell} $, we define the exterior boundary of $\Psi$ in the graph $\G_1 $ by
\beq
 \partial_{\mathrm{ex}}^{ \G_1} \Psi= \set{a\in \Xi_{L,\ell}; \ \dist (a,\Psi)={\rho} \ell}.
 \eeq
(This is similar, but not the same as the definition in \eq{defbdry}.) We  let $\overline{\Psi}= \Psi \cup \partial_{\mathrm{ex}}^{ \G_1} \Psi$. 

 Let  $\Phi\subset  \Xi_{L,\ell}$ be $\G_2$-connected, so $\diam \Phi \le 3\rho \ell \pa{\abs{\Phi}-1}$. We set
 \beq
\wtilde{\Phi}= \Xi_{L,\ell} \cap  \bigcup_{a\in \Phi} \La^\R_{(2{\rho} +1)\ell}(a)= \set{a\in \Xi_{L,\ell}; \; \dist (a,\Phi)\le {\rho} \ell  },
\eeq 
 note that  $\wtilde{\Phi}$ is a  $\G_1$-connected subset of $ \Xi_{L,\ell}$ such that
 \beq\label{diamwphi}
 \diam \wtilde{\Phi} \le  \diam {\Phi} +2\rho \ell  \le  \rho \ell \pa{3\abs{\Phi}-1}, 
 \eeq
and let
\begin{align}\label{constUps}
\Ups_\Phi\up{0}  = \bigcup_{a \in \wtilde{\Phi}}  \La_\ell (a) \qtx{and}
\Ups_\Phi  = \Ups_\Phi \up{0} \cup 
\bigcup_{a \in \partial_{\mathrm{ex}}^{ \G_1}\wtilde{\Phi}} \La_\ell(a)       =     \bigcup_{a \in \overline{\wtilde{\Phi}}}\La_\ell(a).
\end{align}

Now  let $\set{\Phi_r}_{r=1}^R=\set{\Phi_r(\bom)}_{r=1}^R$ denote the $\G_2$-connected components of  $\cA_{N}$ (i.e., connected in the graph $\G_2$). Note  that
\beq
R \in \set{1,2,\ldots,N}\qtx{and}  \quad \sum_{r=1}^R\abs{\Phi_r}=\abs{\cA_N}\le N.
\eeq 
Note also that  $\set{\wtilde{\Phi}_r}_{r=1}^R$ is a collection of disjoint, $\G_1$-connected subsets of $ \Xi_{L,\ell}$, such that
 \beq
 \ \dist (\wtilde{\Phi}_r,\wtilde{\Phi}_s)\ge  2{\rho} \ell \qtx{for}    r\ne s.
 \eeq  
 
It follows from \eq{implyloc} that
\beq
\label{implyloc2}
a \in \cG=\cG(\bom)=  \Xi_{L,\ell} \setminus \bigcup_{r=1}^R\wtilde{\Phi}_r\quad \Longrightarrow  \quad \La_\ell(a) \qtx{is} m\text{-localizing for} \quad  H_{\eps,\bom}.
\eeq
In particular, we conclude that $\La_\ell(a)$ is $m$-localizing for $H_{\eps,\bom}$ for all $a \in \partial_{\mathrm{ex}}^{ \G_1}\wtilde{\Phi}_r$, $r=1,2,\ldots, R$. 
 
Each $\Ups_r=\Ups_{\Phi_r}$, $r=1,2,\dots,R$,  clearly satisfies all the requirements to be a buffered subset of $\La_L$ with $\cG_{\Ups_r}=\partial_{\mathrm{ex}}^{ \G_1}\wtilde{\Phi}_r$  (see Definition~\ref{defbuff}), except that we do not know if $\Ups_r$ is   $L$-level spacing for $H_{\eps,\bom}$.
(Note that the sets $\{\Ups_r\up{0}\}_{r=1}^R$ are disjoint, but the sets $\set{\Ups_r}_{r=1}^R$ are not necessarily disjoint.)  Note also that it follows from \eq{diamwphi} that
\beq\label{diamUps}
\diam \Ups_r \le \diam  \overline{\wtilde{\Phi}}_r  + \ell \le  \rho \ell \pa{3\abs{\Phi_r}+1}+\ell\le 5\ell \abs{\Phi_r},
\eeq
so
\beq\label{sumdiam}
\sum_{r=1}^R \diam \Ups_r\le 5\ell N_\ell  \le 5 \ell^{(\gamma-1)\tzeta  +1} \ll \ell^{\gamma \tau}=L^\tau,
\eeq
since $(\gamma-1)\tzeta  +1 <(\gamma-1)\beta  +1< \gamma \tau$ (see \eq{ttauzeta}).

We can  arrange for $\set{\Ups_r}_{r=1}^R$ to be a collection of buffered subsets of $\La_L$  as follows.  It follows from Lemma~\ref{lemSep} that 
for any $\Th\subset \La_L$ we have
\begin{align}\label{probspaced}
\P\set{\Th \sqtx{is} L\text{-level spacing for}\; H_{\eps,\bom} }\ge 1 -Y_{\eps_0} \e^{-(2\alpha-1)L^\beta}\pa{L+1}^{2d}.
\end{align}
 Let
\begin{align}
\cF_N=\bigcup_{r=1}^N  \cF(r), \sqtx{where} \cF(r)=\set{\Phi \subset \Xi_{L,\ell} ; \;  \Phi \sqtx{is}  \G_2\text{-connected}\sqtx{and} \abs{\Phi}=r}.
\end{align}
Setting $ \cF(r,a)=\set{\Phi \in \cF(r); \, a\in \Phi}$ for $a\in \Xi_{L,\ell}$, and noting that
 each vertex in the graph $\G_2$ has less than $d\pa{3^{d-1}+4^{d-1}}\le d4^d$ nearest neighbors , we get
\begin{align}\label{cFN}
\abs{\cF(r,a)}\le  (r-1)! \pa{d4^d}^{r-1}\quad & \Longrightarrow \quad  \abs{\cF(r)}\le (L+1)^d (r-1)! \pa{d4^d}^{r-1}\\ \notag
& \Longrightarrow \quad \abs{\cF_N}\le (L+1)^d N! \pa{d4^d}^{N-1} .
\end{align}
Letting  $\cS_N$ denote that the event that the box $\La_L$ and the subsets  $\set{\Ups_\Phi}_{\Phi \in  \cF_N}$ are all  $L$-level spacing for  $ H_{\eps,\bom}$, and recalling the choice of $N=N_\ell$ in \eq{setN},   we get from \eq{probspaced} and \eq{cFN} that
\beq\label{probSN}
\P\set{\cS_N^c} \le Y_{\eps_0}\pa{1 +  (L+1)^d N_\ell! \pa{d4^d}^{N_\ell-1} }(L+1)^{2d} \e^{-(2\alpha-1)L^\beta} < \tfrac 12 \e^{-L^\zeta} 
\eeq
for sufficiently large $L$, since $(\gamma-1)\tzeta < (\gamma-1)\beta<\gamma \beta$ and $\zeta < \beta$.

We now define the event  $\cE_N= \cB_N \cap \cS_N$.  
 It follows from \eq{probBN} and \eq{probSN} that
\beq
\P\set{\cE_N}> 1-  \e^{-L^\zeta}.
\eeq
Note that for $\bom \in \cE_N$ the subsets $\set{\Ups_r}_{r=1}^R$ constructed above are buffered subsets. To finish the proof we need to show that for all  $\bom \in \cE_N$ the box $\La_L$ is 
$ M$-localizing for $ H_{\eps,\bom}$, where $M$ is given  in \eq{Minduc}.

Let us fix $\bom \in \cE_N$.   Then we have \eq{implyloc2}, $\La_L$ is level spacing for $ H_{\eps,\bom}$, and  the subsets $\set{\Ups_r}_{r=1}^R$ constructed in \eq{constUps} are buffered subsets of $\La_L$ for $ H_{\eps,\bom}$. It follows from \eq{covproperty} and Definition~\ref{defbuff}(iii) that
\beq\label{LadecompU}
\La_L=\set{ \bigcup_{a \in  \cG} {\Lambda}_{\ell}^{\La_L, \frac \ell {10}}(a)}\cup \set{\bigcup_{r=1}^R\Ups_r^{{\La_L},\frac \ell {10}}}.
\eeq

Since $\eps$ and $\bom$ are now fixed, we omit them from the notation.
 Let $\set{(\psi_\lambda,\lambda)}_{\lambda \in \sigma(H_{\La_L})}$ be an eigensystem for $H_{\La_L}$.  Given $a\in \cG$,
let
 $\set{(\vphi\up{a}_x, \lambda\up{a}_x)}_{x\in \La_\ell(a)}$ be an  $m$-localized eigensystem for $\La_\ell(a)
 $.  For $r=1,2,\ldots,R$,  let  $\set{(\phi_{\nu\up{r}},\nu\up{r})}_{\nu\up{r} \in \sigma(H_{\Ups_r})}$ be an eigensystem for $H_{\Ups_r}$, and set
 \beq\label{sigmanotsigma}
\sigma_{\Ups_r} (H_{\La_L})= \set{\wtilde{\nu}\up{r};  \nu\up{r}\in\sigma_\cB(H_{{\Upsilon_r}})}\subset  \sigma(H_{\La_L})\setminus\sigma_{\cG} (H_{\La_L}),
\eeq
where $\wtilde{\nu}\up{r}$ is given in \eq{injbad}, which also gives $\sigma_{\Ups_r} (H_{\La_L})\subset  \sigma(H_{\La_L})\setminus\sigma_{\cG_{\Ups_r}} (H_{\La_L})$, but the argument actually shows 
$\sigma_{\Ups_r} (H_{\La_L})\subset \sigma(H_{\La_L})\setminus\sigma_{\cG} (H_{\La_L})$.
We also set
\beq
\sigma_\cB(H_{\La_L}) =\bigcup_{r=1}^R \sigma_{\Ups_r} (H_{\La_L})\subset  \sigma(H_{\La_L})\setminus\sigma_{\cG} (H_{\La_L}).
\eeq

We claim
\beq\label{claimsp}
\sigma(H_{\La_L})= \sigma_{\cG} (H_{\La_L}) \cup \sigma_\cB (H_{\La_L}).
\eeq
To see this, suppose we have  $\lambda \in\sigma(H_{\La_L})\setminus \pa{ \sigma_{\cG} (H_{\La_L}) \cup \sigma_\cB (H_{\La_L})}$.  Since $\La_L$ is level spacing for $H$, it follows from Lemma~\ref{lem:ident_eigensyst}(ii)(c) that 
\beq\label{psidecgood6344}
\abs{\psi_\lambda(y)}\le\e^{-m_2{{\ell_\tau}}}   \qtx{for all} y\in \bigcup_{a\in \cG}\La_\ell^{{\La_L},2\ell_\tau}(a),
\eeq
and it follows from Lemma~\ref{lembad}(ii) that
\beq
\abs{\psi_\lambda (y)} \le  \e^{-m_5{\ell_{\tau}}} \qtx{for all}y\in \bigcup_{r=1}^R\Ups_r^{{\La_L},2\ell_\tau }.
\eeq
Using \eq{LadecompU}, we  conclude that (note $m_5\le m_2$)
\beq
1= \norm{\psi_\lambda} \le \e^{-m_5{\ell_{\tau}}} \pa{L+1}^{\frac d 2} <1,
\eeq 
a contradiction.  This establishes the claim.

We will now index  the eigenvalues and eigenvectors of $H_{\La_L}$ by sites in $\La_L$ using
Hall's Marriage Theorem, which states a necessary and sufficient condition for the existence of a perfect matching in a bipartite graph.  (See Appendix~\ref{appHall}, based on  \cite[Chapter~2]{BDM}.) We consider the  bipartite graph $\G = (\La_L,\sigma(H_{\La_L});\E)$, where
  the edge set $\E \subset \La_L\times \sigma(H_{\La_L})$ is defined as follows.
For each $\lambda \in \sigma_{\cG} (H_{\La_L})$ we fix $\lambda\up{a_\lambda}_{x_\lambda}\in  \cE_{\cG}^{\La_L} (\lambda)$, and set (recall \eq{defUpscheck};  we write
$\widehat{\Upsilon}_r=\widehat{\Upsilon_r}$ and $\widehat{\Upsilon}_{r,\tau}=(\widehat{\Upsilon}_{r})_{\tau}$)
\beq
\cN_0(x)=\begin{cases}
 \set{\lambda \in \sigma_{\cG} (H_{\La_L}); \, \norm{x_\lambda -x}<  {\ell_\tau}}&\qtx{for} x \in \La_L \setminus \bigcup_{r=1}^R\widehat{\Upsilon}_r\\
 \emptyset &\qtx{for} x \in  \bigcup_{r=1}^R\widehat{\Upsilon}_r
 \end{cases}.
 \eeq
We define
\beq
\cN(x) =
\begin{cases}
\cN_0(x) & \qtx{for}   x
\in \La_L \setminus \bigcup_{r=1}^R\widehat{\Upsilon}_{r,\tau}\\
\sigma_{\Ups_r} (H_{\La_L}) & \qtx{for}   x \in \widehat{\Upsilon}_r, \ r=1,2,\ldots,R\\
\cN_0(x) \cup \sigma_{\Ups_r} (H_{\La_L})  & \qtx{for}   x \in {\widehat{\Upsilon}_{r,\tau}\setminus \widehat{\Upsilon}_r}, \ r=1,2,\ldots,R
\end{cases},
\eeq
and set
$\E=\set{(x,\lambda)\in\La_L\times \sigma(H_{\La_L}); \; \lambda \in \cN(x)}$.

$\cN(x) $  was defined to ensure $\abs{\psi_\lambda(x)}\ll 1 $ for $\lambda \notin  \cN(x) $. This can be seen as follows:
\begin{itemize}
\item 
 If $ x \in \La_L$ and 
$\lambda \in \sigma_{\cG} (H_{\La_L}) \setminus \cN_0(x)$, we have $\lambda=\widetilde{\lambda}\up{a_\lambda}_{x_\lambda}$ with $ \norm{x_\lambda -x}\ge {\ell_\tau}$, so
\beq
\abs{\psi_\lambda(x)} \le  \abs{\vphi\up{a_\lambda}_{x_\lambda}(x)}+ \norm{\vphi\up{a_\lambda}_{x_\lambda}-\psi_\lambda}\le \e^{-m\ell_\tau} + {2}\e^{-m_1{\ell^\tau}}\e^{L^\beta}\le 3\e^{-m_1{\ell^\tau}}\e^{L^\beta},
\eeq
using  \eq{hypdec} and  \eq{difeq86}. 

\item  If $x \in \La_L \setminus \widehat{\Upsilon}_{r,\tau}$ and $\lambda \in \sigma_{\Ups_r} (H_{\La_L})$, then $\lambda=\wtilde{\nu}\up{r}$ for some $\nu\up{r}\in \sigma_{\cB} (H_{\Ups_r})$, and 
\beq
\abs{\psi_\lambda(x)} \le  \abs{\phi_{\nu\up{r}}(x)}+ \norm{\phi_{\wtilde{\nu}\up{r}}-\psi_\lambda}\le  e^{-m_2{\ell_\tau}} + 2 \e^{-m_4 {\ell_\tau}}\e^{L^\beta}\le  3 \e^{-m_4{ \ell_\tau}}\e^{L^\beta},
\eeq
using  \eq{vthetasmall} and \eq{difeqr}. (Note $\phi_{{\nu}\up{r}}(x)=0$ if $x \notin \Ups_r$.)

\end{itemize}
It follows that for all $x \in \La_L$ and $\lambda \in \sigma (H_{\La_L})\setminus \cN(x)$ we have 
\beq \label{decoutx6}
\abs{\psi_\lambda(x)} \le   3 \e^{-m_4 {\ell_\tau}}\e^{L^\beta}\le \e^{-\frac 12 m_4 {\ell_\tau}}.
\eeq

Since $\abs{\La_L}= \abs{\sigma (H_{\La_L})}$, to apply Hall's Marriage Theorem we only need to verify  Hall's condition \eq{eq:Hallcond1}.
 Let $\cN(\Th)=\bigcup_{x\in \Th} \cN(x)$ for $\Theta \subset \La_L$.   We fix   $\Theta \subset \La_L$, and  let $Q_\Th$ be the orthogonal projection onto the span of $\set{\psi_\lambda; \; \lambda \in \cN(\Th)}$.
For every $\lambda \notin \cN(\Th)$ we have \eq{decoutx6} for all $x \in\Th$, so
\beq
\norm{(1-Q_\Th)\Chi_\Th}\le \abs{\La_L}^{\frac 12 } \abs{\Th}^{\frac 12 }\e^{-\frac 12 m_4{\ell_\tau}}\le (L+1)^d \e^{-\frac 12 m_4{\ell_\tau}}<1,
\eeq
so it follows from Lemma~\ref{lemPQ} that
\beq\label{Hallc0}
\abs{\Th}= \tr \Chi_\Th \le  \tr Q_\Th= \abs{\cN(\Th)},
\eeq
which is Hall's condition \eq{eq:Hallcond1}.

Thus we can apply Hall's Marriage Theorem, concluding that there exists a bijection
\beq
x\in \La_L \mapsto \lambda_x \in \sigma(H_{\La_L}), \qtx{where}  \lambda_x \in \cN(x).
\eeq 
We set $\psi_x=\psi_{\lambda_x}$ for all $x\in \La_L $.

To finish the proof we need to show  that  $\set{(\psi_x,\lambda_x)}_{x\in \La_L}$ is an $M$-localized eigensystem for $\La_L$, where $M$ is given in \eq{Minduc}.
We  fix $x \in \La_L$,
take $y\in  \La_L $, 
and   consider several cases:

\begin{enumerate}
\item Suppose $\lambda_x \in \sigma_{\cG} (\La_L)$. In this case  $x \in \La_\ell (a_{\lambda_x})$ with  $a_{\lambda_x}\in \cG$, and $\lambda_x \in\sigma_{\set{a_{\lambda_x}}}(H_{\La_L}) $. In view of \eq{LadecompU} we consider two cases:

\begin{enumerate}
\item  
If  $y \in \La_{\ell}^{{\La_L}, \frac \ell{10}}(a)$ 
for some $a\in \cG$ and $\norm{y-x} \ge 2\ell$, we must have  $\La_\ell (a_{\lambda_x})\cap  \La_\ell(a)=\emptyset$, so it follows from \eq{sigmaab} that
$\lambda_x \notin  \sigma_{\set{a}}(H_{\La_L})$, and  \eq{psidecgoodpr} yields
\beq\label{psidecgoodpr2} 
\abs{\psi_x(y)}\le \e^{-m_3{\norm{y_1-y}}}  \abs{\psi_x(y_1)} \sqtx{for some} y_1 \in{\partial}^{\La_L,\ell_{\ttau }}  \Lambda_{\ell}(a).
\eeq

\item If    $y\in \Ups_r^{{\La_L},\frac \ell {10}}$ for some $r \in \set{1,2,\ldots,R}$, and $\norm{y-x} \ge \ell+\diam \Ups_r$,  we must have $\La_\ell (a_{\lambda_x})\cap  \Ups_r=\emptyset$.   It follows from   \eq{sigmaab} that $\lambda_x \notin   \sigma_{\cG_{\Ups_r}} (H_{\La_L})$, and clearly $\lambda_x \notin    \sigma_{\Ups_r} (H_{\La_L})$ in view of \eq{sigmanotsigma}.   Thus Lemma~\ref{lembad}(ii)   gives \beq\label{gggsum356}
\abs{\psi_x(y)}\le    \e^{-m_5{\ell^\tau}}\abs{\psi_x(v)}\qtx{for some} v\in {\partial}^{\La_L, 2{\ell_\tau} }  {\Upsilon_r}.
\eeq

\end{enumerate}
\item   Suppose $\lambda_x \notin \sigma_{\cG} (\La_L)$. Then it follows from \eq{claimsp} that we must have $\lambda_x \in    \sigma_{\Ups_s} (H_{\La_L})$ for some $s \in \set{1,2,\ldots,R}$. In view of \eq{LadecompU} we consider two possibilities:

\begin{enumerate}
 \item If  $y \in \La_{\ell}^{{\La_L}, \frac \ell{10}}(a)$ 
for some $a\in \cG$, and $\norm{y-x} \ge \ell+\diam \Ups_s$, we must have  $\La_\ell (a)\cap  \Ups_s=\emptyset$,  and Lemma ~\ref{lem:ident_eigensyst}(i)(c) yields \eq{psidecgoodpr2}.

 \item If $y\in \Ups_r^{{\La_L},\frac \ell {10}}$ for some $r \in \set{1,2,\ldots,R}$, and $\norm{y-x} \ge \diam \Ups_s+\diam \Ups_r$,  we must have $r\ne s$.   Thus   Lemma~\ref{lembad}(ii)  yields \eq{gggsum356}.
 
\end{enumerate}
\end{enumerate}

Now let us fix $x \in \La_L$, and take $y\in  \La_L $ such that $\norm{y-x} \ge {L_\tau}$.
Suppose $\abs{\psi_x(y)} >0$, since otherwise there is nothing to prove.  We estimate $\abs{\psi_x(y)} $ using either \eq{psidecgoodpr2}  or  \eq{gggsum356} repeatedly, as appropriate, stopping when we get too close  to $x$ so we are not in one the cases described above.  (Note that this must happen since $\abs{\psi_x(y)} >0$.) We accumulate decay only when we use \eq{psidecgoodpr2}, and just use $\e^{-m_5{\ell^\tau}}< 1$ when using \eq{gggsum356}, getting
\begin{align}\label{repeateddecay}
\abs{\psi_x(y)} &\le   \e^{-m_3\pa{\norm{y-x} -\sum_{r=1}^R \diam \Ups_r -2\ell}}\le   \e^{-m_3\pa{\norm{y-x} - 5 \ell^{(\gamma-1)\tzeta  +1}   -2\ell}}\\ \notag  &
\le  \e^{-m_3\norm{y-x}\pa{1 - 7 \ell^{(\gamma-1)\tzeta  +1-\gamma \tau} }}\le    \e^{-M\norm{y-x}},
\end{align}
where we used \eq{sumdiam} and took
\begin{align}
M&=m_3\pa{1 - 7 \ell^{(\gamma-1)\tzeta  +1-\gamma \tau} }\ge 
m\pa{1-   C_{d,m_-,\eps_0} \ell^{\frac {\tau -1}2}}\pa{1 - 7 \ell^{(\gamma-1)\tzeta  +1-\gamma \tau} }\\ \notag & \ge m\pa{1-   C_{d,m_-,\eps_0} \ell^{-\min\set{\frac {1- \tau }2, \gamma \tau- (\gamma-1)\tzeta  -1}}}
\end{align}
where we used \eq{m4} .

We conclude that $\set{(\psi_x,\lambda_x)}_{x\in \La_L}$ is an $M$-localized eigensystem for $\La_L$, where $M$ is given in \eq{Minduc}, so the box is  $\La_L$ is 
$ M$-localizing for $ H_{\eps,\bom}$.
\end{proof}

\begin{proof}[Proof of Proposition~\ref{propMSA}]
We assume \eq{initialconinduc} and set  $L_{k+1}=L_k^\gamma$  for $k=0,1,\ldots$. If $L_0$ is sufficiently large it follows from Lemma~\ref{lemInduction} by an induction argument that
\begin{align}
\inf_{x\in \R^d} \P\set{\La_{L_k} (x) \sqtx{is}  m_k   \text{-localizing for} \; H_{\eps,\bom}} \ge 1 -  \e^{-L_k^\zeta} \mqtx{for} k=0,1,\ldots,
\end{align}
where   for $k=1,2,\dots$ we have 
\beq
m_{k} \ge m_{k-1} \pa{1-   C_{d,m_-,\eps_0} \L_{k-1}^{-\vrho}}, \sqtx{with}  \vrho= \min\set{\tfrac {1- \tau }2, \gamma \tau- (\gamma-1)\tzeta  -1}.
\eeq
Thus for all $k=1,2,\dots$, taking $L_0$ sufficiently large we get 
\begin{align}
m_{k}\ge m_0 \prod_{j=0}^{k-1} \pa{1-   C_{d,m_-,\eps_0} \L_{0}^{-\vrho\gamma^{j}}}\ge  m_0 \prod_{j=0}^{\infty} \pa{1-   C_{d,m_-,\eps_0} \L_{0}^{-\vrho\gamma^{j}}}\ge \frac {m_0}2,
\end{align}
finishing the proof of Proposition~\ref{propMSA}.
\end{proof}

\subsection{Removing the restriction on scales}

We will now show how Theorem~\ref{thmMSA} follows from Proposition~\ref{thmfullMSA}.
\begin{proof}[Proof of Theorem~\ref{thmMSA}]
Assume the conclusions of Proposition~\ref{thmfullMSA}, that is, for $L_0\ge \cL $ and  $ \eps\le  \eps_0= \frac 1 {4d} \e^{-L_0^\beta}$, setting $L_{k+1}=L_k^\gamma$ for $k=0,1,\ldots$, we have \eq{MSAlk}.

Given a scale   $L\ge L_1$, let $k=k(L)\in \set{1,2,\ldots}$ be defined by
$L_k \le L <L_{k+1}$, and set  $\wtilde{\ell}= L_{k-1}$.  We have
 $L_k= \wtilde{\ell}^\gamma\le L< L_{k+1}= \wtilde{\ell}^{\gamma^2}$,  so $L=  \wtilde{\ell}^{\gamma^\prime}$ with $\gamma \le \gamma^\prime<\gamma^2$. We proceed as in Lemma~\ref{lemInduction}. We take
$\La_L=\La_L(x_0)$, where $x_0\in \R^d$,  and let
${\mathcal C}_{L,\wtilde{\ell}}={\mathcal C}_{L,\wtilde{\ell}} \left(x_0 \right)$ be the suitable $\wtilde{\ell}$-cover of $\La_L$. We let  $\cB_0$ denote the event that all  boxes in ${\mathcal C}_{L,\wtilde{\ell}}$  are  $\tfrac {m_{\eps,L_0} }2$-localizing for  $ H_{\eps,\bom}$.  It follows from \eq{MSAlk} that
\begin{align}\label{probB0}
\P\set{\cB_0^c}\le \pa{\tfrac{2L} {\wtilde{\ell}}}^{d} \e^{-\wtilde{\ell}^\zeta}= 2^{d} \wtilde{\ell}^{(\gamma^\pr-1)d}\e^{-\wtilde{\ell}^\zeta}\le     2^{d} L^{(1-\frac 1{ \gamma^\pr})d}\e^{-L^{\frac \zeta { \gamma^\pr}}}  < \tfrac 12 \e^{-L^\xi} ,
\end{align}
if $L_0$ is sufficiently large, since $\xi \gamma^2 < \zeta$.  Moreover, letting $\cS_0$ denote the event that  the box $\La_L$ is level spacing for  $H_{\eps,\bom}$, it follows from Lemma~\ref{lemSep} that 
\begin{align}\label{probspacedLL}
\P\set{\cS_0^c }\le Y_{\eps_0} \e^{-(2\alpha-1)L^\beta}\pa{L+1}^{2d}\le  \tfrac 12 \e^{-L^\xi},
\end{align}
if $L_0$ is sufficiently large, since $\xi < \beta$.  Thus, letting $\cE_0= \cB_0 \cap \cS_0$, we have 
\beq
\P\set{\cE_0 }\ge 1-  \e^{-L^\xi}.
\eeq

It only remains to prove that  $\La_L$ is  $\tfrac {m_{\eps,L_0} }4$-localizing for  $ H_{\eps,\bom}$ for all $\bom \in \cE_0$.  To do so,  we fix $\bom \in \cE_0$ and  proceed as in the proof of Lemma~\ref{lemInduction}.   Since $\bom \in \cB_0$, we have $\cG=\cG(\bom)=\Xi_{L,\wtilde{\ell}}$. Since $\eps$ and $\bom$ are now fixed, we omit them from the notation. The proof of Lemma~\ref{lemInduction} applies, we get $\sigma(H_{\La_L})= \sigma_{\cG} (H_{\La_L})$ as in \eq{claimsp}, and obtain an eigensystem  $\set{(\psi_x,\lambda_x)}_{x\in \La_L}$ for $H_{\La_L}$ using Hall's Marriage Theorem. To finish the proof we need to show  that  $\set{(\psi_x,\lambda_x)}_{x\in \La_L}$ is an $\tfrac {m_{\eps,L_0} }4$-localized eigensystem for $\La_L$.  Given $x,y\in \La_L$ with $\norm{y-x} \ge 2 \wtilde{\ell}$, we have $y \in \La_{\ell}^{{\La_L}, \frac {\wtilde{\ell}}{10}}(a)$ 
for some $a\in \Xi_{L,\wtilde{\ell}}$, and \eq{psidecgoodpr2}  holds with  $y_1 \in{\partial}^{\La_L,{\wtilde{\ell}}_{\ttau }}  \Lambda_{\wtilde{\ell}}(a)$.  If $\norm{y-x} \ge L_\tau$, we proceed as in  \eq{repeateddecay}, stopping when we get within $2 \wtilde{\ell}$ of $x$, obtaining
\begin{align}\label{repeateddecay2}
\abs{\psi_x(y)} &\le   \e^{-\widetilde{m}_3\pa{\norm{y-x} -2\wtilde{\ell}}}\le 
  \e^{-\widetilde{m}_3\norm{y-x}\pa{1 - 3 \wtilde{\ell}^{1-\gamma^\pr \tau} }}\le    \e^{-\wtilde{M}\norm{y-x}},
\end{align}
where (recall  \eq{m4}) $\widetilde{m}_3\ge \tfrac {m_{\eps,L_0} }2\pa{1-   C_{d,m_-,\eps_0} \wtilde{\ell}^{\frac {\tau -1}2}}$, and
\begin{align}
\wtilde{M}&= \widetilde{m}_3\pa{1 - 3 \wtilde{\ell}^{1-\gamma^\pr \tau} }\ge  \tfrac {m_{\eps,L_0} }2\pa{1-   C_{d,m_-,\eps_0} \wtilde{\ell}^{-\min\set{\frac {1- \tau }2, \gamma^\pr \tau- 1}}}\\ \notag   &    \ge  \tfrac {m_{\eps,L_0} }2
\pa{1-   C_{d,m_-,\eps_0} L_0^{-\min\set{\frac {1- \tau }2, \gamma^\pr \tau- 1}}}\ge \tfrac {m_{\eps,L_0} }4
\end{align}
for $L_0$ large.
\end{proof}

\section{Deriving localization}\label{seclocproof}
 In this section we consider an   Anderson model $H_{\eps,\bom}$ and derive localization results from  Theorem~\ref{thmMSA}.
We start by proving Theorem~\ref{thmloc},  using the following lemma.

 \begin{lemma}\label{lemWimp}  Fix $\eps_0>0$  and $m_->0$. There exists a finite scale  $\cL_{ \eps_0,m_-}$ such that  
 for all $\ell\ge \cL_{ \eps_0,m_-}$, $a\in \Z^d$,   $\lambda \in \R$,  $\eps \le\eps_0$, and $m\ge m_-$,  given    an $m$-localizing box $\La_\ell(a)$  for the discrete Schr\"odinger operator  $H_\eps$  with   an $m$-localized eigensystem $\set{\vphi_x,\lambda_x}_{x \in  {\Lambda}_{\ell}(a)}$,    we have
 \beq
\max_{b\in \La_{\frac \ell 3}(a)} W\up{a}_{\eps,\lambda}(b)> \e^{-\frac 1 4 m \ell}\quad \Longrightarrow \quad  \min_{x\in \La_\ell^{\ell_\tau}(a)} \abs{\lambda -\lambda_x}  < 
\tfrac 1 2 \e^{-L^\beta}.
\eeq
 \end{lemma}
  
 \begin{proof}  Suppose $ \abs{\lambda -\lambda_u}  \ge\tfrac 1 2\e^{-L^\beta}$ for all $u \in \La_\ell^{\ell_\tau}(a)$.  Let $\psi \in \cV_\eps(\lambda)$.  Then it follows from Lemma~\ref{lemdecay2}(ii) that  for large $\ell$ and $b\in \La_{\frac \ell 3}(a)$  we have
 %%CORRECTION: $ \abs{\psi(b)}$ not $ \abs{\psi(a)}$
 \beq
 \abs{\psi(b)}\le \e^{-m_3 \pa{\frac \ell 3 - \ell_{\ttau}}}\norm{T_a^{-1} \psi}
  \scal{\tfrac \ell 2 + 1}^{\nu} \le \e^{-\frac 1 4 m \ell} \norm{T_a^{-1} \psi}.
 \eeq
 \end{proof}

\begin{proof}[Proof of Theorem~\ref{thmloc}]  Suppose Theorem~\ref{thmMSA} holds for some $L_0$, and let $\eps_0= \frac 1 {4d} \e^{-L_0^{ \beta}}$ and $m_\eps=  \tfrac {m_{\eps,L_0} }4\ge \frac {\log 3}4$.
Consider  $L_0^\gamma \le\ell \in 2\N$  and $a\in \Z^d$.  We have  
 \beq
  \La_{5\ell}(a) =\bigcup_{b\in \set{ a+ \frac 1 2\ell  \Z^d}, \ \norm{b-a}\le  2\ell} \La_{\ell}(b).
  \eeq
  Let $\cY_{\eps,\ell,a}$ denote the event that  $ \La_{5\ell}(a)$ is level spacing for $H_{\eps,\bom}$ and the boxes  $ \La_{\ell}(b)$ are  $m_\eps$-localizing for $H_{\eps,\bom}$ for all $b\in \set{ a+ \frac 1 2\ell  \Z^d}$ with  $\norm{b-a}\le 2 \ell$.  It follows from \eq{concMSA} and Lemma~\ref{lemSep}  that
  \beq
  \P\set{\cY_{\eps,\ell,a}^c}\le 5^d \e^{-\ell^\xi}  +  Y_{\eps_0}\pa{5\ell+1}^{2d} \e^{-(2\alpha-1)(5\ell)^\beta}\le C_{\eps_0} \e^{-\ell^\xi} .
  \eeq
  
  Suppose $\bom \in \cY_{\eps,\ell,a}$, $\lambda \in \R$,  and $\max_{b\in \La_{\frac \ell 3}(a)} W\up{a}_{\eps,\bom,\lambda}(b)> \e^{-\frac 1 4 m_\eps\ell} $. It follows from Lemma~\ref{lemWimp} that $ \min_{x\in \La_\ell^{\ell_\tau}(a)} \abs{\lambda -\lambda_x\up{a}}  < \tfrac 1 2\e^{-L^\beta}$.  Since $ \La_{5\ell}(a)$ is  level spacing for $H_{\eps,\bom}$, using Lemma~\ref{lem:ident_eigensyst} we conclude that 
  \beq
   \min_{x\in \La_\ell^{\ell_\tau}(b)} \abs{\lambda -\lambda\up{b}_x}\ge \e^{-(5\ell)^\beta} - 
   2  \e^{-m_1 {\ell_\tau}} - \tfrac 1 2\e^{-L^\beta}\ge \tfrac 1 2 \e^{-L^\beta}
  \eeq
  for all $b\in \set{ a+ \frac 1 2\ell  \Z^d}$ with $\ell \le \norm{b-a}\le 2\ell$. Since
  \beq\label{Aell2}
 A_\ell(a)\subset \bigcup_{b\in \set{ a+ \frac 1 2\ell  \Z^d}, \ \ell \le \norm{b-a}\le 2\ell} \La_{\ell}^{\frac \ell 7}(b),
  \eeq
 it follows from Lemma~\ref{lemdecay2}(ii)  that for all  $y\in A_\ell(a)$  we have, given  $\psi \in \cV_{\eps,\bom}(\lambda)$,
 \begin{align}
 \abs{\psi(y)}\le  \e^{-(m_\eps)_3\pa{\frac \ell 7 - \ell_{\ttau}}} \norm{T_a^{-1} \psi}\la \tfrac 5 2 \ell  +1 \ra^\nu\le \e^{-m_\eps \frac \ell 8}\norm{T_a^{-1} \psi} \le  \e^{- \frac 7 {132} m_\eps \norm{y-a} }\norm{T_a^{-1} \psi},
 \end{align}
so we get
\beq
W\up{a}_{\eps,\bom,\lambda}(y)\le \e^{-\frac 7 {132}        m_\eps \norm{y-a}} \qtx{for all} y\in A_\ell(a).
\eeq

Since we have \eq{boundGW}, we conclude that for $\bom \in \cY_{\eps,\ell,a}$ we always have
\begin{align} 
W\up{a}_{\eps,\bom,\lambda}(a)W\up{a}_{\eps,\bom,\lambda}(y) & \le 
\max \set{\e^{- \frac 7 {66} m_\eps\norm{y-a}}\la  y-a\ra^\nu, \e^{- \frac 7 {132}m_\eps \norm{y-a}}}\\   \notag 
& \le \e^{- \frac 7 {132} m_\eps \norm{y-a}} \qtx{for all} y\in A_\ell(a).
\end{align}
 \end{proof}

We now turn to Corollary~\ref{corloc}.

\begin{proof}[Proof of Corollary~\ref{corloc}]

Parts (i) and (ii) are proven in the same way as  \cite[Theorem~7.1(i)-(ii)]{GKber}.  Note that we have
$\max_{b\in \La_{\frac \ell 3}(a)} W\up{a}_{\eps,\bom,\lambda}(b)$  in \eq{locimpl}  instead of simply  $ W\up{a}_{\eps,\bom,\lambda}(a)$ because we do not have the unique continuation principle in the lattice. If $\lambda$ is a generalized eigenvalue for  $H_{\eps,\bom}$ we could have 
$ W\up{a}_{\eps,\bom,\lambda}(a)=0$, but we will always have $\max_{b\in \La_{\frac \ell 3}(a)} W\up{a}_{\eps,\bom,\lambda}(b)>0$ for all large $\ell$.

Part (iii) is proven similarly to  \cite[Theorem~7.2(i)]{GKber}.  There are  some small differences, so we give the proof here. We use
the fact that  for any $\ell_0 \in 2\N$, setting $\ell_{k+1}=2\ell_k$ for $k=0,1,2,\ldots$, we have  (recall \eq{Aell})
   \beq
  \Z^d = \La_{3\ell_k}(a)\cup \bigcup_{j=k}^\infty  A_{\ell_j}(a) \qtx{for} k=0,1,2,\ldots.
   \eeq  
   
   We fix $\eps \le \eps_0$. 
Given  $k \in \N$, we set $L_{k}=2^{k}$,
 and consider the event  
\beq
\cY_{\eps,k}:=  \bigcap_{x \in \Z^{d}; \, \norm{x} \le  \e^{\frac 1 {2d}  L_k^\xi}}  \cY_{\eps,L_k,x} ,
\eeq
where $\cY_{\eps,L_k,x} $ is the event given in Theorem~\ref{thmloc}.  It follows from \eq{cUdesiredint} that for sufficiently large $k$ we have
\beq
\P\set{\cY_{\eps,k}}\ge 1 - C_{\eps_0}\pa{2\e^{\frac 1 {2d}  L_k^\xi}+1}^d \e^{-L_k^\xi}\ge  1 - 3^d C_{\eps_0}  \e^{-\frac 1 {2}  L_k^\xi},
\eeq
so we conclude from the  Borel-Cantelli Lemma  that 
\beq\label{cUinfty2}
\P \set{\cY_{\eps,\infty}}=1, \quad \text{where}\quad \cY_{\eps,\infty}=\liminf_{k\to \infty}\cY_{\eps,k}.
\eeq

We now fix $\bom \in  \cY_{\eps,\infty}$, so there exists $k_{\eps,\bom}\in \N$ such that
$\bom \in \cY_{\eps,L_k,x} $ for all $k_{\eps,\bom} \le k \in \N$ and $x \in \Z^d$ with  $\norm{x} \le  \e^{ \frac 1 {2d}  L_k^\xi}$.
Given $x \in \Z^d$, we define $k_x\in \N$  by
\beq\label{defkx}
 \e^{ \frac 1 {2d}  L_{k_x-1}^\xi} < \norm{x} \le  \e^{\frac 1 {2d}  L_{k_x}^\xi} \qtx{if} k_x \ge 2,
\eeq
and set $k_x=1$ otherwise.  We set
$k_{\eps,\bom,x}=\max\set{k^\pr_{\eps,\bom},k_x}$, where $k^\pr_{\eps,\bom}=\max\set{k_{\eps,\bom},2}$.

Let  $x \in \Z^d$.  If $y \in B_{\eps,\bom,x}=\bigcup_{k=k_{\eps,\bom,x}}^\infty  A_{L_k}(x)$, we have $y \in  A_{L_{k_1}}(x)$ for some $k_1 \ge  k_{\eps,\bom,x}$ and  $\bom \in \cY_{\eps,L_{k_1},x} $, so it follows from \eq{WW} that
\beq  \label{WW2}
W\up{x}_{\eps,\bom,\lambda}(x)W\up{x}_{\eps,\bom,\lambda}(y)\le 
 \e^{- \frac 7 {132} m_\eps \norm{y-x}} \qtx{for all} \lambda \in \R.
\eeq
If $y \notin B_{\eps,\bom,x}$, we must have  $\norm{y-x}< \frac 8 7 L_{k_{\eps,\bom,x}}$, so for all $ \lambda \in \R$, using \eq{boundGW} and \eq{defkx},
\begin{align} \label{preSUDEC2}
& W\up{x}_{\eps,\bom,\lambda}(x)W\up{x}_{\eps,\bom,\lambda}(y) =W\up{x}_{\eps,\bom,\lambda}(x)W\up{x}_{\eps,\bom,\lambda}(y) \e^{\frac 7 {132} m_\eps \norm{y-x}}\e^{- \frac 7 {132} m_\eps \norm{y-x}}\\ \notag 
& \qquad
\le \la y-x \ra^\nu \e^{\frac 7 {132} m_\eps \norm{y-x}} \e^{- \frac 7 {132} m_\eps \norm{y-x}} \\ \notag & \qquad  \le \scal{\tfrac 8 7 L_{k_{\eps,\bom,x}} }^\nu  \e^{\frac 2 {33} m_\eps  L_{k_{\eps,\bom,x}}} \e^{- \frac 7 {132} m_\eps \norm{y-x}} 
\\ & \qquad
\le
\begin{cases} 
\scal{ \tfrac {16} 7 \pa{\log \norm{x}^{2d}}^{\frac 1 \xi} }^\nu \e^{\frac 4 {33} m_\eps \pa{\log \norm{x}^{2d}}^{\frac 1 \xi}} 
 \e^{- \frac 7 {132} m_\eps \norm{y-x}} & \quad  \text{if} \quad k_{\eps,\bom,x}=k_x\\
\scal{\tfrac 8 7 L_{k^\pr_{\eps,\bom}} }^\nu  \e^{\frac 2 {33} m_\eps  L_{k^\pr_{\eps,\bom}}} 
 \e^{- \frac 7 {132} m_\eps \norm{y-x}} & \quad  \text{if} \quad k_{\eps,\bom,x}= k^\pr_{\eps,\bom},
\end{cases}   .\notag
\end{align}

%%%CHANGE
Combining \eq{WW2} and \eq{preSUDEC2},  noting $\norm{x}^{2d} >\e$ if $k_x\ge 2$, we conclude that for for all $ \lambda \in \R$ and  $x,y\in\Z^d$ we have 
\begin{align}
&W\up{x}_{\eps,\bom,\lambda}(x)W\up{x}_{\eps,\bom,\lambda}(y) \\ \notag  & \quad 
\le  C_{\eps,m_\eps,\bom,\nu}  \scal{  (2d\log \scal{x})^{\frac 1 \xi} }^\nu\e^{\frac 4 {33} m_\eps   (2d\log \scal{x})^{\frac 1 \xi}}  \e^{- \frac 7 {132} m_\eps  \norm{y-x}} 
 \\ \notag & \quad 
\le C_{\eps,m_\eps,\bom,\nu} \scal{ \tfrac 1{m_\eps}}^\nu\e^{(\frac 4 {33} +\nu)m_\eps  (2d\log \scal{x})^{\frac 1 \xi}}  \e^{- \frac 7 {132} m_\eps \norm{y-x}} \\ \notag  & \quad 
\le  C^\pr_{\eps,m_\eps,\bom,\nu} \e^{(\frac 4 {33} +\nu)m_\eps  (2d\log \scal{x})^{\frac 1 \xi}}  \e^{- \frac 7 {132} m_\eps \norm{y-x}} ,\end{align}
which is \eq{eqWW}.

Part (iv) follows  from (iii), since \eq{eqWW} implies
\begin{align}\label{eqWW24444}
 \abs{\psi(x)}\abs{\psi(y)}&
\le  C_{\eps,m_\eps,\bom,\nu} \, \norm{T_x^{-1} \psi}^2  \e^{(\frac 4 {33} +\nu)m_\eps  (2d\log \scal{x})^{\frac 1 \xi}}  \e^{- \frac 7 {132} m_\eps \norm{y-x}}\\ \notag
& \le C_{\eps,m_\eps,\bom,\nu} \, \norm{T_0^{-1} \psi}^2\la x\ra^{2\nu}  \e^{(\frac 4 {33} +\nu)m_\eps  (2d\log \scal{x})^{\frac 1 \xi}}  \e^{- \frac 7 {132} m_\eps \norm{y-x}},
 \end{align}
for all $x,y \in \Z^d$, which is \eq{eqWW2}.

 Part   (v) also follows from (iii).
Given $\lambda \in \R$, let $\psi \in  \Chi_{\set{\lambda}}(H_{\eps,\bom})\setminus \set{0}$. Clearly there
exists $x_\lambda=x_{\eps,\bom,\lambda} \in \Z^d$  (not unique) such that
\begin{equation}\label{psi87}
\abs{\psi(x_\lambda)} =
\max_{x \in \Z^d} \,\abs{\psi(x)}.
\end{equation}
Since for all $a \in \Z^d$ we have
\begin{equation}
\begin{split}
\lVert T_a^{-1}  {\psi}  \rVert^2& = 
\sum_{x \in \Z^d} \abs{\psi(x)}^2 \scal{x-a}^{-2\nu} 
\le  \abs{\psi(x_\lambda)}^2  \sum_{x \in \Z^d} \ \scal{x-a}^{-2\nu} \\
& =  \abs{\psi(x_\lambda)}^2  \sum_{x \in \Z^d} \ \scal{x}^{-2\nu} 
= C_{d,\nu}^2  \abs{\psi(x_\lambda)}^2  ,
\end{split}
\end{equation}
where  $ C_{d,\nu} =\pa{ \sum_{x \in \Z^d} \ \scal{x}^{-2\nu} }^{\frac 1 2} \in (1,\infty)$,
we get the discrete equivalent of   \cite[Eq.~(4.22)]{GKsudec},
\begin{equation}\label{Tmaxphi}
\lVert T_a^{-1} \psi \rVert \le
 C_{d,\nu} \abs{\psi(x_\lambda)} \quad \text{for all
 $a \in \Z^d$},
\end{equation}
and hence, recalling \eq{defGWx}, we have
\beq
W\up{x_\lambda}_{\eps,\bom,\lambda}(x_\lambda)\ge \frac {\abs{\psi(x_\lambda)}}  {\norm{T_{x_\lambda}^{-1} \psi}} \ge C_{d,\nu}>1.
\eeq
 Thus  \eq{eqWW}  implies that for all $y\in \Z^d$ we have
 \begin{align}
& C_{d,\nu} W\up{x_\lambda}_{\eps,\bom,\lambda}(y)
\le  C_{\eps,m_\eps,\bom,\nu}\e^{(\frac 4 {33} +\nu)m_\eps  (2d\log \scal{x_\lambda})^{\frac 1 \xi}}  \e^{- \frac 7 {132} m_\eps \norm{y-x_\lambda}},
\end{align}
 which yields \eq{SULE}.
\end{proof}

\section{Connection with  the Green's functions multiscale analysis}\label{secGreen}
Consider an Anderson model $H_{\eps,\bom}$ as in Definition~\ref{defineAnd}.   Given $\Th \subset \Z^d$ finite and 
$z \notin \sigma \left( H_{\Th}   \right)$, we set 
\beq  G_{\Th}(z)= (H_{\Th}  -z)^{-1}\mqtx{and} 
G_{\Th}(z;x,y) = \scal{ \delta_{x}, (H_{\Th}  -z)^{-1}\delta_{y}} \;\; \text{for} \;\;  x, \, y \in \Th.
\eeq
\begin{definition} Let $E\in \R$ and  $m > 0$. A box $\La_L$ is said to be  $(m, E)$-regular if  
$E \notin \sigma (H_{\La_L} )$  and
\beq\label{Gdecay}
 \abs{G_{\La_L}(E; x, y)} \leq e^{-m\norm{x -y}}
\;\; \text{for all} \;\;x,y \in\La_L \;\; \text{with} \; \norm{x -y} \geq \tfrac{L}{100} .
\eeq
\end{definition}
Given $x,y\in \R^d$,  a scale $L$, and $m>0$, we define the event
\beq
\cR_{L,m}(x,y)= \set{\text{for all}\;  E\in\R\sqtx{either} \La_{L} (x) \sqtx{or}\La_{L} (y)\sqtx{is}  (m, E)\text{-regular} }.
\eeq

The Green's function multiscale analysis  \cite{FS,FMSS,DK,GKboot,Kle} yields the following theorem.

\begin{theorem}\label{thmGMSA}
Given $0<\zeta<1$,  there exists $\eps_0>0$, a finite scale   $\cL $, and $m>0$, such that, given $L\ge \cL $,   for all  $ 0<\eps\le \eps_0$ we have
 \beq\label{concSingleMSA3}
 \inf_{x\in \R^d}  \P \set{\La_{L} (x)\sqtx{is}  (m, E)\text{-regular} }\ge 1 -  \e^{-L^{\zeta}}  \qtx{for all} E\in\R ,
   \eeq
and
 \begin{align}\label{concEGMSA}
\inf_{\substack{
x,y\in \R^d\\ \norm{x-y} > L
}} \P\set{\cR_{L,m}(x,y) } \ge 1 -  \e^{-L^\zeta}.
\end{align}
\end{theorem}

\eq{concSingleMSA3} are the conclusions of the single energy  multiscale analysis, and \eq{concEGMSA} are the conclusions of the  energy interval   multiscale analysis.

We will now show the connection between  Theorem~\ref{thmMSA}  and Theorem~\ref{thmGMSA}.  We assume $\xi,\zeta, \beta, \tau, \gamma$ satisfy \eq{ttauzeta0}.

We first show that the conclusions of Theorem~\ref{thmMSA} imply the conclusions of Theorem~\ref{thmGMSA}

\begin{proposition}  Let $\eps_0>0$.
 Fix $0<\eps\le \eps_0$ and suppose there exists $0<\xi<1$, a finite scale   $\cL $, and $m>0$, such that  the Anderson model  $H_{\eps,\bom}$ satisfies
 \begin{align}\label{concMSA3}
\inf_{x\in \R^d} \P\set{\La_{L} (x) \sqtx{is}  m  \text{-localizing for} \; H_{\eps,\bom}} \ge 1 -  \e^{-L^\xi} \mqtx{for all } L\ge \cL.
\end{align}
   Then, given $0<\zeta^\pr <\xi$ and $0<m^\pr<m$,  there exists a finite scale   $\cL_1=\cL_1(\cL, \eps_0,\xi,\zeta^\pr,m,m^\pr) $ such that  for all $L\ge \cL _1$  we have 
   \beq\label{concSingleMSA}
 \inf_{x\in \R^d}  \P \set{\La_{L} (x)\sqtx{is}  (m^\pr, E)\text{-regular} }\ge 1 -  \e^{-L^{\zeta^\pr}}  \qtx{for all} E\in\R, 
   \eeq
and
\begin{align}\label{concMSA4}
\inf_{\substack{
x,y\in \R^d\\ \norm{x-y} > L
}} \P\set{\cR_{L,m^\pr}(x,y) } \ge 1 -  \e^{-L^{\zeta^\pr}}.
\end{align}
\end{proposition}

\begin{proof}  Fix  $0<\eps\le \eps_0$,  $0<\zeta^\pr <\xi$, and $0<m^\pr<m$,  , and assume  \eq{concMSA3} for all  $L\ge \cL $.
Let  $L\ge \cL $ and 
suppose the box $\La_L$ is $m$-localizing    with   an $m$-localized eigensystem $\set{\vphi_x,\lambda_x}_{x \in  {\Lambda}_{L}}$. Let $x,u,v \in  {\Lambda}_{L}$ with $\norm{u-v}\ge \frac L {100}$. In this case either $ \norm{u-x}\ge L_\tau$ or $ \norm{v-x}\ge L_\tau$.  Say  $ \norm{u-x}\ge L_\tau$, then 
\beq
\abs{\vphi_x(u)\vphi_x(v)}\le \begin{cases} \e^{-m \pa{\norm{u-x} +\norm{v-x}}}\le  \e^{-m \norm{u-v}} & \text{if}\quad  \norm{v-x}\ge L_\tau\\
 \e^{-m\norm{u-x}}\le \e^{-m \pa{\norm{u-v}-L_\tau}}
 & \text{if}\quad  \norm{v-x}< L_\tau
\end{cases},
\eeq
so we conclude that
\beq
\abs{\vphi_x(u)\vphi_x(v)}\le  \e^{-m_1 \norm{u-v}}, \qtx{where} m_1 \ge m(1-CL^{\tau-1}).
\eeq

 Fix an energy $E\in \R$ and  assume  $ \norm{G_{\La_L}(E)}\le \e^{L^\beta}$.   Then  for $u,v \in  {\Lambda}_{L}$ with $\norm{u-v}\ge \frac L {100}$ we have
 \begin{align}
  \abs{G_{\La_L}(E; u, v)}& \leq \sum_{x\in \La_L} \abs{E-\lambda_x}^{-1} \abs{\vphi_x(u)\vphi_x(v)}\le \e^{L^\beta}\e^{-m_1\norm{u-v}}(L+1)^d \\   \notag & \le  \e^{-m_2 \norm{u-v}}, 
    \end{align}
  where
\beq
m_2\ge m_1\pa{1-C\pa{1+ \tfrac 1 {m_1}}L^{\beta-1}} \ge   m\pa{1-C\pa{1+ \tfrac 1 {m}}L^{\tau-1}}  \ge m^\pr\eeq    
  for large $L$.   We conclude that the
 box $\La_L$ is $( m^\pr,E)$-regular.

  Since the Wegner estimate (see, \cite[Corollary~5.25]{Ki}, \cite[Section~2]{CGK2}) gives
  \beq
  \P\set{ \norm{G_{\La_L}(E)}\le \e^{L^\beta}}\ge 1-  Q_\mu(2\e^{-L^\beta})\abs{\La_L}\ge 1 - \wtilde{K}2^\alpha \e^{-\alpha L^\beta}(L+1)^d\ge 1 - \tfrac 12  \e^{-L^{\zeta^\pr}}
  \eeq
  for large $L$, combining with \eq{concMSA3} we get \eq{concSingleMSA}.
  
 Now let $L\ge \cL$ and   
 consider two boxes $\La_L(x_1)$ and  $\La_L(x_2)$, where $x_1,x_2 \in \R^d$, $\norm{x_1-x_2} >L$.  Define the events
\begin{align}
\cA&= \set{\La(x_1)\sqtx{and}\La(x_2)\sqtx{are both} m\text{-localizing}},\\ \notag
\cB&= \set{\dist(\sigma(H_{\La_L(x_1)}), \sigma(H_{\La_L(x_2)}))\ge 2\e^{-L^{\beta}}}
\end{align}
It follows from
 \eq{concMSA3} that 
\beq
\P\set{\cA}\ge 1- 2\e^{-L^\xi}\ge 1 - \tfrac 12  \e^{-L^{\zeta^\pr}}.
\eeq
Since  $\norm{x_1-x_2} >L$,    the boxes are disjoint, and  the  Wegner estimate between  boxes (see \cite[Corollary~5.28]{Ki})
gives
\begin{align} 
&\P\set{\cB}\ge 1 -Q_\mu(4\e^{-L^\beta})\abs{\La_L(x_1)}\abs{{\La_L(x_2)}}  \ge 1 -\wtilde{K}4^\alpha \e^{-\alpha L^\beta}(L+1)^{2d}\ge 1 - \tfrac 12  \e^{-L^{\zeta^\pr}}.
\end{align}
Thus we have
\beq
\P\set{\cA\cap\cB}\ge 1 -  \e^{-L^{\zeta^\pr}}.
\eeq
Moreover, for $\bom \in \cA\cap\cB$ and $E\in \R$,  the boxes $\La(x_1)$ and $\La(x_2)$ are 
both $m$-localizing, and
we must have either $ \norm{G_{\La_L(x_1)}(E)}\le \e^{L^\beta}$ or  $ \norm{G_{\La_L(x_2)}(E)}\le \e^{L^\beta}$, so the previous argument shows that either  $\La(x_1)$ or  $\La(x_2)$ is $(m^\pr,E)$-regular for large $L$.   We proved \eq{concMSA4}.
\end{proof}

Conversely,  the conclusions of Theorem~\ref{thmGMSA} \emph{almost} imply the conclusions of Theorem~\ref{thmMSA}. To get  Theorem~\ref{thmMSA} we have to use Hall's Marriage Theorem for the labeling of eigenpairs, as in the proof of Proposition~\ref{propMSA}.

\begin{proposition}  Let $\eps_0>0$.
 Fix $0<\eps<\eps_0$ and suppose there exists $0<\zeta<1$, a finite scale   $\cL $, and $m>0$, such that  the Anderson model  $H_{\eps,\bom}$ satisfies  \eq{concEGMSA} for all  $L\ge \cL $.  Then, given $0<\xi<\zeta$ and $0<m^\pr<m$,  there exists a finite scale   $\cL_1=\cL_1(\cL,\eps_0,\zeta,\xi,m,m^\pr) $ such that  for all $L\ge \cL _1$  we have 
\begin{align}\label{concMSA2}
\inf_{x\in \R^d} \P\set{\La_{L} (x) \sqtx{is}  m^\pr  \text{-localizing for} \; H_{\eps,\bom}} \ge 1 -  \e^{-L^{\xi}} .
\end{align}
\end{proposition}

\begin{proof}  Fix  $0<\eps\le \eps_0$,  $0<\zeta^\pr <\xi$, and $0<m^\pr<m$, and assume  \eq{concEGMSA} for all  $L\ge \cL $.
Let $L=\ell^\gamma$ with $\ell \ge \cL$. We take
$\La_L=\La(x_0)$, where $x_0\in \R^d$,  and let
${\mathcal C}_{L,\ell}={\mathcal C}_{L,\ell} \left(x_0 \right)$ be the suitable $\ell$-cover of $\La_L$, with $ \Xi_{L,\ell}=\Xi_{L,\ell} \left(x_0 \right)$.  We define the event
\beq
\cR_{L,m}(x_0)= \bigcap_{\substack{
a,b\in\Xi_{L,\ell}\\ \norm{a-b} > \ell}}\cR_{\ell,m}(a,b),
\eeq
so we have
\beq\label{probR}
\P\set{\cR_{L,m}(x_0)}\ge 1 - \pa{\tfrac{2L}\ell }^{2d}\e^{-\ell^\zeta}\ge 1 -\tfrac 1 2 \e^{-L^\xi}
\eeq
for sufficiently large $L$.

Fix $\bom \in \cR_{L,m}(x_0)$, and let  $H_{\Theta}=H_{\eps,\bom,\Theta}$ for $\Theta \subset \Z^d$.  Suppose   $(\vphi,\lambda)$ is an eigenpair for  $ \sigma(H_{\La_L})$. Recall that for any box  $\La_\ell \subset \La_L$ it follows from $\pa{H_{\La_L}-\lambda}\vphi =0$ and \eq{Hdecomp}  that
\beq
 \Chi_{{\Lambda}_{\ell}} \pa{H_{{\Lambda}_{\ell}}-\lambda}\vphi= - \eps \Chi_{{\Lambda}_{\ell}}\Gamma_{ \boldsymbol{ \partial}^{ \La_L}  {\Lambda}_{\ell}}\vphi,
 \eeq
 so,  for all  $x\in \Lambda_{\ell}$  we have
\beq\label{vphiRy}
\vphi(x)= -   \pa{ \eps G_{{\Lambda}_{\ell}}(\lambda) \Gamma_{ \boldsymbol{ \partial}^{ \La_L}  {\Lambda}_{\ell}}\vphi}(x)=  \sum_{(u,v)\in \boldsymbol{ \partial}^{ \La_L}  {\Lambda}_{\ell} } \eps G_{{\Lambda}_{\ell}}(\lambda; x,u) \vphi(v).
\eeq

Since $\bom \in \cR_{L,m}(x_0)$, there exists $a_\lambda\in \Xi_{L,\ell}$ such that
$\La_\ell(b)$ is $(m, \lambda)$-regular for all $b\in \Xi_{L,\ell}$ with $\norm{b-a_\lambda} >\ell$. 
 In particular, if $y \notin \La_{(2\rho +1)\ell} (a_\lambda)$ we have that  ${\Lambda}_{\ell}^{(y)}$ (as in \eq{covproperty}) is $(m, \lambda)$-regular.  Thus it follows from \eq{vphiRy} and \eq{Gdecay} that
\beq\label{justsmall1}
\abs{\vphi(y)}\le \eps s_d \ell^{d-1}\e^{-m\pa{\norm{y-y_1}-1}}\abs{\vphi(y_1)}\le \e^{-m_1\norm{y-y_1}}\abs{\vphi(y_1)}
\eeq
for some $y_1\in  \partial_{\mathrm{ex}}^{ \La_L}{\Lambda}_{\ell}^{(y)}$,
so $\norm{y-y_1} \ge \frac \ell{10}$, where\beq
m_1 \ge m \pa{1- \tfrac {10}\ell -C_{d} \tfrac {\log \ell}{m\ell}- \tfrac{\log \eps_0}{m\ell}}\ge    m \pa{1-C^\pr_{d,\eps_0} (1 +\tfrac 1 m) \tfrac {\log \ell}{\ell}}.
\eeq
Since we can repeat the procedure if  $y_1 \notin \La_{(2\rho +1)\ell} (a_\lambda)$, we conclude that
\beq\label{justsmall}
\abs{\vphi(y)}\le   \e^{-m_1\pa{\norm{y-a_\lambda}-  \pa{\rho+\frac 1 2}\ell}}\le   \e^{-m_1\pa{\norm{y-a_\lambda}-  \frac 3 2\ell}}.
\eeq
Since $\frac 1 \gamma < \tau <1$, it follows that 
\beq\label{justsmall3}
\abs{\vphi(y)}\le   \e^{-m_2\norm{y-a_\lambda}}\qtx{if}\norm{y-a_\lambda} \ge \tfrac 1 2 L_\tau,
\eeq
where
\beq
m_2\ge m_1 \pa{1- \tfrac {3}{ \ell^{(\gamma-1)\tau}} }\ge m  \pa{1- C^{\pr\pr}_{d,\eps_0} (1 +\tfrac 1 m)\tfrac {1}{ \ell^{(\gamma-1)\tau}} }.
\eeq

Let $\set{(\vphi_j,\lambda_j)}_{j=1}^{\abs{\La_L}}$  be an eigensystem for $H_{\La_L}$. We let $a_j=a_{\lambda_j}$. (Note that the map  $j \mapsto a_j$ is not necessarily an injection.) We claim
\beq\label{U2ell}
\La_L =\bigcup_{j=1}^{\abs{\La_L}}\La_{(2\rho +1)\ell} (a_j).
\eeq
This can be seen as follows. Suppose $y\in \La_L \setminus \bigcup_{j=1}^{\abs{\La_L}}\La_{(2\rho +1)\ell} (a_j)$. It follows from \eq{justsmall1} that
\begin{align}
1= \norm{\delta_y}^2=\sum_{j=1}^{\abs{\La_L}} \abs{\vphi_j(y)}^2 \le (L+1)^d\e^{-m_1\frac \ell{5}}<1,
\end{align}
a contradiction.

Given $x\in \La_L$, we set
\beq
\cN(x)= \set{j\in \set{1,2,\ldots,\abs{\La_L}}; x \in \La_{(2\rho +1)\ell} (a_j)}.
\eeq
Note $\cN(x)\ne \emptyset$ in view of \eq{U2ell}.

Let $\cN(\Th)=\bigcup_{x\in \Th} \cN(x)$ for $\Theta \subset \La_L$.   We fix   $\Theta \subset \La_L$, and  let $Q_\Th$ be the orthogonal projection onto the span of $\set{\vphi_j; \; j \in \cN(\Th)}$.
For every $j \notin \cN(\Th)$ we have \eq{justsmall1} for $\vphi_j(x)$ for all $x \in\Th$, so
\beq
\norm{(1-Q_\Th)\Chi_\Th}\le \abs{\La_L}^{\frac 12 } \abs{\Th}^{\frac 12 }\e^{-m_1\frac \ell{10}}\le (L+1)^d \e^{-m^\pr\frac \ell{10}} <1,
\eeq
so it follows from Lemma~\ref{lemPQ} that
\beq\label{Hallc012}
\abs{\Th}= \tr \Chi_\Th \le  \tr Q_\Th= \abs{\cN(\Th)},
\eeq
which is Hall's condition \eq{eq:Hallcond1}. 
Thus we can apply Hall's Marriage Theorem, concluding that there exists a bijection
\beq
x\in \La_L \mapsto j_x \in \set{1,2,\ldots,\abs{\La_L}}, \qtx{where}  j_x \in \cN(x).
\eeq 
We set $(\vphi_x,\lambda_x)=(\vphi_{j_x},\lambda_{j_x})$ for all $x\in \La_L $. If $\norm{y-x}\ge L_\tau$
we have 
\beq
\norm{y-a_{j_x}}\ge \norm{y-x} -\norm{x-a_{j_x}}  \ge L_\tau - (\rho +\tfrac 12)\ell > \tfrac 12 L_\tau,
\eeq
so it follows from \eq{justsmall3} that
\begin{align}\label{justsmall4}
\abs{\vphi_x(y)}&\le   \e^{-m_2\norm{y-a_{j_x}}}\le   \e^{-m_2\pa{ \norm{y-x} -\norm{x-a_{j_x}}}}
\le   \e^{-m_2\pa{ \norm{y-x} -(\rho +\tfrac 12)\ell}}\\ \notag  &
\le  \e^{-m_3\norm{y-x}},
\end{align}
where 
\beq
m_3 \ge m_2\pa{1- \tfrac {3}{2 \ell^{(\gamma-1)\tau}} }\ge  m  \pa{1- C^{\pr\pr\pr}_{d,\eps_0}(1 +\tfrac 1 m)\tfrac {1}{ \ell^{(\gamma-1)\tau}} }\ge m^\pr,
\eeq
for $\ell$ sufficiently large.

Thus for  $\bom \in \cR_{L,m}(x_0)$ the box $\La_L(x_0)$ would be $m^\pr$-localizing for $H$ if it would be level spacing. Since it follows from \eq{probsep} that
 this is true for $ L$ large with probability
$\ge 1 -\tfrac 1 2 \e^{-L^\xi}$, we have \eq{concMSA2}.
\end{proof}

\appendix

\section{Lemma about orthogonal projections}

\begin{lemma} \label{lemPQ} Let $P$ and $Q$ be  orthogonal projections  on a  Hilbert space $\H$.
Then
\beq
\norm{(1-P)Q} <1 \quad \Longrightarrow \quad \tr Q \le  \tr  P.
\eeq
In particular, taking $Q=1$ we get  
\beq
\norm{1-P} <1 \quad \Longrightarrow \quad P=1.
\eeq

\end{lemma}

\begin{proof}
Since $(1-(1-P)Q)Q=PQ$ and $1-(1-P)Q$ is invertible by the assumption of the lemma, we infer that
\beq Q=(1-(1-P)Q)^{-1}PQ \Longrightarrow \tr Q\le \tr P,\eeq
where in the last step we have used $A=BCD\Longrightarrow \Rank A\le \Rank C$.
\end{proof}

\section{Estimating the probability of level spacing}\label{appspspaced}
In this appendix we review an estimate of the probability of level spacing due to Klein and Molchanov \cite{KlM}.
Let us consider a generalized Anderson model
\beq \label{defAndgen}
H_{\bom} :=  H_0 + V_{\bom} \quad \text{on} \quad  \ell^2(\Z^d), 
\eeq 
where $H_0$ is a bounded self-adjoint operator on  $\ell^2(\Z^d)$ and  $V_{\bom}$ is a random potential:      $V_{\bom}(x)= \omega_x$ for  $ x \in \Z^d$, where
$\bom=\{ \omega_x \}_{x\in
\Z^d}$ is a family of independent 
identically distributed random
variables,  whose  common probability 
distribution $\mu$ is non-degenerate and   H\"older continuous of order  $\alpha \in ( \frac 12,1]$:  
\beq
S_\mu(t) \le K t^\alpha \qtx{for all} t \in [0,1],
\eeq
where $K$ is  a constant and  $S_\mu(t):= \sup_{a\in \R} \mu \set{[a, a+t]} $ is the concentration function of the measure $\mu$.  We set
$Q_\mu(t)= S_\mu(t)$ and $\wtilde{K}=K$ if $\alpha=1$ and $Q_\mu(t)= 8S_\mu(t)$ and  $\wtilde{K}=8K$ if  $\alpha \in ( \frac 12,1)$.

\begin{lemma}[{\cite[Lemma~2]{KlM}}]\label{lemKlM}
 Let $\Th \subset \Z^d$ be a finite subset, $I\subset \R$ be a bounded interval, and $\eta\in (0,\frac 1 2]$.  Let
$\mathcal{E}_{\Th,I,\eta}$ denote the event   that 
 $\tr \Chi_{J}(H_{\bom,\Th}) \le 1$ for all subintervals $J \subset I $
with length $|J| \le  \eta $.  Then
\begin{equation}\label{speig}
\P\set{\mathcal{E}_{\Th,I,\eta}}\ge  1 -  \wtilde{K}^2 (|I| +1) \pa{2\eta}^{2\alpha-1}\abs{\Th}^2.  
\end{equation}
\end{lemma}

\begin{proof} We recall  the proof for completeness.  The proof is based on Minami's inequality \cite{M}, which we use in the form given in  \cite[Theorem~3.3]{CGK1} and \cite[Theorem~2.1]{CGK2}:
\beq\label{minami}
\P\set{\tr \Chi_{J}(H_{\bom,\Th}) \ge 2}\le \tfrac 1 2\E\set{\tr \Chi_{J}(H_{\bom,\Th})\pa{\tr \Chi_{J}(H_{\bom,\Th})-1}}\le \tfrac 1 2\pa{Q_\mu(\abs{J})\abs{\Th}}^2.
\eeq

We   cover the interval $I$ by $2 \left \lceil \frac{\abs{I}}{2\eta}\right  \rceil \le 
 \frac{\abs{I}}{\eta}  +2 $
intervals of length $2\eta  $, in such a way that 
any subinterval $J \subset I $
with length $|J| \le \eta $ will be contained in one of these intervals. Then
\begin{align}\notag
\P\set{\mathcal{E}_{\Th,I,\eta}^c}&\le \P\set{\text{there exists an interval}\ J\subset I, \abs{J}\le \eta, \sqtx{with} \tr \Chi_{J}(H_{\bom,\Th}) \ge 2 }\\  & 
\; \le \pa{ \tfrac{\abs{I}}{\eta}  +2}\tfrac 1 2\pa{Q_\mu(2\eta)\abs{\Th}}^2
\le \wtilde{K}^2 (|I| +1) \pa{2\eta}^{2\alpha-1}\abs{\Th}^2, 
\end{align}
where we used Minami's inequality \eq{minami}.
\end{proof}

\section{Hall's Marriage Theorem}  \label{appHall}

Hall's Marriage Theorem (see \cite[Chapter~2]{BDM}) gives a necessary and sufficient condition for the existence of a perfect matching in a bipartite graph. 
Let $\G = (A,B;\E)$ be a bipartite graph with vertex sets $A$ and $B$ and edge set $\E\subset  A\times B$ (the bipartite condition).  $\M\subset\E$ is called a {\em matching} if every vertex of $\G$ coincides with at most one edge from $\M$;  
it is a {\em perfect matching} if every vertex of $\G$ coincides with exactly one edge from $M$, i.e., every vertex in $A$ is matched with a unique vertex in $B$ and vice-versa. In particular,  $|A | = |B |$ is a necessary condition for the 
 the existence of a perfect matching. 
Given a vertex $a \in A$, let $\cN(a)= \set{b\in B; \ (a,b) \in \E}$, the  set of  neighbors of $a$.  Let  $\cN(U)=\cup_{u\in U} \cN(u)$ for   $U\subset A$. 

\begin{hall} Let $\G = (A,B;\E)$ be a bipartite graph with $|A| = |B |$. There exists a perfect matching  in $\G$ if and only if   the graph $\G$ fulfills Hall's condition
\beq\label{eq:Hallcond1} |U| \le|\cN(U)| \qtx{for all} U\subset A.
\eeq
\end{hall}


\begin{thebibliography}{FrWBSE}



\bibitem[AS]{AS} Abdul-Rahman, H.,  Stolz , G : A uniform area law for the entanglement of eigenstates in the disordered XY chain.  Preprint,  \href{http://arxiv.org/abs/1505.02117}{arXiv:1505.02117 [math-ph]}




\bibitem[Ai]{A}  Aizenman, M.: {Localization at weak disorder: some
elementary bounds}. Rev. Math. Phys. {\bf 6}, 1163-1182 (1994)


\bibitem[AiSFH]{ASFH}  Aizenman, M.,  Schenker, J., Friedrich, R., 
Hundertmark, D.: Finite volume fractional-moment criteria for 
Anderson localization.
Commun. Math. Phys. \textbf{224}, 219-253 (2001) 

\bibitem[AiENSS]{AENSS}  Aizenman, M., Elgart, A., Naboko, S.,   
 Schenker, J.,  Stolz, G.: Moment analysis for localization in random Schr\"odinger 
operators. Inv. Math. \textbf{163}, 343-413 (2006)

\bibitem[AiM]{AM}  Aizenman, M.,  Molchanov, S.:  {Localization at large
disorder and extreme energies:  an elementary derivation}.  Commun. Math.
Phys. {\bf 157}, 245-278 (1993)

\bibitem[AiW]{AW}  Aizenman, M.,  Warzel, S.: \emph{Random operators.
Disorder effects on quantum spectra and dynamics}.
Graduate Studies in Mathematics \textbf{168}. American Mathematical Society, Providence, RI,  2015.

\bibitem[AlGKL]{AlGKL} Altshuler, B.L., Gefen, Y. , Kamenev, A., Levitov, L.S: Quasiparticle Lifetime in a Finite System: A Nonperturbative Approach. Phys. Rev. Lett. \textbf{78}, 2803 (1997)


\bibitem[An]{And}  Anderson, P.:  Absence of diffusion in certain random
lattices.  Phys. Rev. {\bf 109}, 1492-1505 (1958)

\bibitem[BAA]{BAA} Basko, D.M.,  Aleiner,  I.L.,  Altshuler, B.L.:  Metal-insulator transition in a weakly interacting many-electron system with localized single-particle states.  Ann. Phys. \textbf{321}, 1126-1205 (2006).


\bibitem[BoK]{BK} Bourgain, J., Kenig, C.: On localization in the continuous Anderson-Bernoulli model in higher dimension.  Invent. Math. \textbf{161}, 389-426 (2005)

\bibitem[BuDM]{BDM}Burkard, R., Dell'Amico, M.,  Martello, S.:  \emph{Assignment problems}.
Society for Industrial and Applied Mathematics (SIAM), Philadelphia,
 PA,  2009.
 
 
\bibitem[BurO]{BurO} Burrell, C., Osborne, T.: Bounds on the speed of information propagation in disordered quantum spin chains. Phys. Rev. Lett. \textbf{99}  167201 (2007)
 
\bibitem[CGK1]{CGK1} Combes,  J.M., Germinet, F.,  Klein, A.: Poisson statistics for eigenvalues of continuum random Schr\"odinger operators,  Analysis and PDE~{\bf 3}, 49-80  (2010). \doi{10.2140/apde.2010.3.49}

 \bibitem[CGK2]{CGK2}  Combes,  J.M., Germinet, F.,  Klein, A.: Generalized eigenvalue-counting estimates for the Anderson model.  J. Stat. Phys. \textbf{135},   201-216 (2009). \doi{10.1007/s10955-009-9731-3}
 
 \bibitem[CoH]{CH}  Combes,  J.M.,   Hislop, P.D.: {Localization for some
continuous, random Hamiltonians in d-dimension}. J. Funct. Anal. \textbf{124},
149-180 (1994)
\bibitem[DJLS1]{DRJLS0}  Del Rio, R.,  Jitomirskaya, S.,  Last, Y., Simon, B.:{ What is Localization?} Phys. Rev. Lett. 75, 117-119 (1995)

\bibitem[DJLS2]{DRJLS} Del Rio, R.,  Jitomirskaya, S.,  Last, Y., Simon, B.: {Operators with singular continuous spectrum IV: Hausdorff dimensions, rank one
perturbations and localization}. J. d'Analyse Math. {\bf 69}, 153-200 (1996)

 
 \bibitem[Dr]{Dr}  von Dreifus, H.: {\em On the effects of randomness in
ferromagnetic models and Schr\"odinger operators}.  Ph.D. thesis, New York
University (1987)
 
 \bibitem[DrK]{DK} von Dreifus, H.,  Klein, A.:  {A new proof of localization in
the Anderson tight binding model}.  Commun. Math. Phys. \textbf{124},
285-299  (1989) 

    \bibitem[EFG]{EFG} Eisert, J. ,  Gogolin, C.: Equilibration, thermalisation, and the emergence of statistical mechanics in closed quantum systems. 	Preprint,  \href{http://arxiv.org/abs/1503.07538}{arXiv:1503.07538 [quant-ph]}
    
  \bibitem[ElK]{EK} Elgart, A., Klein, A.: Eigensystem multiscale analysis for Anderson localization in energy intervals. In preparation

\bibitem[FK1]{FK1} Figotin, A.,  Klein, A.: Localization phenomenon in gaps of the spectrum of random lattice operators. J. Statist. Phys. \textbf{75}, 997-1021 (1994)


\bibitem[FK2]{FK} Figotin, A.,  Klein, A.: {Localization of classical waves I:
Acoustic waves}.  Commun. Math. Phys. {\bf 180}, 439-482 (1996)

\bibitem[FlA]{FlA}Fleishman, L., Anderson, P. W. : Interactions and the Anderson transition. Phys. Rev. B \textbf{21}, 2366 (1980)

\bibitem[FrWBSE]{FrWBSE} Friesdorf, M. ,  Werner, A. H., Brown, W.,  Scholz, V. B.,   Eisert, J.:
Many-body localization implies that eigenvectors are matrix-product states.
Phys. Rev. Lett. \textbf{114}, 170505 (2015)

\bibitem[FroS]{FS}  Fr\"ohlich, J.,  Spencer, T.: {Absence of diffusion with
Anderson tight binding model for large disorder or low energy}. Commun.
Math. Phys. {\bf 88}, 151-184 (1983)

\bibitem[FroMSS]{FMSS} Fr\"ohlich, J.:    Martinelli, F.,  Scoppola, E., 
Spencer, T.:  {Constructive proof of localization in the Anderson tight
binding model}. Commun. Math. Phys. {\bf 101}, 21-46 (1985)

\bibitem[GK1]{GKboot} Germinet, F.,  Klein, A.: {Bootstrap multiscale analysis and
localization in random media}.  Commun. Math. Phys. \textbf{222},
415-448 (2001). \doi{10.1007/s002200100518}

\bibitem[GK2]{GKfinvol} Germinet, F,  Klein, A.:
Explicit finite volume criteria for localization in continuous
 random media and applications. Geom. Funct. Anal. \textbf{13}, 1201-1238 (2003)

\bibitem[GK3]{GKsudec} Germinet, F., Klein, A.: New characterizations of the region of complete localization for random Schr\"odinger operators.  J. Stat. Phys. \textbf{122}, 73-94 (2006)

\bibitem[GK4]{GKber} Germinet, F.,  Klein, A.:  A comprehensive proof of localization for continuous Anderson models with singular random potentials.  J. Eur. Math. Soc. \textbf{15}, 53-143 (2013). \doi{10.4171/JEMS/356}

\bibitem[GoMP]{GoMP}Gornyi, I. V. ,  Mirlin, A. D. , Polyakov, D. G. : Interacting electrons in  disordered  wires:  Anderson localization and low-T transport. Phys. Rev. Lett. \textbf{95}, 206603 (2005)

\bibitem[HSS]{HSS}
Hamza, E.,  Sims, R.,  Stolz, G.: Dynamical localization in disordered quantum spin systems. Comm. Math. Phys.  \textbf{315},
 215Ð239 (2012)
 
 \bibitem[HoM]{HM}   Holden, H., Martinelli, F.:  On absence of diffusion near
the bottom of the spectrum for a random Schr\"odinger operator.
 Commun. Math. Phys. {\bf 93}, 197-217 (1984)


\bibitem[HorJ]{HJ}
 Horn, R.A., Johnson, C.R.: \emph{Matrix Analysis}. Cambridge
University Press, 1985.


\bibitem[I1]{Im1}  Imbrie, J.: On many-body localization for quantum spin chains. Preprint, \href{http://arxiv.org/abs/1403.7837}{arXiv:1403.7837 [math-ph]}

\bibitem[I2]{Im2}  Imbrie, J.:  Multi-scale Jacobi method for Anderson localization.  Preprint, \href{http://arxiv.org/abs/1406.2957}{arXiv:1406.2957 [math-ph]}

\bibitem[K]{Ki} Kirsch, W.: {An invitation to random Schr\"odinger operators. In Random Schr\"odinger-Operators. Panoramas et Syntheses \textbf{25}}. 
Societe Mathematique de France, Paris, 1-119 (2008)

\bibitem[KSS]{KSS} Kirsch, W.,   Stollmann,  P, Stolz, G.:
{Localization for random perturbations of periodic Schr\"odinger
 operators}.  Random Oper. Stochastic Equations {\bf 6},
 241-268 (1998) 

\bibitem[Kl]{Kle} Klein, A.:   Multiscale analysis and localization of random operators.
 In \emph{Random Schr\" odinger Operators}.  Panoramas et Synth\`{e}ses \textbf{25},   121-159,  
 Soci\'{e}t\'{e}  Math\'{e}matique de France, Paris 2008

\bibitem[KlM]{KlM} Klein. A., Molchanov, S.: Simplicity of eigenvalues in the Anderson
model.  J. Stat. Phys.~{\bf 122}  , 95-99 (2006)

\bibitem[KlT]{KlT} Klein. A., Tsang, C.S.S.: Eigensystem bootstrap multiscale analysis for the Anderson model. Preprint. \href{http://arxiv.org/abs/1605.03637}{arXiv:1605.03637} 

\bibitem[M]{M} Minami, N.: Local fluctuation of the spectrum of a multidimensional
Anderson tight binding model. Commun. Math. Phys.~{\bf 177}, 709-725 (1996)

\bibitem[NH]{NH}  Nandkishore, R.,  Huse, D.A.: Many-body localization and thermalization in quantum statistical mechanics. Annu. Rev. Condens. Matter Phys. \textbf{6}, 15-38 (2015)

\bibitem[OH]{OH} Oganesyan, V. ,    Huse, D.A.:  Localization of interacting fermions at high temperature. Phys. Rev. B
\textbf{75}, 155111 (2007)


\bibitem[PH]{PH} Pal, A. ,  Huse, D.A.: The many-body localization phase transition. Phys. Rev. B \textbf{82}, 174411 (2010)

\bibitem[PaS]{PS}Pastur,L.,  Slavin, V.: On the Area Law for Disordered Free Fermions. Phys. Rev. Lett. \textbf{113}, 150404  (2014)

\bibitem[S]{Sp}  Spencer,  T.:  Localization for random and
quasiperiodic potentials.    J. Stat. Phys. {\bf 51}, 1009-1019 (1988)

\end{thebibliography}
\end{document}